\documentclass{article}
\usepackage{fullpage}
\usepackage{amsmath}
\usepackage{amssymb}
\usepackage{amsthm}
\usepackage{thmtools}
\usepackage{enumitem}
\usepackage[oldenum,olditem]{paralist}
\usepackage{graphicx}
\usepackage{outlines}
\usepackage{booktabs}
\usepackage{csquotes}
\usepackage[labelfont={bf}]{caption}
\usepackage[hypcap=true]{subcaption}
\usepackage[backend=biber,style=alphabetic,natbib=true,sortcites=true,sorting=ynt,maxcitenames=3,mincitenames=1,maxbibnames=99,maxalphanames=4,minalphanames=3]{biblatex}
\usepackage[normalem,normalbf]{ulem}
\usepackage{makecell}
\usepackage{multirow}
\usepackage{colortbl}
\usepackage{algorithm}
\usepackage{algorithmicx}
\usepackage[noend]{algpseudocode}
\usepackage{tikz}
\usetikzlibrary{arrows.meta,calc,decorations.markings,decorations.pathmorphing,shapes.geometric}
\usepackage{hyperref}
\hypersetup{
  colorlinks=true,   %
  allcolors=red!50!blue,  %
}
\usepackage{cleveref}

\newcommand{\email}[1]{\href{mailto:#1}{\nolinkurl{#1}}}

\makeatletter
\providecommand*{\theHALG@line}{}
\renewcommand*{\theHALG@line}{\thealgorithm.\arabic{ALG@line}}
\AddToHook{env/algorithmic/begin}{\def\@currentcounter{ALG@line}}
\crefalias{ALG@line}{line}
\AddToHook{cmd/appendix/before}{%
  \crefalias{section}{appendix}%
  \crefalias{subsection}{appendix}%
  \crefalias{subsubsection}{appendix}}
\IfFormatAtLeastTF{2026-06-01}{%
  \renewcommand\thmt@autorefsetup{%
    \@xa\def\csname\thmt@envname autorefname\@xa\endcsname\@xa{\thmt@thmname}}%
}{}
\makeatother

\newcommand{\xx}{\texttt{x}}
\algdef{SE}[UPON]{Upon}{EndUpon}[1]{\textbf{upon}\ #1\ \algorithmicdo}{\algorithmicend\ \textbf{upon}}%
\algtext*{EndUpon}%

\newtheorem{theorem}{Theorem}[section]

\newtheorem{corollary}[theorem]{Corollary}
\newtheorem{proposition}[theorem]{Proposition}
\newtheorem{lemma}[theorem]{Lemma}

\newtheorem{observation}[theorem]{Observation}
\crefname{observation}{observation}{observations}
\Crefname{observation}{Observation}{Observations}

\newtheorem{question}[theorem]{Question}
\newtheoremstyle{subclaim}
  {0.5\topsep}
  {0.5\topsep}
  {\normalfont}
  {}
  {\itshape}
  {.}
  {0.5em}
  {\thmname{#1}\thmnumber{ #2}\thmnote{ (#3)}}
\theoremstyle{subclaim}
\newtheorem{claim}[theorem]{Claim}
\newtheorem*{claim*}{Claim}
\crefname{claim}{claim}{claims}
\Crefname{claim}{Claim}{Claims}
\theoremstyle{plain}
\theoremstyle{definition}
\newtheorem{definition}[theorem]{Definition}
\crefname{definition}{definition}{definitions}
\Crefname{definition}{Definition}{Definitions}

\crefname{example}{example}{examples}
\Crefname{example}{Example}{Examples}
\theoremstyle{remark}
\newtheorem{remark}[theorem]{Remark}
\theoremstyle{plain}

\NewDocumentCommand{\proofqed}{s o m}{%
  \renewcommand{\qedsymbol}{%
    \ensuremath{\text{\openbox}_{\text{\normalfont
      \IfBooleanTF{#1}{#3}{%
        \IfNoValueTF{#2}{\Cref{#3}}{%
          \nameCref{#3}~\hyperref[#2]{\ref*{#3}}}}}}}}}

\newcounter{proofinstance}
\newcounter{proofpart}

\makeatletter
\newcommand*{\proof@instance}{0}
\renewcommand*{\theHproofpart}{\proof@instance.\arabic{proofpart}}

\renewcommand\paragraph{\@startsection{paragraph}{4}{\z@}%
  {3.25ex \@plus1ex \@minus.2ex}%
  {-1em}%
  {\normalfont\normalsize\bfseries\itshape}}

\AddToHook{env/proof/begin}{%
  \edef\proof@savedpart{\arabic{proofpart}}%
  \stepcounter{proofinstance}%
  \edef\proof@instance{\arabic{proofinstance}}%
  \setcounter{proofpart}{0}%
  \renewcommand{\qedsymbol}{\openbox}%
}
\AddToHook{env/proof/end}{%
  \setcounter{proofpart}{\proof@savedpart}%
}
\makeatother

\crefname{proofpart}{part}{parts}
\Crefname{proofpart}{Part}{Parts}

\newcommand{\proofpart}[1]{%
  \par
  \refstepcounter{proofpart}%
  \paragraph*{Part~\theproofpart. #1.}%
}

\newcommand{\deq}{\mathrel{\stackrel{\scriptscriptstyle\Delta}{=}}}

\newcommand{\R}{\mathbb{R}}

\newcommand{\ts}{\textsf{TS}}

\newcommand{\nextOpParens}{\texttt{next()}}
\newcommand{\alg}{\mathcal{A}}
\NewDocumentCommand{\interval}{om}{\IfNoValueTF{#1}{\mathcal{I}(#2)}{\mathcal{I}_{#1}(#2)}}
\newcommand{\receiverSet}[2]{D_{#1}^{#2}}
\NewDocumentCommand{\latency}{om}{\IfNoValueTF{#1}{\delta(#2)}{\delta_{#1}(#2)}}

\newcommand{\byzprocess}[1]{{}^{\mathrm{Byz}}#1}

\newcommand{\Deq}{\;\deq\;}

\DeclareMathOperator{\Aff}{Affine}
\DeclareRobustCommand{\circcase}[1]{%
  \tikz[baseline=(char.base)]{
    \node[circle, draw=black, fill=white, inner sep=0pt, minimum size=10.5pt, font=\scriptsize] (char) {#1};
  }%
}
\newcommand{\inlineoplength}{4.5}
\DeclareRobustCommand{\opinvocationmark}{%
  \tikz[baseline=-0.55ex, x=1ex, y=1ex]{%
    \draw[line width=0.8pt, line cap=round] (0,-0.45) -- (0,0.45) (0,0) -- (\inlineoplength,0);
  }%
}
\DeclareRobustCommand{\opreturnmark}{%
  \tikz[baseline=-0.55ex, x=1ex, y=1ex]{%
    \draw[line width=0.8pt, line cap=round] (0,0) -- (\inlineoplength,0) (\inlineoplength,-0.45) -- (\inlineoplength,0.45);
  }%
}
\DeclareRobustCommand{\solidmessagearrow}{%
  \tikz[baseline=-0.55ex, x=1ex, y=1ex]{%
    \draw[-{Latex}, line width=0.55pt, line cap=round] (0,0) -- (\inlineoplength,0);
  }%
}
\DeclareRobustCommand{\wavymessagearrow}{%
  \tikz[baseline=-0.55ex, x=1ex, y=1ex]{%
    \draw[-{Latex}, line width=0.55pt, decorate, decoration={snake, amplitude=0.25ex, segment length=1.2ex, post length=1ex}] (0,0) -- (\inlineoplength,0);
  }%
}

\makeatletter
\newcommand{\openclosed}{\mathopen{\mathpalette\openclosed@aux\relax}}
\newcommand{\closedopen}{\mathclose{\mathpalette\closedopen@aux\relax}}
\newcommand{\openclosed@aux}[2]{\ooalign{$\m@th#1[$\cr$\m@th#1($\cr}}
\newcommand{\closedopen@aux}[2]{\ooalign{$\m@th#1]$\cr$\m@th#1)$\cr}}
\makeatother

\newcommand{\allMsgs}{M}
\newcommand{\ultMsgs}{M_{\mathrm{ult}}}

\newcommand{\penMsgs}{M_{\mathrm{pen}}}
\newcommand{\latePenMsgs}{M_{\mathrm{late}}}

\DeclareMathOperator{\origin}{origin}
\DeclareMathOperator{\seq}{seq}
\DeclareMathOperator{\rank}{rank}
\DeclareMathOperator{\labelof}{label}
\newcommand{\enc}{\mathcal{E}}
\newcommand{\Label}{\textsc{Label}}
\newcommand{\Req}{\textnormal{\textsc{Req}}}
\newcommand{\Reply}{\textnormal{\textsc{Reply}}}
\newcommand{\Prop}{\textnormal{\textsc{Prop}}}
\newcommand{\Echo}{\textnormal{\textsc{Echo}}}
\newcommand{\signed}[2]{\langle #1\rangle_{\sigma_{#2}}}
\newcommand{\ind}[1]{\mathbf 1\{#1\}}
\newcommand{\merge}{\mathbin{\triangleleft}}
\newcommand{\varfrontier}[1]{\texttt{frontier}_{#1}}
\newcommand{\varlabel}[1]{\texttt{label}_{#1}}
\newcommand{\varacks}[1]{\texttt{acks}_{#1}}
\newcommand{\ltlex}{\mathrel{<_{\mathrm{lex}}}}
\newcommand{\gtlex}{\mathrel{>_{\mathrm{lex}}}}
\newcommand{\gtagcolor}{black!55}
\newcommand{\gtag}[1]{\hfill\textcolor{\gtagcolor}{$#1$}}
\newcommand{\gtagtext}[1]{\textcolor{\gtagcolor}{#1}}

\newcommand{\assertlt}[3]{%
  \pgfmathparse{#1 < #2 ? 1 : 0}%
  \ifnum\pgfmathresult=0
    \PackageError{MyAssertions}{#3}{}%
  \fi
}

\newcommand{\assertle}[3]{%
  \pgfmathparse{#1 <= #2 ? 1 : 0}%
  \ifnum\pgfmathresult=0
    \PackageError{MyAssertions}{#3}{}%
  \fi
}

\colorlet{darkgreen}{green!80!black}
\colorlet{transformshade}{blue!10}
\colorlet{windowshade}{gray!20}
\colorlet{observationcolor}{black!60}
\newcommand{\pzerocolor}{blue!70!black}
\newcommand{\ponecolor}{orange!85!black}
\newcommand{\byzantinecolor}{red}
\tikzset{
    scbase/.style={
        >=Latex,
        line cap=round,
        line join=round,
    },
    sctimeline/.style={line width=1pt},
    sctimelineend/.style={sctimeline, -{Classical TikZ Rightarrow}},
    scwavytimeline/.style={sctimeline, draw=blue, decorate, decoration={snake, amplitude=0.6pt, segment length=4pt}},
    screftime/.style={line width=1pt},
    scendpoint/.style={line width=1pt},
    scconnector/.style={
        dotted,
        draw=blue,
        line width=0.8pt,
        postaction={decorate},
        decoration={markings, mark=at position 0.5 with {\arrow{Stealth}}}
    },
    scspecialconnector/.style={
        scconnector,
        draw=darkgreen,
    },
    scopmark/.style={line width=0.9pt},
    scopspan/.style={<->, line width=0.9pt},
    scstepdot/.style={circle, fill=black, inner sep=1.5pt},
    scbadstep/.style={circle, fill=red, inner sep=1.5pt},
    scspecialstep/.style={rectangle, fill=darkgreen, draw=darkgreen, inner sep=0pt, minimum size=5.5pt},
    scbadspecialstep/.style={regular polygon, regular polygon sides=3, fill=red, draw=red, inner sep=0pt, minimum size=6pt},
    scgoodmsg/.style={->, line width=1.0pt},
    scbadmsg/.style={->, line width=1.0pt, draw=red},
    scwindow/.style={fill=windowshade, draw=none},
    sctlab/.style={fill=white, inner sep=1pt},
    scslab/.style={font=\normalsize},
    screftick/.pic={
        \draw[screftime] (0,-0.17) -- (0,0.17);
    },
    scendpointtick/.pic={
        \draw[scendpoint] (0,-0.2) -- (0,0.2);
    },
    pics/scopabove/.style args={#1/#2/#3}{
        code={
            \draw[scopmark] (0,-\opBarExtend) -- (0,#2+\opBarExtend);
            \draw[scopmark] (#1,-\opBarExtend) -- (#1,#2+\opBarExtend);
            \draw[scopspan] (0,#2) -- (#1,#2);
            \node[sctlab] at ({0.5*(#1)},#2) {$#3$};
        }
    },
    pics/scopbelow/.style args={#1/#2/#3}{
        code={
            \draw[scopmark] (0,#2-\opBarExtend) -- (0,\opBarExtend);
            \draw[scopmark] (#1,#2-\opBarExtend) -- (#1,\opBarExtend);
            \draw[scopspan] (0,#2) -- (#1,#2);
            \node[sctlab] at ({0.5*(#1)},#2) {$#3$};
        }
    },
    scmsgarcraise/.initial=0,
    pics/scgoodmsgarc/.style args={#1}{
        code={
            \pgfmathsetmacro{\scMsgArcHeight}{\arcHeight + \pgfkeysvalueof{/tikz/scmsgarcraise}}
            \draw[scgoodmsg] (0,0)
                parabola bend ({0.5*(#1)}, {\scMsgArcHeight}) (#1,0);
        }
    },
    pics/scbadmsgarc/.style args={#1}{
        code={
            \pgfmathsetmacro{\scMsgArcHeight}{\arcHeight + \pgfkeysvalueof{/tikz/scmsgarcraise}}
            \draw[scbadmsg] (0,0)
                parabola bend ({0.5*(#1)}, {\scMsgArcHeight}) (#1,0);
        }
    },
}

\newcommand{\Crefpart}[2]{\nameCref{#1}~\hyperref[#2]{\ref*{#1}\ref*{#2}}}

\title{On Latency-Optimal Implementations of the Timestamp Object}
\author{%
  Aurojit Panda\thanks{New York University. Email: \email{apanda@cs.nyu.edu}.}
  \and
  Scott Shenker\thanks{UC Berkeley and ICSI. Email: \email{shenker@icsi.berkeley.edu}.}
  \and
  Sam Toueg\thanks{University of Toronto. Email: \email{sam@cs.toronto.edu}.}
  \and
  Wen Zhang\thanks{Google. Email: \email{wnz@google.com}.}
}
\date{}

\begin{document}

\maketitle

\begin{abstract}
Timestamps are widely used in distributed algorithms to order operations.
In this paper, we study the \emph{timestamp object} (\ts{}), which provides a single operation that returns strictly increasing integers.
We ask two questions: What is the maximum \emph{resilience} and the minimum \emph{latency} of \ts{} implementations in asynchronous message-passing systems, where latency is measured by worst-case message delay?
We answer these questions under various types of failure: crash failures, and Byzantine failures with and without digital signatures.
To do so, we give lower bounds and algorithms that determine (1)~the optimal resilience of \ts{} implementations under each type of failure, and (2)~the optimal latency in all cases but one.
Along the way, we introduce a new lower-bound proof technique called a \emph{cyclic chain argument},
which constructs a cycle of indistinguishable schedules along which the returned timestamps must strictly increase, yielding a contradiction.

\end{abstract}

\tableofcontents
\clearpage

\section{Introduction}\label{sec:intro}
The \emph{timestamp object} (\ts{}) provides a single operation, \nextOpParens, which returns a positive integer called a \emph{timestamp}.
Each \nextOpParens{} operation returns a timestamp larger than all previously returned ones; timestamps need not be consecutive.
Many message-passing algorithms, notably implementations of atomic registers, generate ordered tags or timestamps internally (see \Cref{sec:related-work}).
We instead study~\ts{} as an object in its own right, implemented in asynchronous message-passing systems subject to process failures.

We ask two questions: What is the maximum \emph{resilience} and the minimum \emph{latency} of \ts{} implementations?
We focus on systems in which network communication, rather than local computation, dominates the time an operation takes.
Accordingly, we measure latency in terms of \emph{worst-case message delay}: intuitively, the longest time an operation can take when local computation is instantaneous and every message sent during the operation is received within one time unit (we define this precisely in \Cref{sec:delay}).

We answer these questions under various types of failure: crash failures, and Byzantine failures with and without digital signatures.
To do so, we develop lower bounds and algorithms that determine (1)~the optimal resilience of~\ts{} implementations under each type of failure, and (2)~the optimal latency in all cases but one.

Our results are summarized in \Cref{tab:summary}, where $n\geq 2$ is the number of processes in the system and $1\leq f<n$ is the maximum number of faulty processes.
As the table shows, the optimal latency depends on whether more than one process may fail and, under Byzantine failures, on whether digital signatures are available.
Only one case remains open: Byzantine failures without digital signatures when $n>4$ and $f=1$.
There, we narrow the optimal latency down to either two or three message delays.

We also introduce a new lower-bound technique, which we call a \emph{cyclic chain argument}.
The argument constructs a cycle of indistinguishable schedules along which the returned timestamps must strictly increase, yielding a contradiction.

We defer the proofs of several auxiliary results to the appendices.

\suppressfloats[t]
\begingroup
\crefname{theorem}{Thm.}{Thms.}
\Crefname{theorem}{Thm.}{Thms.}
\crefname{proposition}{Prop.}{Props.}
\Crefname{proposition}{Prop.}{Props.}
\crefname{observation}{Obs.}{Obs.}
\Crefname{observation}{Obs.}{Obs.}
\crefname{corollary}{Cor.}{Cors.}
\Crefname{corollary}{Cor.}{Cors.}
\crefname{algorithm}{Alg.}{Algs.}
\Crefname{algorithm}{Alg.}{Algs.}
\newcommand{\resultref}[1]{{\footnotesize(\hyperref[#1]{\Cref*{#1}})}}
\newcommand{\obsref}[1]{{\footnotesize\textcolor{observationcolor}{(\hyperref[#1]{\textcolor{observationcolor}{\Cref*{#1}}})}}}
\newcommand{\onlyifref}[1]{{\footnotesize(only if: \hyperref[#1]{\Cref*{#1}})}}
\newcommand{\onlyifobsref}[1]{{\footnotesize\textcolor{observationcolor}{(only if: \hyperref[#1]{\textcolor{observationcolor}{\Cref*{#1}}})}}}
\begin{table}[htbp]
  \caption[Summary of results]{Summary of results on implementing the timestamp object~(\ts{}), where $n\geq 2$ is the number of processes in the system and $1\leq f<n$ is the maximum number of faulty processes.
  When \ts{} is implementable, we report a lower bound and an upper bound on the \emph{optimal latency} of implementing~\ts{}, i.e., the smallest worst-case message delay of any implementation of~\ts{}.
  A check mark~\checkmark{} indicates that the lower and upper bounds match.
  References \textcolor{observationcolor}{set in gray} point to observations, which follow immediately from results elsewhere in the table.%
 }\label{tab:summary}
  \centering
  \small
  \setlength{\tabcolsep}{5pt}
  \renewcommand{\arraystretch}{1.3}
  \setlength{\aboverulesep}{0pt}
  \setlength{\belowrulesep}{0pt}
  \setlength{\extrarowheight}{1pt} %
  \newcommand{\caserule}{\arrayrulecolor{black!25}\cline{3-9}\arrayrulecolor{black}}
  \newcommand{\caseruleub}{\arrayrulecolor{black!25}\cline{3-6}\cline{9-9}\arrayrulecolor{black}}
  \begin{tabular}{|l|l|ll|r@{\hspace{.4em}}l|r@{\hspace{.4em}}l|c|}
    \toprule
    & & & & \multicolumn{5}{c|}{\textbf{Optimal latency} {\footnotesize(in message delay units)}} \\
    \cline{5-9}
    \textbf{Failure type}
      & \textbf{\ts{} implementable iff}
      & \multicolumn{2}{r|}{\textbf{\# of failures}}
      & \multicolumn{2}{l|}{\textbf{Lower bound}}
      & \multicolumn{2}{l|}{\textbf{Upper bound}}
      & \textbf{Tight?} \\
    \midrule
    \multirow{2}{*}{\textbf{Crash}} & \multirow{2}{*}{\makecell[l]{$n>2f$\\\onlyifref{thm:crash-resilience}}}
      & & $f\geq 2$
      & $3$ & \resultref{thm:crash-lower-bound-fgeq2}
      & $3$ & \resultref{alg:crash-three-delays} & \checkmark \\
      \caserule
      & & & $f=1$
      & $2$ & \resultref{thm:crash-two-delay-lower-bound}
      & $2$ & \resultref{alg:crash-two-delays} & \checkmark \\
    \midrule
    \multirow{2}{*}{\makecell[l]{\textbf{Byzantine}\\\textbf{with signatures}}} & \multirow{2}{*}{\makecell[l]{$n>3f$\\\onlyifref{thm:byzantine-resilience-with-signatures}}}
      & & $f\geq 2$
      & $3$ & \obsref{thm:byz-sig-lower-bound-fgeq2}
      & $3$ & \obsref{thm:byz-sig-upper-bound} & \checkmark \\
      \caserule
    & & & $f=1$
      & $2$ & \obsref{thm:byz-sig-lower-bound-f1}
      & $2$ & \resultref{alg:byz-two-delays-sig} & \checkmark \\
    \midrule
    \multirow{3}{*}{\textbf{Byzantine}} & \multirow{3}{*}{\makecell[l]{$n>3f$\\\onlyifobsref{thm:byzantine-resilience-no-signatures}}}
        &             & $f\geq 2$
      & $3$ & \obsref{thm:byz-lower-bound-fgeq2}
      & \multirow{3}{*}{$3$} & \multirow{3}{*}{\resultref{alg:byz-three-delays}} & \checkmark \\
      \caseruleub
      & & $n=4$,      & $f=1$
      & $3$ & \resultref{thm:byz-lower-bound-n4f1}
      & &
      & \checkmark \\
      \caseruleub
      & & $n>4$,     & $f=1$
      & $2$ & \obsref{thm:byz-lower-bound-f1}
      & &
      & \\
    \bottomrule
  \end{tabular}
\end{table}
\endgroup

\section{System Model and Latency under Crash Failures}\label{sec:model-latency}
We first define the system model and latency measure for crash failures; \Cref{sec:byz-model} adapts these definitions to systems with Byzantine processes.

\subsection{System Model}\label{sec:model}
The system consists of $n\geq 2$ processes connected pairwise by reliable asynchronous channels.
At most $f$~processes may crash, where $1\leq f<n$.
Processes take \emph{steps} asynchronously.
In each step, a process may invoke an operation, receive and send finitely many messages, and return from an operation by producing a response.
A process may invoke an operation in any step, provided its earlier operations have returned.

An \emph{algorithm}~$\alg$ assigns to each process~$p$ a deterministic automaton~$\alg_p$.
Each transition of~$\alg_p$ maps process~$p$'s current state, an optional operation invocation, and the received messages to a new state, outgoing messages, and any response to $p$'s pending operation.

A \emph{schedule}~$S$ of~$\alg$ is an infinite sequence of steps, each taken according to~$\alg$, in which at least $n-f$ processes take infinitely many steps.
The processes that take infinitely many steps are \emph{correct} in~$S$; all other processes are \emph{faulty}.
In a schedule, every message is received at most once and only if it was sent to its recipient in an earlier step, and every message sent to a process that takes infinitely many steps is eventually received.
A \emph{schedule prefix}~$\sigma$ of~$\alg$ is a finite prefix of a schedule of~$\alg$.
An \emph{operation} consists of an invocation step and its corresponding return step, if any.
An operation is \emph{complete} if it returns and \emph{incomplete} otherwise.

\subsection{Worst-Case Message Delay}\label{sec:delay}
With the system model in place, we now turn to measuring the latency of algorithms.
Defining such a measure can be subtle in asynchronous systems, where an algorithm's running time is at the mercy of message and computation delays beyond its control.
(See, for example, recent work by~\citet[Section~6]{DBLP:conf/wdag/BezerraSKR25} for a discussion of the subtleties involved.)

Following convention, we measure an algorithm's latency in units of message delay, an approach well suited to environments in which latency is dominated by communication rather than by local computation.
Specifically, we study an algorithm's \emph{worst-case message delay}:
intuitively, the longest duration an operation can take when local computation is instantaneous and every message sent \emph{during the operation}
is received within one unit of time.
(A similar measure appears as early as \citet[Section~2.1]{DBLP:journals/jacm/Awerbuch85}; see \citet[Section~2.1.2]{DBLP:books/daglib/0017536} for a more detailed exposition.)

A schedule records the order of steps but not the times at which they occur.
To measure the latency of messages and operations, we assign a real time to each step.
These times are used only for analysis and are not accessible to the processes.

\begin{definition}[Run]\label{def:run}
  A \emph{run}~$R=(S,T)$ of~$\alg$ consists of a schedule~$S$ of~$\alg$ and a non-decreasing, unbounded sequence~$T$ of real numbers, where $T[j]$ is the time of step~$S[j]$, such that:
  \begin{enumerate}[label=(\arabic*),ref=\arabic*]
    \item if a message sent in step~$S[j]$ is received in step~$S[k]$, then $T[j]<T[k]$; and
    \item if steps $S[j]$ and~$S[k]$ with $j<k$ are taken by the same process, then $T[j]<T[k]$.
  \end{enumerate}
  A process is \emph{correct} in~$R$ if it is correct in~$S$, and \emph{faulty} otherwise.
\end{definition}

A run assigns a single time to each step, treating it as instantaneous even though it may perform an arbitrary finite amount of local computation.
This choice reflects the intuitive notion of worst-case message delay described in \Cref{sec:intro}, in which local computation takes no time.

We now define the latency of a message and of an operation in a run~$R$.

\begin{definition}[Message latency]
  In a run~$R=(S,T)$, if a message~$m$ is sent in step~$S[j]$ and received in step~$S[k]$, then its \emph{latency} in~$R$ is $\latency[R]{m} \deq T[k] - T[j]$.
\end{definition}

\begin{definition}[Timely message]
  In a run, a message is \emph{timely} if its latency is at most one unit of time.
\end{definition}

\begin{definition}
  In a run $R=(S,T)$, let $O$ be an operation invoked in step~$S[a]$.
  \begin{itemize}
    \item The \emph{time interval} of~$O$ is
      \[
        \interval[R]{O} \Deq
        \begin{cases}
          [T[a], T[b]], & \text{if $O$ returns in step~$S[b]$},\\
          [T[a], \infty), & \text{if $O$ is incomplete}.
        \end{cases}
      \]
    \item The \emph{latency} of~$O$, denoted~$\latency[R]{O}$, is the length of~$\interval[R]{O}$.
      In particular, $\latency[R]{O} = \infty$ if $O$ is incomplete.
    \item A message is \emph{sent during~$O$} if its send time lies in~$\interval[R]{O}$.
  \end{itemize}
\end{definition}

\begin{definition}\label{def:message-timeliness}
  In a run~$R=(S,T)$, let $L$ be the set of processes that take infinitely many steps in~$S$, and let $O$ be an operation.
  We say that \emph{messages are timely during~$O$} if every message sent during~$O$ to a process in~$L$ is timely.
\end{definition}

All these notions are relative to a run, which we specify explicitly when needed.
When the run~$R$ is clear from context, we omit it from the notation, writing $\interval{O}$, $\latency{m}$, and~$\latency{O}$.

\begin{definition}[Worst-case message delay]\label{def:worst-case-message-delay}
  The \emph{worst-case message delay} of an algorithm~$\alg$ is
  \[
    \latency{\alg} \Deq \sup \left\{ \latency[R]{O}~\middle|
    \begin{array}{l}
      \text{$R$ is a run of~$\alg$}, \\
      \text{$O$ is an operation in~$R$ invoked by a correct process in~$R$}, \\
      \text{messages are timely during~$O$}
    \end{array}
    \right\}.
  \]
\end{definition}

\begin{samepage}
If $\latency{\alg}$ is finite, two simple but useful consequences follow:

\begin{restatable}{proposition}{livenessFromFiniteDelay}\label{thm:liveness-from-finite-delay}
  If $\latency{\alg} < \infty$, then $\alg$ satisfies liveness: In every schedule of~$\alg$, every operation invoked by a correct process eventually returns.
\end{restatable}
\end{samepage}

\begin{restatable}{proposition}{integerWorstCaseDelay}\label{thm:integer-worst-case-delay}
  If $\latency{\alg} < \infty$, then $\latency{\alg} = k$ for some nonnegative integer~$k$.
\end{restatable}

Accordingly, whenever we refer to a worst-case message delay of~$k$, we take~$k$ to be a nonnegative integer.

Both propositions are intuitively clear, but their proofs require some technical care; we defer them to~\Cref{sec:finite-delay}.

\section{The Timestamp Object}\label{sec:timestamp}
We now define the timestamp object and the requirements for implementing it.

\begin{definition}[Timestamp object]\label{def:timestamp-object}
  The \emph{timestamp object}, denoted~\ts{}, has a single operation \nextOpParens{} and the following (nondeterministic) sequential specification:
  \[
    \begin{array}{@{}l@{}}
      \text{\textbf{State:} $\texttt{last}\in\mathbb{Z}_{\ge 0}$, initially~$0$.} \\[0.25ex]
      \nextOpParens\textbf{:} \\
      \qquad \text{choose any $\Delta\in\mathbb{Z}_{\ge 1}$} \\
      \qquad \texttt{last}\gets\texttt{last}+\Delta \\
      \qquad \text{return~\texttt{last}}
    \end{array}
  \]
\end{definition}

Thus, the sequential specification of~\ts{} requires each \nextOpParens{} operation to return a positive integer strictly larger than every value returned by an earlier operation.

We require implementations of~\ts{} to satisfy linearizability with respect to this specification.
Following \citet{DBLP:journals/toplas/HerlihyW90}, a \emph{history} is the sequence of operation invocations and responses in a schedule.
In a history, an operation \emph{precedes} another if its response occurs before the other's invocation.
Intuitively, a history~$H$ is \emph{(crash) linearizable} with respect to~\ts{} if, after discarding some incomplete operations and appending responses to the remaining incomplete operations, the resulting operations can be arranged into a sequential history~$H'$ that
(1)~satisfies the sequential specification of~\ts{} and
(2)~preserves the precedence of operations in~$H$.

\begin{definition}[\ts{} implementation]\label{def:timestamp-implementation}
  Let $\alg$ be an algorithm whose \nextOpParens{} operation returns positive integers.
  We say that $\alg$ \emph{implements}~\ts{} in a system with crash failures if
  every schedule~$S$ of~$\alg$, with corresponding history~$H$,
  satisfies the following conditions:
  \begin{enumerate}[label=(CT\arabic*),ref=CT\arabic*,leftmargin=*]
    \item\label{cond:timestamp-safety} \emph{Safety:} $H$ is linearizable with respect to~\ts{}.
    \item\label{cond:timestamp-liveness} \emph{Liveness:} Every operation invoked by a correct process in~$S$ eventually returns.
  \end{enumerate}
  We call such an algorithm~$\alg$ an \emph{implementation of~\ts{}}.
\end{definition}

The following equivalent characterization of \ts{} implementations replaces the linearizability condition with two conditions on the values returned by complete operations.

\begin{proposition}[Characterization of \ts{} implementations]\label{thm:timestamp-characterization}
  Let $\alg$ be an algorithm whose \nextOpParens{} operation returns positive integers.
  Then $\alg$ implements~\ts{} in a system with crash failures if and only if
  every schedule~$S$ of~$\alg$, with corresponding history~$H$,
  satisfies the following conditions:
  \begin{enumerate}[label=(\getrefnumber{cond:timestamp-safety}\alph*),ref=\getrefnumber{cond:timestamp-safety}\alph*,leftmargin=*]
    \item\label{cond:timestamp-ordering} \emph{Ordering:} If a complete operation~$O_1$ precedes a complete operation~$O_2$ in~$H$, then~$O_1$ returns a smaller value than~$O_2$.
    \item\label{cond:timestamp-uniqueness} \emph{Uniqueness:} No two distinct complete operations in~$H$ return the same value.
    \item[(\getrefnumber{cond:timestamp-liveness})] \emph{Liveness:} Every operation invoked by a correct process in~$S$ eventually returns.
  \end{enumerate}
\end{proposition}
\begin{proof}
  \proofqed{thm:timestamp-characterization}
  The liveness condition~(\ref{cond:timestamp-liveness}) is the same in \Cref{def:timestamp-implementation} and the proposition.
  It therefore suffices to show that, for the history of every schedule of~$\alg$, (\ref{cond:timestamp-safety}) $\Longleftrightarrow$ (\ref{cond:timestamp-ordering}) $\wedge$ (\ref{cond:timestamp-uniqueness}).
  We prove the two directions separately, in a form that we will reuse for the Byzantine analogue (\Cref{thm:byz-timestamp-characterization}).

  \begin{claim}\label{claim:timestamp-characterization-forward}
    For every history~$H$, (\ref{cond:timestamp-safety}) $\Longrightarrow$ (\ref{cond:timestamp-ordering}) $\wedge$ (\ref{cond:timestamp-uniqueness}).
  \end{claim}
  \begin{proof}
    \proofqed{claim:timestamp-characterization-forward}
    Suppose that $H$ satisfies~(\ref{cond:timestamp-safety}).
    Since $H$ is linearizable with respect to~\ts{}, we can discard some incomplete operations from~$H$, append responses to the remaining incomplete operations, and arrange the resulting operations into a sequential history~$H'$ that (1)~satisfies the sequential specification of~\ts{} and (2)~preserves the precedence of operations in~$H$.
    Every complete operation in~$H$ therefore appears in~$H'$ with the same return value.
    By the sequential specification of~\ts{}, each \nextOpParens{} operation increases \texttt{last} by some $\Delta\geq 1$ and returns the new value of \texttt{last}.
    Hence, the values returned by successive operations in~$H'$ are strictly increasing.

    For~(\ref{cond:timestamp-ordering}), suppose that a complete operation~$O_1$ precedes a complete operation~$O_2$ in~$H$.
    Since $H'$ preserves the precedence of operations in~$H$, $O_1$ appears before~$O_2$ in~$H'$ and therefore returns a smaller value.
    For~(\ref{cond:timestamp-uniqueness}), let $O_1$ and~$O_2$ be distinct complete operations in~$H$.
    Both appear in~$H'$ with the same return values as in~$H$.
    Since $H'$ is sequential, one appears before the other; because return values strictly increase in~$H'$, they return distinct values.
  \end{proof}

  \begin{claim}\label{claim:timestamp-characterization-backward}
    For every history~$H$ whose complete operations return positive integers, (\ref{cond:timestamp-ordering}) $\wedge$ (\ref{cond:timestamp-uniqueness}) $\Longrightarrow$ (\ref{cond:timestamp-safety}).
  \end{claim}
  \begin{proof}
    \proofqed{claim:timestamp-characterization-backward}
    Suppose that $H$ satisfies~(\ref{cond:timestamp-ordering}) and~(\ref{cond:timestamp-uniqueness}).
    Discard all incomplete operations from~$H$, and arrange the complete operations into a sequential history~$H'$ in increasing order of their return values.%
    \footnote{The history~$H'$ is well defined even if $H$ contains infinitely many complete operations: The complete operations in~$H$ return positive integers and, by~(\ref{cond:timestamp-uniqueness}), have distinct return values. Thus, an operation returning~$v$ has at most~$v-1$ complete operations with smaller return values, so the complete operations can be listed in increasing order of return value.}
    To show that $H$ is linearizable with respect to~\ts{}, it remains to show that $H'$ satisfies the sequential specification of~\ts{} and preserves the precedence of operations in~$H$.

    The return values in~$H'$ are positive integers and, by construction and~(\ref{cond:timestamp-uniqueness}), strictly increasing.
    Thus, $H'$ satisfies the sequential specification of~\ts{}.

    Next, suppose that a complete operation~$O_1$ precedes a complete operation~$O_2$ in~$H$.
    By~(\ref{cond:timestamp-ordering}), $O_1$ returns a smaller value than~$O_2$.
    Since $H'$ orders operations by increasing return value, $O_1$ appears before~$O_2$ in~$H'$.
    Hence, $H'$ preserves the precedence of operations in~$H$.

    Therefore, $H$ is linearizable with respect to~\ts{}.
  \end{proof}

  Since the \nextOpParens{} operation of~$\alg$ returns positive integers, \Cref{claim:timestamp-characterization-forward,claim:timestamp-characterization-backward} prove the proposition.
\end{proof}

\section{Results for Crash Failures}\label{sec:crash}
In this section, we consider implementations of \ts{} in systems with crash failures.
Recall that we assume there are $n \geq 2$ processes in the system, of which at most $f \geq 1$ may crash.

\subsection{Maximum Resilience}

\begin{proposition}\label{thm:crash-resilience}
  In a system with crash failures, there exists an implementation of~\ts{} only if $n>2f$.
\end{proposition}
\begin{proof}
  \proofqed{thm:crash-resilience}
  By a standard partitioning argument.
\end{proof}

\subsection{Latency Upper Bounds}

We now give two upper bounds on the worst-case message delay of implementing~\ts{} under crash failures, first for $f=1$ and then for general~$f$.

\subsubsection{Two-Delay Algorithm for \texorpdfstring{$f=1$}{f=1}}\label{sec:crash2delay}

When $f=1$, \ts{} can be implemented in two message delays.

\begin{theorem}\label{thm:crash-upper-bound-f1}
  In a system with crash failures, if $n>2f$ and $f=1$, then there exists an implementation of~\ts{} (\Cref{alg:crash-two-delays}) whose worst-case message delay is two.
\end{theorem}

\begin{proof}
  \proofqed{thm:crash-upper-bound-f1}
We first establish some supporting properties of the algorithm.%
\footnote{In our pseudocode, \textbf{upon} blocks are uninterruptible: Each runs to completion before another begins.}
We then show that the algorithm implements~$\ts{}$ by proving that it returns positive integers and satisfies ordering~(\ref{cond:timestamp-ordering}), uniqueness~(\ref{cond:timestamp-uniqueness}), and liveness~(\ref{cond:timestamp-liveness}).
Finally, we show that its worst-case message delay is two.

\begin{algorithm}
  \caption{A worst-case two-delay implementation of~$\ts{}$ in a system with crash failures where $n>2f$ and $f=1$.}\label{alg:crash-two-delays}
  \begin{algorithmic}[1]
    \Statex {\bfseries\boldmath Code for process~$p_i$} ($0 \leq i < n$):
    \State $\xx_i \gets 0$ \Comment{$p_i$'s \emph{local clock}.}
    \State $\texttt{seq}_i \gets 0$ \Comment{Sequence number of $p_i$'s ongoing invocation, or of its next invocation when idle.}
    \State $\texttt{reqSeq}_i[j] \gets -1$ for each $j \neq i$ \Comment{Largest $q$ in a $\langle \Req~q,-\rangle$ received from $p_j$.}
    \Statex

    \State $\texttt{phase}_i \gets \textsf{idle}$ \Comment{Whether $p_i$ is \textsf{idle} or has an \textsf{ongoing} invocation.}
    \State $\texttt{replies}_i \gets \emptyset$ \Comment{Set of $\langle$sender index, reply value$\rangle$ received for $p_i$'s ongoing invocation.}
    \Statex

    \Upon{\textbf{invocation} of \nextOpParens{}}
      \State $\texttt{replies}_i \gets \emptyset$
      \State $\texttt{phase}_i \gets \textsf{ongoing}$
      \State \textbf{send} $\langle \Req~\texttt{seq}_i, \xx_i \rangle$ \textbf{to} every process except~$p_i$ \label{line:crash-two-delays-broadcast}
    \EndUpon
    \Statex

    \Upon{\textbf{receive} $\langle \Req~q, v \rangle$ \textbf{from} $p_j$}
      \If{$q > \texttt{reqSeq}_i[j]$} \label{line:crash-two-delays-fresh-request}
        \State $\texttt{reqSeq}_i[j] \gets q$
        \State $\xx_i \gets \max(\xx_i, v) + 1$ \label{line:crash-two-delays-reply-max}
        \State \textbf{send} $\langle \Reply~q, \xx_i \rangle$ \textbf{to} $p_j$ \label{line:crash-two-delays-send-reply}
      \EndIf
    \EndUpon
    \Statex

    \Upon{\textbf{receive} $\langle \Reply~q, v \rangle$ \textbf{from} $p_j$}
      \If{$\texttt{phase}_i = \textsf{ongoing}$ \textbf{and} $q = \texttt{seq}_i$}
        \State $\texttt{replies}_i \gets \texttt{replies}_i \cup \{\langle j, v \rangle\}$ \label{line:crash-two-delays-record-reply}
        \If{$|\texttt{replies}_i| = n-2$} \label{line:crash-two-delays-enough-replies}
          \State $t \gets \max \{v' : \langle j', v' \rangle \in \texttt{replies}_i\}$ \Comment{\emph{Base timestamp}.}\label{line:crash-two-delays-return-compute}
          \State $\xx_i \gets \max(\xx_i, t)$ \label{line:crash-two-delays-return-update}
          \State $\texttt{phase}_i \gets \textsf{idle}$
          \State $\texttt{seq}_i \gets \texttt{seq}_i + 1$
          \State \Return $n\cdot t+i$ \Comment{Completes the \nextOpParens{} invocation.}\label{line:crash-two-delays-return}
        \EndIf
      \EndIf
    \EndUpon
  \end{algorithmic}
\end{algorithm}

\begin{claim}\label{claim:crash-two-delays-clock-monotone}
  For every process~$p_i$, its local clock~$\xx_i$ is non-decreasing.
\end{claim}
\begin{proof}
  \proofqed{claim:crash-two-delays-clock-monotone}
  After initialization, only lines~\ref{line:crash-two-delays-reply-max} and~\ref{line:crash-two-delays-return-update} assign to $\xx_i$, and neither decreases it.
\end{proof}

\begin{claim}\label{claim:crash-two-delays-return-support}
  If an invocation of~\nextOpParens{} computes base timestamp~$t$ on line~\ref{line:crash-two-delays-return-compute}, then at the time of return, at least two distinct processes~$p_j$ satisfy $\xx_j\geq t$.
\end{claim}
\begin{proof}
  \proofqed{claim:crash-two-delays-return-support}
  Let $p_i$ be the process performing the invocation.
  Write $R=\{j : \langle j,-\rangle\in\texttt{replies}_i\}$ when $p_i$ computes~$t$ on line~\ref{line:crash-two-delays-return-compute}.
  By line~\ref{line:crash-two-delays-enough-replies}, $|\texttt{replies}_i|=n-2$; by line~\ref{line:crash-two-delays-fresh-request}, each sender contributes at most one element to~$\texttt{replies}_i$, so $|R|=n-2\geq 1$.
  Moreover, $p_i$ sends no request to itself (line~\ref{line:crash-two-delays-broadcast}) and therefore receives no reply from itself, so $i\notin R$.

  For each $j\in R$, let $v_j$ be the unique value such that $\langle j,v_j\rangle\in\texttt{replies}_i$ when $p_i$ computes~$t$ (line~\ref{line:crash-two-delays-return-compute}).
  Since $t=\max_{j\in R} v_j$, choose $\ell\in R$ such that $v_\ell=t$.
  By lines~\ref{line:crash-two-delays-reply-max} and~\ref{line:crash-two-delays-send-reply}, process~$p_\ell$ has $\xx_\ell=v_\ell=t$ when it sends its reply.
  By \Cref{claim:crash-two-delays-clock-monotone}, $\xx_\ell\geq t$ still holds when $p_i$ returns.

  For the invoking process~$p_i$, line~\ref{line:crash-two-delays-return-update} sets $\xx_i\gets\max(\xx_i,t)$ before line~\ref{line:crash-two-delays-return} returns.
  Hence, $\xx_i\geq t$ at the time of return.
  Since $\ell\in R$ and $i\notin R$, the processes $p_\ell$ and $p_i$ are distinct, and the claim follows.
\end{proof}

\begin{claim}\label{claim:crash-two-delays-later-return}
  Let $h$ be an integer, and suppose that when a process~$p_i$ invokes~\nextOpParens{}, there are distinct processes~$p_j$ and~$p_k$ such that $\xx_j\geq h$ and $\xx_k\geq h$.
  Then, if the invocation completes, its base timestamp~$t$ satisfies $t>h$.
\end{claim}
\begin{proof}
  \proofqed{claim:crash-two-delays-later-return}
  Let $v_i$ be the value of~$\xx_i$ included in the \Req{} message that $p_i$ sends for this invocation on line~\ref{line:crash-two-delays-broadcast}.
  When $p_i$ computes~$t$ on line~\ref{line:crash-two-delays-return-compute}, write $R=\{j : \langle j,-\rangle\in\texttt{replies}_i\}$.
  As argued in the proof of \Cref{claim:crash-two-delays-return-support}, $|R|=n-2$ and $i\notin R$.

  If $i\in \{j,k\}$, then $v_i\geq h$ by assumption and monotonicity (\Cref{claim:crash-two-delays-clock-monotone}).
  For each $\ell\in R$, line~\ref{line:crash-two-delays-reply-max} makes $p_\ell$'s reply value at least $v_i+1>h$.
  Since $t$ is the maximum of these reply values (line~\ref{line:crash-two-delays-return-compute}), we have $t>h$.

  Otherwise, $i\notin \{j,k\}$.
  Because $|R|=n-2$ and $i\notin R$, the set~$R$ omits only one of the $n-1$ indices other than~$i$; hence $j\in R$ or $k\in R$.
  Without loss of generality, assume $j\in R$.
  By \Cref{claim:crash-two-delays-clock-monotone}, $\xx_j\geq h$ when $p_j$ receives the \Req{} message from~$p_i$.
  Line~\ref{line:crash-two-delays-reply-max} therefore makes $p_j$'s reply value greater than~$h$.
  This value is among those whose maximum is~$t$ (line~\ref{line:crash-two-delays-return-compute}), so $t>h$.
\end{proof}

\begin{claim}[Ordering for base timestamps]\label{claim:crash-two-delays:ordering}
  Suppose two \nextOpParens{} operations compute base timestamps~$t$ and~$t'$ on line~\ref{line:crash-two-delays-return-compute}, respectively.
  If the first returns before the second is invoked, then $t<t'$.
\end{claim}
\begin{proof}
  \proofqed{claim:crash-two-delays:ordering}
  By \Cref{claim:crash-two-delays-return-support}, when the first operation returns, there exist distinct processes $p_j$ and $p_k$ such that $\xx_j\geq t$ and $\xx_k\geq t$.
  By \Cref{claim:crash-two-delays-clock-monotone}, both inequalities still hold when the second operation is invoked.
  Applying \Cref{claim:crash-two-delays-later-return} with $h=t$ therefore shows that the second operation has base timestamp $t'>t$.
\end{proof}

\begin{claim}[Return values]\label{claim:crash-two-delays:return-values}
  \Cref{alg:crash-two-delays} returns positive integers from \nextOpParens{} and satisfies ordering~(\ref{cond:timestamp-ordering}) and uniqueness~(\ref{cond:timestamp-uniqueness}).
\end{claim}
\begin{proof}
  \proofqed{claim:crash-two-delays:return-values}
  Local clocks~$\xx_i$ start at~$0$ and never decrease (\Cref{claim:crash-two-delays-clock-monotone}), so by line~\ref{line:crash-two-delays-reply-max}, every value sent in a \Reply{} message is a positive integer.
  Hence, every base timestamp~$t$ is a positive integer, and so is $n\cdot t+i$.

  For~(\ref{cond:timestamp-ordering}), suppose a complete operation at~$p_i$ with base timestamp~$t$ precedes one at~$p_j$ with base timestamp~$t'$.
  By \Cref{claim:crash-two-delays:ordering}, $t<t'$, so $n\cdot t+i<n\cdot(t+1)\leq n\cdot t'+j$.

  For~(\ref{cond:timestamp-uniqueness}), timestamps returned by different processes differ modulo~$n$, and those returned by the same process are distinct by~(\ref{cond:timestamp-ordering}) since its operations are sequential.
\end{proof}

\begin{claim}[Liveness]\label{claim:crash-two-delays:liveness}
  In \Cref{alg:crash-two-delays}, every invocation by a correct process eventually returns.
\end{claim}
\begin{proof}
  \proofqed{claim:crash-two-delays:liveness}
  Let $p_i$ be a correct process that invokes~\nextOpParens{}.
  Line~\ref{line:crash-two-delays-broadcast} sends a \Req{} message to each of the other $n-1$ processes, at least $n-2$ of which are correct.
  Each correct recipient eventually receives this fresh request, passes the check on line~\ref{line:crash-two-delays-fresh-request}, and sends a \Reply{} on line~\ref{line:crash-two-delays-send-reply}, which $p_i$ eventually receives.
  Thus line~\ref{line:crash-two-delays-record-reply} eventually makes $|\texttt{replies}_i|=n-2$, causing $p_i$ to pass the check on line~\ref{line:crash-two-delays-enough-replies} and return on line~\ref{line:crash-two-delays-return}.
\end{proof}

\begin{samepage}
\begin{claim}[Latency]\label{claim:crash-two-delays:latency}
  The worst-case message delay of \Cref{alg:crash-two-delays} is two.
\end{claim}
\begin{proof}[Proof sketch]
  \proofqed{claim:crash-two-delays:latency}
  An invocation follows one \Req{} delay and one \Reply{} delay, as illustrated below:
  \[
    \begin{tikzpicture}[
      baseline=(current bounding box.south),
      event/.style={circle, fill, inner sep=2pt},
      msg/.style={-{Stealth[length=2mm]}, thin},
      every node/.style={font=\small}
    ]
      \node[event, label={[align=right, xshift=-1em, yshift=-2pt]below:{invoke\\\nextOpParens{}}}] (invoke) at (0,0) {};
      \node[event] (req1) at (2.4,0.7) {};
      \node[event] (req2) at (2.4,-0.7) {};
      \node[event, label={[align=left, xshift=1em, yshift=-2pt]below:{return}}] (return) at (4.8,0) {};

      \node[fill=white, inner sep=1pt] at (1.2,1.1) {\Req};
      \node[fill=white, inner sep=1pt] at (3.6,1.1) {\Reply};

      \foreach \a in {req1,req2} {
        \draw[msg] (invoke) -- (\a);
        \draw[msg] (\a) -- (return);
      }
    \end{tikzpicture}\qedhere
  \]
\end{proof}
\end{samepage}

By \Cref{claim:crash-two-delays:return-values,claim:crash-two-delays:liveness}, \Cref{alg:crash-two-delays} satisfies the conditions of \Cref{thm:timestamp-characterization} and hence implements~\ts{}.
By \Cref{claim:crash-two-delays:latency}, its worst-case message delay is two.
\end{proof}

\subsubsection{Three-Delay Algorithm for General \texorpdfstring{$f$}{f}}\label{sec:crash3delay}

In the general case, \ts{} can be implemented in three message delays.

\begin{theorem}\label{thm:crash-upper-bound-fgeq2}
  In a system with crash failures, if $n>2f$, then there exists an implementation of~\ts{} (\Cref{alg:crash-three-delays}) whose worst-case message delay is three.
\end{theorem}

\begin{proof}
  \proofqed{thm:crash-upper-bound-fgeq2}
The algorithm adopts the ``max-min technique''~\cite[Section~1.1]{DBLP:conf/podc/DuttaGLC04} previously used to achieve worst-case three-delay reads for SWMR and MWMR atomic registers in message-passing systems~\cite{DBLP:conf/netys/HadjistasiNS17}.
We first establish some supporting properties of the algorithm.
We then show that the algorithm implements~$\ts{}$ by proving that it returns positive integers and satisfies ordering~(\ref{cond:timestamp-ordering}), uniqueness~(\ref{cond:timestamp-uniqueness}), and liveness~(\ref{cond:timestamp-liveness}).
Finally, we show that its worst-case message delay is three.

\begin{algorithm}
  \caption{A worst-case three-delay implementation of~$\ts{}$ in a system with crash failures where $n>2f$.}\label{alg:crash-three-delays}
  \begin{algorithmic}[1]
    \Statex {\bfseries\boldmath Code for process~$p_i$} ($0 \leq i < n$):
    \State $\xx_i \gets 0$ \Comment{$p_i$'s \emph{local clock}.}
    \State $\texttt{seq}_i \gets 0$ \Comment{Sequence number of $p_i$'s ongoing invocation, or of its next invocation when idle.}
    \State $\texttt{reqSeq}_i[k] \gets -1$ for each $k$ \Comment{Largest $q$ in a $\langle \Req~q\rangle$ received from $p_k$.}
    \State $\texttt{propSeq}_i[k] \gets -1$ for each $k$ \Comment{Largest $q$ in a received $\langle \Prop~\langle k,q\rangle,-\rangle$.}
    \State $\texttt{props}_i[k] \gets \emptyset$ for each $k$
    \Statex \Comment{Set of $\langle$sender index, proposal value$\rangle$ received for sequence number $\texttt{propSeq}_i[k]$ of~$p_k$.}
    \Statex

    \State $\texttt{phase}_i \gets \textsf{idle}$ \Comment{Whether $p_i$ is \textsf{idle} or has an \textsf{ongoing} invocation.}
    \State $\texttt{replies}_i \gets \emptyset$ \Comment{Set of $\langle$sender index, reply value$\rangle$ received for $p_i$'s ongoing invocation.}
    \Statex

    \Upon{\textbf{invocation} of \nextOpParens{}}
      \State $\texttt{replies}_i \gets \emptyset$
      \State $\texttt{phase}_i \gets \textsf{ongoing}$
      \State \textbf{send} $\langle \Req~\texttt{seq}_i \rangle$ \textbf{to} every process \label{line:crash-three-delays-req-send}
    \EndUpon
    \Statex

    \Upon{\textbf{receive} $\langle \Req~q \rangle$ \textbf{from} $p_k$}
      \If{$q > \texttt{reqSeq}_i[k]$} \label{line:crash-three-delays-fresh-request}
        \State $\texttt{reqSeq}_i[k] \gets q$ \label{line:crash-three-delays-reqseq-update}
        \State \textbf{send} $\langle \Prop~\langle k,q\rangle, \xx_i \rangle$ \textbf{to} every process \label{line:crash-three-delays-prop-send}
      \EndIf
    \EndUpon
    \Statex

    \Upon{\textbf{receive} $\langle \Prop~\langle k,q\rangle, v \rangle$ \textbf{from} $p_j$}
      \If{$q > \texttt{propSeq}_i[k]$}
        \State $\texttt{propSeq}_i[k] \gets q$ \label{line:crash-three-delays-propseq-update}
        \State $\texttt{props}_i[k] \gets \emptyset$
      \EndIf
      \If{$q = \texttt{propSeq}_i[k]$ \textbf{and} $|\texttt{props}_i[k]| < n-f$} \label{line:crash-three-delays-props-guard}
        \State $\texttt{props}_i[k] \gets \texttt{props}_i[k] \cup \{\langle j,v\rangle\}$ \label{line:crash-three-delays-props-add}
        \If{$|\texttt{props}_i[k]| = n-f$} \label{line:crash-three-delays-props-full}
          \State $v_{\max} \gets \max\{v' : \langle j',v'\rangle \in \texttt{props}_i[k]\}$
          \State $\xx_i \gets \max(\xx_i,v_{\max}) + 1$ \label{line:crash-three-delays-clock-update}
          \State \textbf{send} $\langle \Reply~q, \xx_i \rangle$ \textbf{to} $p_k$ \label{line:crash-three-delays-reply-send}
        \EndIf
      \EndIf
    \EndUpon
    \Statex

    \Upon{\textbf{receive} $\langle \Reply~q, v \rangle$ \textbf{from} $p_j$}
      \If{$\texttt{phase}_i = \textsf{ongoing}$ \textbf{and} $q = \texttt{seq}_i$}
        \State $\texttt{replies}_i \gets \texttt{replies}_i \cup \{\langle j,v\rangle\}$ \label{line:crash-three-delays-record-reply}
        \If{$|\texttt{replies}_i| = n-f$} \label{line:crash-three-delays-replies-full}
          \State $t \gets \min \{v' : \langle j',v'\rangle \in \texttt{replies}_i\}$ \Comment{\emph{Base timestamp}.}\label{line:crash-three-delays-return-compute}
          \State $\texttt{phase}_i \gets \textsf{idle}$
          \State $\texttt{seq}_i \gets \texttt{seq}_i + 1$
          \State \Return $n\cdot t+i$ \Comment{Completes the \nextOpParens{} invocation.}\label{line:crash-three-delays-return}
        \EndIf
      \EndIf
    \EndUpon
  \end{algorithmic}
\end{algorithm}

\begin{claim}\label{claim:crash-three-delays-clock-monotone}
  For every process~$p_i$, its local clock~$\xx_i$ is non-decreasing.
\end{claim}
\begin{proof}
  \proofqed{claim:crash-three-delays-clock-monotone}
  After initialization, the only assignment to~$\xx_i$ is on line~\ref{line:crash-three-delays-clock-update}, which cannot decrease~$\xx_i$.
\end{proof}

\begin{claim}\label{claim:crash-three-delays-unique-messages}
  For all processes~$p_i$ and~$p_j$ and every sequence number~$q$, process~$p_j$ sends at most one $\langle \Prop~\langle i,q\rangle,-\rangle$ message to each process and at most one $\langle \Reply~q,-\rangle$ message to~$p_i$.
\end{claim}
\begin{proof}
  \proofqed{claim:crash-three-delays-unique-messages}
  Process~$p_j$ sends $\langle \Prop~\langle i,q\rangle,-\rangle$ on line~\ref{line:crash-three-delays-prop-send} only upon receiving $\langle \Req~q\rangle$ from~$p_i$ with $q>\texttt{reqSeq}_j[i]$, and it sets $\texttt{reqSeq}_j[i]\gets q$ on line~\ref{line:crash-three-delays-reqseq-update} before sending.
  As line~\ref{line:crash-three-delays-reqseq-update} is the only assignment to~$\texttt{reqSeq}_j[i]$ and never decreases it, the test $q>\texttt{reqSeq}_j[i]$ fails thereafter, so $p_j$ sends no further \Prop{} message for~$\langle i,q\rangle$.

  Process~$p_j$ sends $\langle \Reply~q,-\rangle$ to~$p_i$ on line~\ref{line:crash-three-delays-reply-send} only when an insertion on line~\ref{line:crash-three-delays-props-add} brings $|\texttt{props}_j[i]|$ to $n-f$ while $\texttt{propSeq}_j[i]=q$.
  Thereafter, the test on line~\ref{line:crash-three-delays-props-guard} fails for every later $\langle \Prop~\langle i,q\rangle,-\rangle$ message received, so no further insertion---and hence no further reply---occurs while $\texttt{propSeq}_j[i]=q$.
  Moreover, any later assignment to~$\texttt{propSeq}_j[i]$ occurs on line~\ref{line:crash-three-delays-propseq-update}, where it can only increase the value past~$q$, so $\texttt{propSeq}_j[i]$ never equals~$q$ again.
\end{proof}

\begin{claim}\label{claim:crash-three-delays-return-support}
  If an invocation of~\nextOpParens{} computes base timestamp~$t$ on line~\ref{line:crash-three-delays-return-compute}, then at the time of return, at least $n-f$ processes~$p_j$ satisfy $\xx_j\geq t$.
\end{claim}
\begin{proof}
  \proofqed{claim:crash-three-delays-return-support}
  Let $p_i$ be the process performing the invocation, and consider the time when $p_i$ executes line~\ref{line:crash-three-delays-return-compute} before returning.
  Write $R=\{j : \langle j,-\rangle\in\texttt{replies}_i\}$.
  Line~\ref{line:crash-three-delays-return-compute} is reached only when $|\texttt{replies}_i|=n-f$ (line~\ref{line:crash-three-delays-replies-full}).
  Every pair in~$\texttt{replies}_i$ comes from a \Reply{} message with sequence number~$\texttt{seq}_i$, and by \Cref{claim:crash-three-delays-unique-messages} each process sends at most one such message to~$p_i$, so we have $|R|=n-f$.

  Fix an index~$j\in R$, and let $v_j$ be the value with $\langle j,v_j\rangle \in \texttt{replies}_i$.
  By line~\ref{line:crash-three-delays-return-compute}, $t\leq v_j$.
  Process~$p_j$ had $\xx_j=v_j$ when it sent this reply on line~\ref{line:crash-three-delays-reply-send}, so by \Cref{claim:crash-three-delays-clock-monotone}, $\xx_j\geq v_j\geq t$ still holds when $p_i$ returns.
  As $j\in R$ was arbitrary and $|R|=n-f$, the claim follows.
\end{proof}

\begin{claim}\label{claim:crash-three-delays-later-return}
  Let $h$ be an integer, and suppose that when a process~$p_i$ invokes~\nextOpParens{}, at least $n-f$ processes~$p_k$ have $\xx_k\geq h$.
  Then, if the invocation completes, its base timestamp~$t$ satisfies $t>h$.
\end{claim}
\begin{proof}
  \proofqed{claim:crash-three-delays-later-return}
  Let $A$ be a set of $n-f$ indices such that $\xx_k\geq h$ for every $k\in A$ when the invocation begins.
  When $p_i$ executes line~\ref{line:crash-three-delays-return-compute}, write $R=\{j : \langle j,-\rangle\in\texttt{replies}_i\}$.
  Fix an index~$j\in R$, and let $v_j$ be the value with $\langle j,v_j\rangle \in \texttt{replies}_i$.
  Before $p_j$ sends this reply on line~\ref{line:crash-three-delays-reply-send}, $|\texttt{props}_j[i]|=n-f$ (line~\ref{line:crash-three-delays-props-full}); by \Cref{claim:crash-three-delays-unique-messages}, the pairs in~$\texttt{props}_j[i]$ carry distinct indices, so $p_j$ received \Prop{} messages for this invocation from a set~$P_j$ of $n-f$ process indices.
  Since $|A|,|P_j|\geq n-f$ and $n>2f$, we have $|A\cap P_j|>0$.

  Fix some $k\in A\cap P_j$. As $k\in A$, we have $\xx_k\geq h$ when $p_i$ invoked~\nextOpParens{}.
  Process~$p_k$ sent its \Prop{} message on line~\ref{line:crash-three-delays-prop-send} upon receiving $p_i$'s \Req{} message from line~\ref{line:crash-three-delays-req-send}, hence after the invocation began, so \Cref{claim:crash-three-delays-clock-monotone} implies that its proposal value~$w_k$ satisfies $w_k\geq h$.
  As $k\in P_j$, we have $\langle k,w_k\rangle \in \texttt{props}_j[i]$, so $v_{\max}\geq w_k\geq h$ when $p_j$ executes line~\ref{line:crash-three-delays-clock-update}.
  After that line, we have $\xx_j\geq v_{\max}+1>h$, and $p_j$ sends this value of~$\xx_j$ as~$v_j$ on line~\ref{line:crash-three-delays-reply-send}, so $v_j>h$.

  Since $j\in R$ was arbitrary, line~\ref{line:crash-three-delays-return-compute} yields $t=\min\{v_j : j\in R\}>h$.
\end{proof}

\begin{claim}[Ordering for base timestamps]\label{claim:crash-three-delays:ordering}
  Suppose two \nextOpParens{} operations compute base timestamps~$t$ and~$t'$ on line~\ref{line:crash-three-delays-return-compute}, respectively.
  If the first returns before the second is invoked, then $t<t'$.
\end{claim}
\begin{proof}
  \proofqed{claim:crash-three-delays:ordering}
  By \Cref{claim:crash-three-delays-return-support}, when the first operation returns, at least $n-f$ processes~$p_j$ satisfy $\xx_j\geq t$.
  By \Cref{claim:crash-three-delays-clock-monotone}, these inequalities still hold when the second operation is invoked.
  Applying \Cref{claim:crash-three-delays-later-return} with $h=t$ therefore shows that the second operation has base timestamp $t'>t$.
\end{proof}

\begin{claim}[Return values]\label{claim:crash-three-delays:return-values}
  \Cref{alg:crash-three-delays} returns positive integers from \nextOpParens{} and satisfies ordering~(\ref{cond:timestamp-ordering}) and uniqueness~(\ref{cond:timestamp-uniqueness}).
\end{claim}
\begin{proof}
  \proofqed{claim:crash-three-delays:return-values}
  Local clocks~$\xx_i$ start at~$0$ and never decrease (\Cref{claim:crash-three-delays-clock-monotone}), so by line~\ref{line:crash-three-delays-clock-update}, every value sent in a \Reply{} message is a positive integer.
  Hence, every base timestamp~$t$ is a positive integer, and so is $n\cdot t+i$.

  For~(\ref{cond:timestamp-ordering}), suppose a complete operation at~$p_i$ with base timestamp~$t$ precedes one at~$p_j$ with base timestamp~$t'$.
  By \Cref{claim:crash-three-delays:ordering}, $t<t'$, so $n\cdot t+i<n\cdot(t+1)\leq n\cdot t'+j$.

  For~(\ref{cond:timestamp-uniqueness}), timestamps returned by different processes differ modulo~$n$, and those returned by the same process are distinct by~(\ref{cond:timestamp-ordering}) since its operations are sequential.
\end{proof}

\begin{claim}[Liveness]\label{claim:crash-three-delays:liveness}
  In \Cref{alg:crash-three-delays}, every invocation by a correct process eventually returns.
\end{claim}
\begin{proof}
  \proofqed{claim:crash-three-delays:liveness}
  Let $p_i$ be a correct process that invokes~\nextOpParens{}, and let $q$ be the value of~$\texttt{seq}_i$ during this invocation.
  Line~\ref{line:crash-three-delays-req-send} sends $\langle \Req~q\rangle$ to every process, at least $n-f$ of which are correct.
  Each correct recipient eventually receives this fresh request, passes the check on line~\ref{line:crash-three-delays-fresh-request}, and sends $\langle \Prop~\langle i,q\rangle,-\rangle$ to every process on line~\ref{line:crash-three-delays-prop-send}.
  Since $p_i$ sends no request with a higher sequence number before this invocation returns, each correct process~$p_j$ eventually receives these \Prop{} messages from at least $n-f$ distinct processes while $\texttt{propSeq}_j[i]=q$.
  Thus line~\ref{line:crash-three-delays-props-add} eventually makes $|\texttt{props}_j[i]|=n-f$, causing $p_j$ to pass the check on line~\ref{line:crash-three-delays-props-full} and send a \Reply{} to~$p_i$ on line~\ref{line:crash-three-delays-reply-send}, which $p_i$ eventually receives.
  Hence, line~\ref{line:crash-three-delays-record-reply} eventually makes $|\texttt{replies}_i|=n-f$, causing $p_i$ to pass the check on line~\ref{line:crash-three-delays-replies-full} and return on line~\ref{line:crash-three-delays-return}.
\end{proof}

\begin{claim}[Latency]\label{claim:crash-three-delays:latency}
  The worst-case message delay of \Cref{alg:crash-three-delays} is three.
\end{claim}
\begin{proof}[Proof sketch]
  \proofqed{claim:crash-three-delays:latency}
  An invocation follows one \Req{} delay, one \Prop{} delay, and one \Reply{} delay, as illustrated below:
  \[
    \begin{tikzpicture}[
      baseline=(current bounding box.south),
      event/.style={circle, fill, inner sep=2pt},
      msg/.style={-{Stealth[length=2mm]}, thin},
      every node/.style={font=\small}
    ]
      \node[event, label={[align=right, xshift=-1em, yshift=-2pt]below:{invoke\\\nextOpParens{}}}] (invoke) at (0,0) {};
      \node[event] (req1) at (2.4,0.7) {};
      \node[event] (req2) at (2.4,0) {};
      \node[event] (req3) at (2.4,-0.7) {};
      \node[event] (prop1) at (5.0,0.7) {};
      \node[event] (prop2) at (5.0,0) {};
      \node[event] (prop3) at (5.0,-0.7) {};
      \node[event, label={[align=left, xshift=1em, yshift=-2pt]below:{return}}] (return) at (7.4,0) {};

      \node[fill=white, inner sep=1pt] at (1.2,1.1) {\Req};
      \node[fill=white, inner sep=1pt] at (3.7,1.1) {\Prop};
      \node[fill=white, inner sep=1pt] at (6.2,1.1) {\Reply};

      \foreach \a in {req1,req2,req3} {
        \draw[msg] (invoke) -- (\a);
        \foreach \b in {prop1,prop2,prop3} {
          \draw[msg] (\a) -- (\b);
        }
      }
      \foreach \b in {prop1,prop2,prop3} {
        \draw[msg] (\b) -- (return);
      }
    \end{tikzpicture}\qedhere
  \]
\end{proof}

By \Cref{claim:crash-three-delays:return-values,claim:crash-three-delays:liveness}, \Cref{alg:crash-three-delays} satisfies the conditions of \Cref{thm:timestamp-characterization} and hence implements~\ts{}.
By \Cref{claim:crash-three-delays:latency}, its worst-case message delay is three.
\end{proof}

\subsection{Latency Lower Bounds}\label{sec:crash:lower-bounds}

We now turn to lower bounds on the worst-case message delay of implementing~\ts{} under crash failures.

\subsubsection{Auxiliary Lemmas}

\begin{definition}[Loquacious algorithm]\label{def:loquacious}
  An algorithm is \emph{loquacious} if every step sends the same message to every process in the system (including the sender itself).
  We say that such a step \emph{broadcasts} that message.
\end{definition}

The lower-bound proofs will use two auxiliary lemmas.
The first provides a transformation from arbitrary implementations of~\ts{} to loquacious ones.

\begin{restatable}{lemma}{loquaciousTransformation}\label{thm:loquacious-transformation}
  In a system with crash failures, let~$\alg$ be an implementation of~\ts{}.
  Then there exists a loquacious implementation~$\alg'$ of~\ts{} such that $\latency{\alg'}\leq\latency{\alg}$.%
\end{restatable}

Since this transformation does not increase worst-case message delay, any latency lower bound for loquacious implementations of~\ts{} also holds for arbitrary implementations of~\ts{} in the same system.
We defer the proof of this lemma to \Cref{sec:loquacious-transformation-proofs:crash}.

The second auxiliary lemma concerns completion in two-delay algorithms.
Intuitively, an algorithm with worst-case message delay at most two must return from an operation once its invoker completes a round trip with $n-f$ processes.
The invoker cannot wait for the remaining $f$~processes, which may have crashed, or for another communication step, which would require a third message delay.
The following lemma makes this intuition precise.

\begin{restatable}{lemma}{twoDelayCompletion}\label{thm:two-delay-completion}
  In a system with crash failures, let $\alg$~be a loquacious algorithm with $\latency{\alg} \leq 2$.
  Let $p$~be a process, and let $Q$ be a set of processes such that $|Q|=n-f$ and $p\in Q$.

  Suppose that $\sigma$ is a schedule prefix of~$\alg$ in which the following hold:
  \begin{enumerate}
    \item Process~$p$ invokes an operation and broadcasts a message~$m$ in step~$\sigma[a]$.
    \item For each $q \in Q$, process~$q$ receives message~$m$ and broadcasts a message~$m_q$ in step~$\sigma[j_q]$.\label{item:two-delay-completion-first-wave}
    \item Process~$p$ receives the set of messages $\{m_q : q\in Q\}$ in step~$\sigma[b]$.\label{item:two-delay-completion-second-wave}
    \item Between steps $\sigma[a]$ and $\sigma[b]$, inclusive, the only steps are those specified above.\label{item:two-delay-completion-no-other-steps}
  \end{enumerate}
  Then the operation invoked by~$p$ in step~$\sigma[a]$ returns no later than step~$\sigma[b]$.
\end{restatable}

We defer the proof of this lemma to \Cref{sec:completion-lemma-proofs:crash}.

\subsubsection{The Bounds}

The first lower bound is intuitively clear: A process invoking~\nextOpParens{} cannot safely return before completing at least one round trip with another process, which requires at least two message delays.

\begin{restatable}{proposition}{crashTwoDelayLowerBound}\label{thm:crash-two-delay-lower-bound}
  In a system with crash failures, any implementation of~\ts{} has worst-case message delay at least two.
\end{restatable}

Making this intuition precise requires some technical care; we defer the proof to \Cref{sec:crash-lower-bound-proofs}.

We now show that implementing~\ts{} requires at least \emph{three} message delays when two or more processes may crash.
Our proof uses a novel \emph{cyclic chain argument}: a chain of indistinguishable schedules that closes on itself, forcing the returned timestamps to increase strictly around a cycle.

\begin{theorem}\label{thm:crash-lower-bound-fgeq2}
  In a system with crash failures, if $f\geq 2$, then any implementation of~\ts{} has worst-case message delay at least three.
\end{theorem}
\definecolor{ctdproc0}{HTML}{0072B2}
\definecolor{ctdproc1}{HTML}{E69F00}
\definecolor{ctdproc2}{HTML}{D55E00}
\definecolor{ctdproc3}{HTML}{009E73}
\definecolor{ctdproc4}{HTML}{CC79A7}
\definecolor{ctdproc5}{HTML}{56B4E9}

\begin{figure}
  \centering
  \input{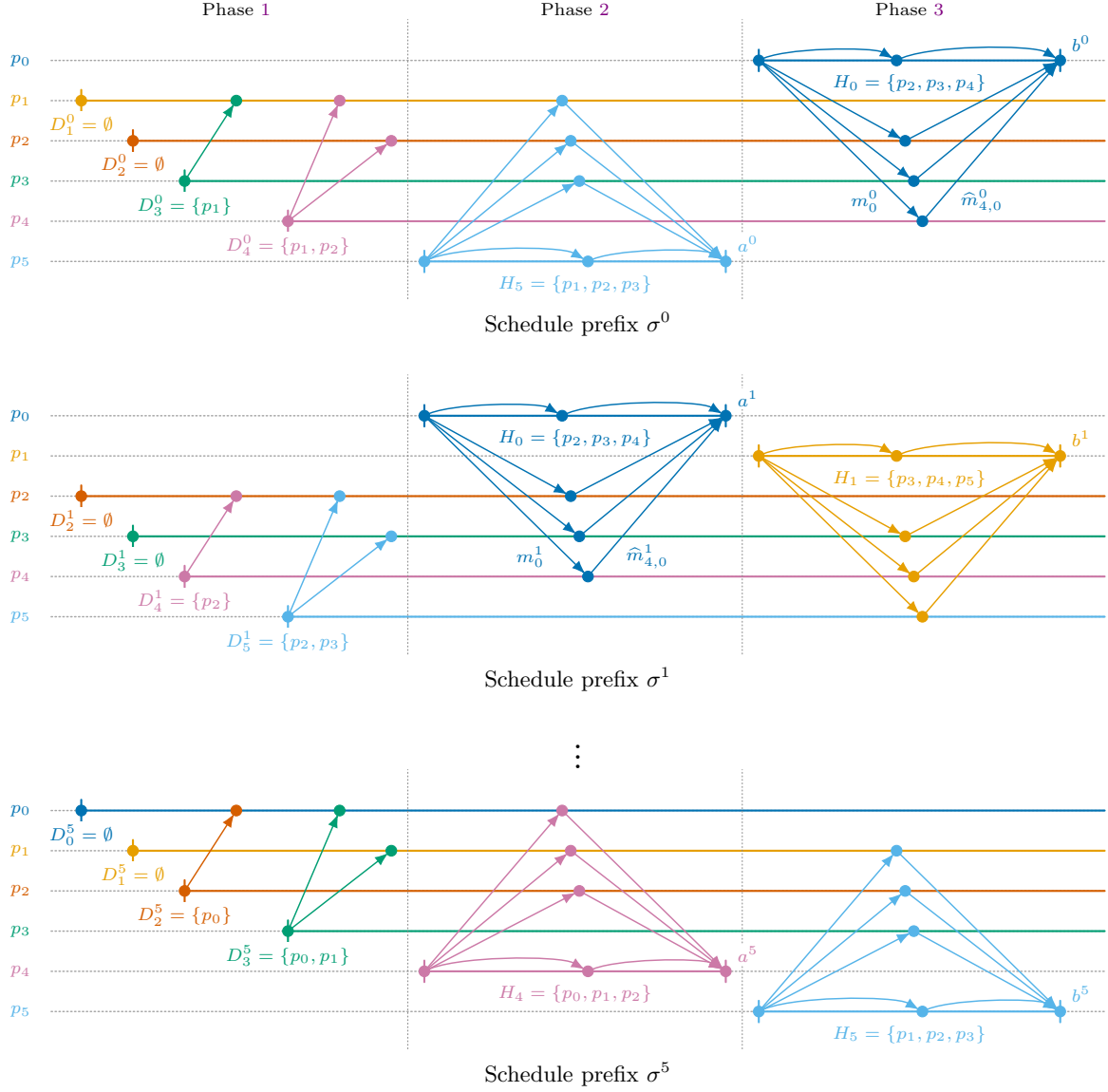}
  \caption{\textbf{Proof of \Cref{thm:crash-lower-bound-fgeq2}.}
  In a system with crash failures where $f\geq 2$, suppose toward a contradiction that some loquacious implementation~$\alg$ of~\ts{} has $\latency{\alg}\leq2$.
  For $n=6$ and $f=2$, the figure shows $\sigma^0$, $\sigma^1$, and~$\sigma^5$.
  Each dot marks a step and is colored by the invocation it serves, while each arrow shows a message received within the prefix.
  A horizontal segment tracks a \nextOpParens{} operation: \opinvocationmark{} marks its invocation step, and \opreturnmark{} labeled~$r$ indicates that it has returned~$r$ by that step.
  Without a return marker, the figure makes no claim about whether the operation returns within the prefix.
  Condition~(\ref{cond:timestamp-ordering}) gives $a^i<b^i$.
  Furthermore, $\sigma^i$ and $\sigma^{i+1}$ are indistinguishable to~$p_i$ through the return of its operation, so $b^i=a^{i+1}$.
  Together, these relations yield the contradiction $a^0<a^1<\cdots<a^5<a^0$.
  }
  \label{fig:crash-three-delay-schedules}
\end{figure}

\begin{proof}
  \proofqed{thm:crash-lower-bound-fgeq2}
  Consider a system with $n$~processes, of which at most $f\geq2$ may crash.
  Suppose, toward a contradiction, that there is an implementation of~\ts{} whose worst-case message delay is less than three.
  By \Cref{thm:loquacious-transformation}, there is a loquacious implementation~$\alg$ of~\ts{} with $\latency{\alg}<3$.
  Since finite worst-case message delays are integers (\Cref{thm:integer-worst-case-delay}), $\latency{\alg}\leq2$.
  By \Cref{thm:timestamp-characterization}, $\alg$ satisfies ordering~(\ref{cond:timestamp-ordering}).
  Finally, by \Cref{thm:crash-resilience}, $n>2f$.

  We name the processes $p_0,p_1,\ldots,p_{n-1}$, with process indices taken from the cyclic group~$\mathbb{Z}_n$.
  So arithmetic on process indices is modulo~$n$.
  For $i, j\in\mathbb{Z}_n$, we write $i,i+1,\ldots,j$ for the inclusive cyclic range obtained by starting at~$i$ and repeatedly adding one until $j$~is first reached (in particular, when $i=j$, this range is simply~$i$).

  \paragraph{Roadmap.}
  For each $i\in\mathbb{Z}_n$, we will construct a schedule prefix~$\sigma^i$ of~$\alg$ as follows (\Cref{fig:crash-three-delay-schedules} illustrates the construction for $n=6$, $f=2$):%
  \begin{itemize}
    \item For $j=i+1,i+2,\ldots,i-2$,
      process~$p_j$ invokes~\nextOpParens{}; the message broadcast by its invocation step is then received by the processes in some \emph{receiver set}~$\receiverSet{j}{i}$.
    \item Process~$p_{i-1}$ invokes~\nextOpParens{}, completes a round trip of communication with itself and some set~$H_{i-1}$ of \emph{helper} processes, and returns value~$a^i$.
    \item Process~$p_i$ invokes~\nextOpParens{}, completes a round trip with itself and the helper processes in~$H_i$, and returns value~$b^i$.
  \end{itemize}

  Each~$\sigma^i$ can be extended to a schedule of~$\alg$ in which $p_{i-1}$ returns before $p_i$ invokes, so condition~(\ref{cond:timestamp-ordering}) implies
  \[
    a^i<b^i \qquad \text{for each $i\in\mathbb{Z}_n$}.
  \]

  The key to the argument will be to choose the sets $\receiverSet{j}{i}$ and~$H_i$ such that $\sigma^i$ and~$\sigma^{i+1}$ are indistinguishable to process~$p_i$ for each $i\in\mathbb{Z}_n$.
  Consequently, $p_i$ returns the same value from its operation in~$\sigma^i$ and~$\sigma^{i+1}$, and so
  \[
    a^{i+1} = b^i \qquad \text{for each $i\in\mathbb{Z}_n$}.
  \]
  Thus $a^i<a^{i+1}$ for every $i\in\mathbb{Z}_n$.
  Chaining these inequalities around the cycle yields $a^0<a^1<\cdots<a^{n-1}<a^0$, a contradiction.

  We now give the full proof following this roadmap, starting by defining the sets~$H_i$ and~$\receiverSet{j}{i}$.

  \paragraph{Helper and receiver sets.}
  For each $i\in\mathbb{Z}_n$, define the helper set
  \[
    H_i \Deq \{p_{i+2},p_{i+3},\ldots,p_{i-f}\}.
  \]
  By construction, $|H_i|=n-f-1$.
  Moreover, since $n>2f$ and $f\geq2$,
  \begin{equation}\label{eq:crash-lb-h-excludes}
    p_{i-1},p_i,p_{i+1}\notin H_i.
  \end{equation}

  For $i\in\mathbb{Z}_n$ and $j\in\{i+1,i+2,\ldots,i-2\}$, define the receiver set
  \begin{equation}\label{eq:lb-receiver-def}
    \receiverSet{j}{i}
    \Deq
    H_j\cap H_{j+1}\cap\cdots\cap H_{i-1}.
  \end{equation}
  For each fixed $i\in\mathbb{Z}_n$, these sets are nested as their indices advance toward~$i-1$ in cyclic order:
  \begin{equation}\label{eq:lb-receiver-nested}
    \receiverSet{i+1}{i}
    \subseteq \receiverSet{i+2}{i}
    \subseteq \cdots
    \subseteq \receiverSet{i-2}{i}
    \subseteq H_{i-1}.
  \end{equation}
  Unwinding the definitions, for every $p_k\in H_{i-1}$ and
  $x\in\{i+1,i+2,\ldots,i-2\}$,
  \begin{equation}\label{eq:lb-receiver-inv}
    p_k\in\receiverSet{x}{i}
    \iff
    x\in\{k+f,k+f+1,\ldots,i-1\}.
  \end{equation}

  \paragraph{Constructing the schedule prefixes.}
  Fix $i\in\mathbb{Z}_n$.
  We construct~$\sigma^i$ in three phases, listing the steps taken in each phase.
  For each step, we specify the process that takes it, the operation it invokes, if any, and the messages it receives;
  the algorithm~$\alg$ then determines the message the step broadcasts, whether the step returns from an operation, and, if so, the return value.

  Any message whose receipt we do not specify is not received during~$\sigma^i$.
  Since $\alg$ is loquacious, every step broadcasts a message, but we name only those messages relevant to the proof.
  If we do not state whether a step returns from an operation, then it is irrelevant to the proof whether it does.

  The three phases of~$\sigma^i$ are as follows.
  \begin{enumerate}[ref=\arabic*, label={\textit{Phase \arabic*:}}, leftmargin=*]
    \item\label{item:lb-phase-s1}
      \emph{background invocations.} For $j=i+1,i+2,\ldots,i-2$, in that cyclic order:
      \begin{enumerate}
        \item\label{item:lb-phase-s1-invoke} process~$p_j$ takes a step that receives no messages and invokes~\nextOpParens{}, broadcasting a message~$m_j^i$; and then
        \item\label{item:lb-phase-s1-receive} the processes in~$\receiverSet{j}{i}$ (in any order) each take a step that receives~$m_j^i$.
      \end{enumerate}
    \item\label{item:lb-phase-s2}
      \emph{a complete operation by~$p_{i-1}$.}
      \begin{enumerate}
        \item\label{item:lb-phase-s2-invoke} process~$p_{i-1}$ takes a step that receives no messages and invokes~\nextOpParens{}, broadcasting a message~$m_{i-1}^i$;
        \item\label{item:lb-phase-s2-respond} each process $p_k\in H_{i-1}\cup\{p_{i-1}\}$ (in any order) takes a step that receives~$m_{i-1}^i$, broadcasting a message~$\widehat m_{k,i-1}^i$; and then
        \item\label{item:lb-phase-s2-aggregate} process~$p_{i-1}$ takes a step that receives the set of messages $\{ \widehat m_{k,i-1}^i : p_k\in H_{i-1}\cup\{p_{i-1}\} \}$ broadcast above.
      \end{enumerate}
      By \Cref{thm:two-delay-completion} (with $p=p_{i-1}$ and $Q=H_{i-1}\cup\{p_{i-1}\}$), the operation invoked by~$p_{i-1}$ in step~\ref{item:lb-phase-s2-invoke} has returned by this point; denote its return value by~$a^i$.
    \item\label{item:lb-phase-s3}
      \emph{a complete operation by~$p_i$,} repeating phase~\ref{item:lb-phase-s2} with $p_i$ in place of~$p_{i-1}$ and $H_i$ in place of~$H_{i-1}$:
      \begin{enumerate}
        \item\label{item:lb-phase-s3-invoke} process~$p_i$ takes a step that receives no messages and invokes~\nextOpParens{}, broadcasting a message~$m_i^i$;
        \item\label{item:lb-phase-s3-respond} each process $p_k\in H_i\cup\{p_i\}$ (in any order) takes a step that receives~$m_i^i$, broadcasting a message~$\widehat m_{k,i}^i$; and then
        \item\label{item:lb-phase-s3-aggregate} process~$p_i$ takes a step that receives the set of messages $\{ \widehat m_{k,i}^i : p_k\in H_i\cup\{p_i\} \}$ broadcast above.
      \end{enumerate}
      By \Cref{thm:two-delay-completion} (with $p=p_i$ and $Q=H_i\cup\{p_i\}$), the operation invoked by~$p_i$ in step~\ref{item:lb-phase-s3-invoke} has returned by this point; denote its return value by~$b^i$.
  \end{enumerate}
  By construction, each process invokes only one operation, receives only messages broadcast in an earlier step, and receives each message at most once.
  Thus, $\sigma^i$ is a schedule prefix of~$\alg$.
  We call the steps specified in \ref{item:lb-phase-s1-invoke}, \ref{item:lb-phase-s2-invoke}, and~\ref{item:lb-phase-s3-invoke} \emph{invocation steps}, and all other steps of~$\sigma^i$ \emph{receive steps}.

  \paragraph{Return-value order within each prefix.}
  As previewed in the roadmap, we first use condition~(\ref{cond:timestamp-ordering}) to order the return values of the two complete operations in each prefix.
  \begin{claim}\label{claim:crash-lb-a-lt-b}
    For every $i\in\mathbb{Z}_n$, $a^i<b^i$.
  \end{claim}
  \begin{proof}
    \proofqed{claim:crash-lb-a-lt-b}
    Fix $i\in\mathbb{Z}_n$.
    Extend~$\sigma^i$ to a schedule of~$\alg$.
    Both operations are complete and $p_{i-1}$'s operation returns before $p_i$'s operation is invoked.
    The claim follows from condition~(\ref{cond:timestamp-ordering}).
  \end{proof}

  \paragraph{Return-value equality across consecutive prefixes.}
  Next, we establish several auxiliary claims about the steps taken by processes in the schedule prefixes.
  We then use these claims to show that two complete operations by the same process in consecutive prefixes return the same value.

  \begin{claim}\label{claim:crash-lb-initial-invocations}
    For every $i,j\in\mathbb{Z}_n$, process~$p_j$'s invocation step is its first step in~$\sigma^i$.
    Consequently,
    \[
      m_j^i=m_j^{i'}
      \qquad\text{for all $i,i',j\in\mathbb{Z}_n$}.
    \]
  \end{claim}
  \begin{proof}
    \proofqed{claim:crash-lb-initial-invocations}
    Fix $i,j\in\mathbb{Z}_n$.
    In~$\sigma^i$, the invocation steps occur in the cyclic order
    \[
      p_{i+1},\,p_{i+2},\,\ldots,\,p_{i-2},\,p_{i-1},\,p_i.
    \]

    Suppose first that $j\neq i$.
    Before $p_j$~invokes, every step in~$\sigma^i$ is either an invocation step by a process~$p_x$ preceding~$p_j$ in the order above, or a receive step taken by a process in~$\receiverSet{x}{i}$ for such an~$x$.
    Since $p_x$ precedes~$p_j$ in the invocation order, \eqref{eq:lb-receiver-def} gives $\receiverSet{x}{i}\subseteq H_j$.
    Since $p_j\notin H_j$ by~\eqref{eq:crash-lb-h-excludes}, none of these steps is taken by~$p_j$.

    Now suppose that $j=i$.
    Every invocation step preceding $p_i$'s invocation is taken by some~$p_x$ with $x\neq i$.
    Every preceding receive step is taken either by a process in $\receiverSet{x}{i}\subseteq H_{i-1}$ during phase~\ref{item:lb-phase-s1}, where the inclusion follows from~\eqref{eq:lb-receiver-nested}, or by a process in $H_{i-1}\cup\{p_{i-1}\}$ during phase~\ref{item:lb-phase-s2}.
    Since $p_i\notin H_{i-1}$ by~\eqref{eq:crash-lb-h-excludes} and $p_i\neq p_{i-1}$, none of these steps is taken by~$p_i$.

    Thus, $p_j$'s invocation step is taken from its initial state in every prefix.
    These steps invoke the same operation~\nextOpParens{} and receive no messages, so they must broadcast the same message.
    Therefore, $m_j^i=m_j^{i'}$ for all $i,i',j\in\mathbb{Z}_n$.
  \end{proof}

  For each $j\in\mathbb{Z}_n$, denote by~$m_j$ the message~$m_j^i$, which is the same for every $i\in\mathbb{Z}_n$.

  \begin{claim}\label{claim:crash-lb-helper-history-current}
    For every $i\in\mathbb{Z}_n$ and $p_k\in H_i$, in~$\sigma^i$ the steps taken by~$p_k$ through its receipt of~$m_i$ consist of $p_k$'s invocation step followed, in order, by one receive step for each message in the sequence
    \begin{equation}\label{eq:crash-lb-helper-message-sequence}
      m_{k+f},\,m_{k+f+1},\,\ldots,\,m_i.
    \end{equation}
  \end{claim}
  \begin{proof}
    \proofqed{claim:crash-lb-helper-history-current}
    Fix $i\in\mathbb{Z}_n$ and $p_k\in H_i$, and consider the prefix~$\sigma^i$.
    By~\eqref{eq:crash-lb-h-excludes}, $p_k\notin\{p_{i-1},p_i\}$, so $p_k$~invokes during phase~\ref{item:lb-phase-s1}.
    We distinguish two cases according to whether $p_k\in H_{i-1}$:
    \begin{itemize}
      \item If $p_k\in H_{i-1}$, then~\eqref{eq:lb-receiver-inv} and the construction of phase~\ref{item:lb-phase-s2} imply that, after invoking, $p_k$ receives $m_{k+f},m_{k+f+1},\ldots,m_{i-1}$ in that order during phases \ref{item:lb-phase-s1} and~\ref{item:lb-phase-s2}.
        It then receives~$m_i$ during phase~\ref{item:lb-phase-s3}, completing the sequence in~\eqref{eq:crash-lb-helper-message-sequence}.
      \item If $p_k\notin H_{i-1}$, then $p_k\in H_i\setminus H_{i-1}$, which implies $k+f=i$.
        Thus the sequence in~\eqref{eq:crash-lb-helper-message-sequence} consists only of~$m_i$.
        By~\eqref{eq:lb-receiver-nested}, no receiver set used in phase~\ref{item:lb-phase-s1} contains~$p_k$, so it receives no messages in that phase.
        Moreover, $p_k\notin H_{i-1}\cup\{p_{i-1}\}$, so it does not receive~$m_{i-1}$ in phase~\ref{item:lb-phase-s2} either.
        Hence $p_k$ takes no receive step before receiving~$m_i$ in phase~\ref{item:lb-phase-s3}, as claimed.\qedhere
    \end{itemize}
  \end{proof}

  \begin{claim}\label{claim:crash-lb-helper-history-next}
    For every $i\in\mathbb{Z}_n$ and $p_k\in H_i$, in~$\sigma^{i+1}$ the steps taken by~$p_k$ through its receipt of~$m_i$ consist of $p_k$'s invocation step followed, in order, by one receive step for each message in the sequence~\eqref{eq:crash-lb-helper-message-sequence}.
  \end{claim}
  \begin{proof}
    \proofqed{claim:crash-lb-helper-history-next}
    Fix $i\in\mathbb{Z}_n$ and $p_k\in H_i$, and consider the prefix~$\sigma^{i+1}$.
    By~\eqref{eq:crash-lb-h-excludes}, $p_k\notin\{p_i,p_{i+1}\}$, so $p_k$ takes its invocation step during phase~\ref{item:lb-phase-s1}.
    In~$\sigma^{i+1}$, \eqref{eq:lb-receiver-inv} implies that during phase~\ref{item:lb-phase-s1} process~$p_k$ receives, in order, exactly those messages in~\eqref{eq:crash-lb-helper-message-sequence} that precede~$m_i$, each in a separate receive step.
    During phase~\ref{item:lb-phase-s2}, it receives~$m_i$ in one more receive step because $p_k\in H_i$.
    These are precisely the steps stated in the claim.
  \end{proof}

  \begin{claim}\label{claim:crash-lb-response-messages}
    For every $i\in\mathbb{Z}_n$ and $p_k\in H_i\cup\{p_i\}$,
    \[
      \widehat m_{k,i}^i=\widehat m_{k,i}^{i+1}.
    \]
  \end{claim}
  \begin{proof}
    \proofqed{claim:crash-lb-response-messages}
    Fix $i\in\mathbb{Z}_n$ and $p_k\in H_i\cup\{p_i\}$.
    By~\eqref{eq:crash-lb-h-excludes}, $p_i\notin H_i$, so exactly one of the following cases holds:
    \begin{itemize}
      \item If $p_k\in H_i$, then by \Cref{claim:crash-lb-helper-history-current,claim:crash-lb-helper-history-next}, $p_k$~takes the same sequence of steps in~$\sigma^i$ and~$\sigma^{i+1}$ through its receive step for~$m_i$.
        This step broadcasts $\widehat m_{k,i}^i$ in~$\sigma^i$ and~$\widehat m_{k,i}^{i+1}$ in~$\sigma^{i+1}$, so $\widehat m_{k,i}^i=\widehat m_{k,i}^{i+1}$.
      \item If $p_k=p_i$, then \Cref{claim:crash-lb-initial-invocations} implies that its first step in both $\sigma^i$ and~$\sigma^{i+1}$ is its invocation step, which broadcasts~$m_i$ in each.
        By construction, its next step in each prefix is a receive step for~$m_i$.
        These receive steps are therefore identical, so $\widehat m_{i,i}^i=\widehat m_{i,i}^{i+1}$.\qedhere
    \end{itemize}
  \end{proof}

  For each $i\in\mathbb{Z}_n$ and $p_k\in H_i\cup\{p_i\}$, we denote by~$\widehat m_{k,i}$ the message~$\widehat m_{k,i}^i=\widehat m_{k,i}^{i+1}$.

  \begin{claim}\label{claim:crash-lb-b-eq-a}
    For every $i\in\mathbb{Z}_n$, $b^i=a^{i+1}$.
  \end{claim}
  \begin{proof}
    \proofqed{claim:crash-lb-b-eq-a}
    Fix $i\in\mathbb{Z}_n$, and compare the steps of~$p_i$ in $\sigma^i$ and~$\sigma^{i+1}$, which belong to phase~\ref{item:lb-phase-s3} of the former and phase~\ref{item:lb-phase-s2} of the latter.
    In both prefixes, $p_i$~takes exactly three steps up to and including the last step of its round trip:
    \begin{enumerate}
      \item\label{item:lb-pi-invoke} $p_i$ invokes~\nextOpParens{} and broadcasts~$m_i$;
      \item\label{item:lb-pi-self} $p_i$ receives~$m_i$ and broadcasts~$\widehat m_{i,i}$; and
      \item\label{item:lb-pi-aggregate} $p_i$ receives all messages~$\widehat m_{k,i}$ for $p_k\in H_i\cup\{p_i\}$.
    \end{enumerate}
    By \Cref{claim:crash-lb-initial-invocations}, $p_i$'s invocation is its first step in both prefixes and broadcasts the same message~$m_i$.
    By construction, its next step in each prefix receives~$m_i$, so steps \ref{item:lb-pi-invoke} and~\ref{item:lb-pi-self} agree in the two prefixes.
    By \Cref{claim:crash-lb-response-messages}, $p_i$ receives the same set of messages in step~\ref{item:lb-pi-aggregate}, so this step agrees as well.
    By \Cref{thm:two-delay-completion}, $p_i$'s operation returns by this step in each prefix.
    It therefore returns at the same step and with the same value in both prefixes.
    This value is~$b^i$ in~$\sigma^i$ and~$a^{i+1}$ in~$\sigma^{i+1}$, so $b^i=a^{i+1}$.
  \end{proof}

  \paragraph{Contradiction.}
  By \Cref{claim:crash-lb-a-lt-b,claim:crash-lb-b-eq-a}, $a^i<b^i=a^{i+1}$ for every $i\in\mathbb{Z}_n$.
  Chaining these inequalities around the cycle gives
  \[
    a^0<a^1<\cdots<a^{n-1}<a^0,
  \]
  a contradiction.
  Thus every implementation of~\ts{} has worst-case message delay at least three.
\end{proof}

\subsection{Summary}
\Cref{thm:crash-resilience} and \Cref{thm:crash-upper-bound-fgeq2} together give the following necessary and sufficient condition for implementing~\ts{} under crash failures.

\begin{corollary}
  In a system with crash failures, there exists an implementation of~\ts{} if and only if $n>2f$.
\end{corollary}

We next state the optimal latency for implementations of~\ts{} under crash failures.
When $n>2f$ and $f\geq 2$, the latency upper bound in \Cref{thm:crash-upper-bound-fgeq2} matches the lower bound in \Cref{thm:crash-lower-bound-fgeq2}, yielding the following.

\begin{corollary}
  In a system with crash failures, if $n>2f$ and $f\geq 2$, then the optimal latency of implementing~\ts{} is three message delays.
\end{corollary}

When $n>2f$ and $f=1$, the latency upper bound in \Cref{thm:crash-upper-bound-f1} matches the lower bound in \Cref{thm:crash-two-delay-lower-bound}, yielding the following.

\begin{corollary}
  In a system with crash failures, if $n>2f$ and $f=1$, then the optimal latency of implementing~\ts{} is two message delays.
\end{corollary}

Thus, the optimal latency of implementing~\ts{} under crash failures is settled in all cases where an implementation exists.

\section{Adapting to Byzantine Failures}\label{sec:byz-model}
We now adapt the system model and latency measure (\Cref{sec:model-latency}) and the implementation requirements for~\ts{} (\Cref{sec:timestamp}) to systems with Byzantine processes.
These adaptations apply in both settings we consider: with digital signatures (\Cref{sec:byz-sig}) and without them (\Cref{sec:byz}).

\subsection{System Model and Latency}\label{sec:byz-model:system-latency}

We continue to assume $n\geq 2$ processes and $1\leq f<n$.
The definitions of processes, the message-passing model, and algorithms are unchanged; only the definition of a schedule changes.
In a system with Byzantine processes, a schedule~$S$ of algorithm~$\alg$ is an infinite sequence of steps in which at least $n-f$ processes take infinitely many steps, all according to~$\alg$; these processes are \emph{correct} in~$S$.
Every other process is \emph{faulty}---it may crash, deviate arbitrarily from~$\alg$, or both.
The message-delivery requirements on a schedule are unchanged.

Under these revised notions of schedules and correct processes, the remaining definitions in \Cref{sec:model} and the latency definitions in \Cref{sec:delay} carry over verbatim.
The two consequences of finite worst-case message delay, \Cref{thm:liveness-from-finite-delay,thm:integer-worst-case-delay}, also carry over.
(Their proofs in \Cref{sec:finite-delay} apply equally under crash and Byzantine failures.)

In the setting with digital signatures (\Cref{sec:byz-sig}), no process can forge the signature of a correct process, and any process can relay a signed message and verify its signature.

\begin{remark}
  The eventual-delivery requirement in \Cref{sec:model} and the definition of message timeliness in \Cref{def:message-timeliness} cover messages sent to any process that ``takes infinitely many steps,'' even if that process is Byzantine.
  This is natural because both constrain network behavior, not process behavior: The network does not deliver a message any less reliably or promptly merely because its recipient deviates from~$\alg$.
  In particular, if a Byzantine process~$q$ takes infinitely many steps, every message~$m$ sent to~$q$ is eventually delivered.
  However, $q$~may handle the delivered message arbitrarily: It may act as if it had not received~$m$ or as if it had received a different message~$m'\neq m$.
\end{remark}

\subsection{Implementations of the Timestamp Object}\label{sec:byz-model:timestamp}

To adapt \Cref{def:timestamp-implementation} (implementation of~\ts{}) to Byzantine failures, we replace crash linearizability with Byzantine linearizability~\cite[Section~4]{DBLP:conf/wdag/CohenK21}.
Intuitively, the history~$H$ of a schedule~$S$ is \emph{Byzantine linearizable} with respect to~\ts{} if there exists a history~$\widehat H$ such that
(1)~the projection of~$\widehat H$ onto the correct processes in~$S$ is the same as that of~$H$, and
(2)~$\widehat H$ is crash linearizable with respect to~\ts{}.
In other words, $H$ can be made crash linearizable with respect to~\ts{} by adding, omitting, or modifying operations of processes that are faulty in~$S$.
Below, we underline the \uline{changes for Byzantine failures}.

\begin{definition}[\ts{} implementation under Byzantine failures]\label{def:byz-timestamp-implementation}
  Let $\alg$ be an algorithm whose \nextOpParens{} operation returns positive integers.
  We say that $\alg$ \emph{implements}~\ts{} in a system \uline{with Byzantine processes} if
  every schedule~$S$ of~$\alg$, with corresponding history~$H$,
  satisfies the following conditions:
  \begin{enumerate}[label=(BT\arabic*),ref=BT\arabic*,leftmargin=*]
    \item\label{cond:byz-timestamp-safety} \emph{Safety:} $H$ is \uline{Byzantine linearizable} with respect to~\ts{}.
    \item\label{cond:byz-timestamp-liveness} \emph{Liveness:} Every operation invoked by a correct process in~$S$ eventually returns.
  \end{enumerate}
\end{definition}

The following analogue of \Cref{thm:timestamp-characterization} characterizes \ts{} implementations under Byzantine failures in terms of ordering and uniqueness for operations invoked by correct processes.

\begin{proposition}[Characterization of \ts{} implementations under Byzantine failures]\label{thm:byz-timestamp-characterization}
  Let $\alg$ be an algorithm whose \nextOpParens{} operation returns positive integers.
  Then $\alg$ implements~\ts{} in a system \uline{with Byzantine processes} if and only if
  every schedule~$S$ of~$\alg$, with corresponding history~$H$,
  satisfies the following conditions:
  \begin{enumerate}[label=(\getrefnumber{cond:byz-timestamp-safety}\alph*),ref=\getrefnumber{cond:byz-timestamp-safety}\alph*,leftmargin=*]
    \item\label{cond:byz-timestamp-ordering} \emph{Ordering:} If a complete operation~$O_1$ precedes a complete operation~$O_2$ in~$H$, \uline{where both are invoked by correct processes in~$S$}, then $O_1$ returns a smaller value than~$O_2$.
    \item\label{cond:byz-timestamp-uniqueness} \emph{Uniqueness:} No two distinct complete operations in~$H$ \uline{invoked by correct processes in~$S$} return the same value.
    \item[(\getrefnumber{cond:byz-timestamp-liveness})] \emph{Liveness:} Every operation invoked by a correct process in~$S$ eventually returns.
  \end{enumerate}
\end{proposition}
\begin{proof}
  \proofqed{thm:byz-timestamp-characterization}
  The liveness condition~(\ref{cond:byz-timestamp-liveness}) is the same in \Cref{def:byz-timestamp-implementation} and the proposition.
  It therefore suffices to show that, for every schedule~$S$ of~$\alg$ with corresponding history~$H$, (\ref{cond:byz-timestamp-safety}) $\Longleftrightarrow$ (\ref{cond:byz-timestamp-ordering}) $\wedge$ (\ref{cond:byz-timestamp-uniqueness}).
  Each direction reuses the corresponding claim from the proof for systems with crash failures (\Cref{thm:timestamp-characterization}).

  $(\Longrightarrow)$
  Suppose that $H$ is Byzantine linearizable with respect to~\ts{}~(\ref{cond:byz-timestamp-safety}).
  Then there exists a history~$\widehat H$ that is crash linearizable with respect to~\ts{}~(\ref{cond:timestamp-safety}) and whose projection onto the correct processes in~$S$ is the same as that of~$H$.
  By \Cref{claim:timestamp-characterization-forward}, $\widehat H$ satisfies (\ref{cond:timestamp-ordering}) and~(\ref{cond:timestamp-uniqueness}).
  Since $H$ and~$\widehat H$ have the same projection onto the correct processes in~$S$, complete operations invoked by correct processes have the same return values and precedence in the two histories.
  Hence, $H$ satisfies (\ref{cond:byz-timestamp-ordering}) and~(\ref{cond:byz-timestamp-uniqueness}).

  $(\Longleftarrow)$
  Suppose that $H$ satisfies~(\ref{cond:byz-timestamp-ordering}) and~(\ref{cond:byz-timestamp-uniqueness}).
  Let~$H_C$ denote the projection of~$H$ onto the correct processes in~$S$.
  Then $H_C$ satisfies (\ref{cond:timestamp-ordering}) and~(\ref{cond:timestamp-uniqueness}).
  Moreover, because correct processes follow~$\alg$, every complete operation in~$H_C$ returns a positive integer.
  Hence, by \Cref{claim:timestamp-characterization-backward}, $H_C$ is crash linearizable with respect to~\ts{}~(\ref{cond:timestamp-safety}), which makes $H$ Byzantine linearizable with respect to~\ts{}~(\ref{cond:byz-timestamp-safety}).
\end{proof}

\section{Results for Byzantine Failures with Digital Signatures}\label{sec:byz-sig}
In this section, we consider implementations of \ts{} in systems with Byzantine processes and digital signatures.
Recall that we assume there are $n \geq 2$ processes, of which at most $f \geq 1$ may be Byzantine.

\subsection{Maximum Resilience}
\begin{proposition}\label{thm:byzantine-resilience-with-signatures}
  In a system with Byzantine processes and digital signatures, there exists an implementation of~\ts{} only if $n > 3f$.
\end{proposition}

\begin{proof}
  \proofqed{thm:byzantine-resilience-with-signatures}
  Suppose, toward a contradiction, that an implementation~$\alg$ of~\ts{} exists in such a system.
  By \Cref{thm:byz-timestamp-characterization}, $\alg$ satisfies ordering~(\ref{cond:byz-timestamp-ordering}).
  Because $2\leq n\leq3f$, we can choose pairwise disjoint sets $A,B,C$ whose union is the set of processes, such that $|A|,|B|,|C|\leq f$ and $A,C\neq\emptyset$.
  Fix processes $p_a\in A$ and $p_c\in C$.

  Consider a schedule of~$\alg$ in which the processes in~$C$ crash before taking any steps (where $|C|\leq f$), the processes in~$A\cup B$ are correct, and~$p_a$ invokes \nextOpParens{}.
  By liveness, this operation eventually returns some value~$r_a$.
  Let~$\sigma_A$ be the schedule prefix of~$\alg$ ending at its return; thus, $\sigma_A$ contains no steps by processes in~$C$.
  Symmetrically, let~$\sigma_C$ be a schedule prefix of~$\alg$ in which $p_c$ invokes \nextOpParens{} and returns~$r_c$ and processes in~$A$ take no steps (where $|A|\leq f$).
  Without loss of generality, assume $r_a\geq r_c$.

  Now consider the concatenation $\sigma_A\circ\sigma_C$.
  It extends to a schedule of~$\alg$ in which the processes in~$B$ are Byzantine (where $|B|\leq f$) and those in~$A\cup C$ are correct.
  The processes in~$B$ first reproduce their steps in~$\sigma_A$, and then ``roll back'' to their initial states and reproduce their steps in~$\sigma_C$; digital signatures do not prevent a Byzantine process from resetting its own state.
  The messages sent to~$C$ during~$\sigma_A$ and to~$A$ during~$\sigma_C$ are delayed until after both operations return, so that the processes in~$A\cup C$ follow~$\alg$ throughout.
  Since $p_a, p_c\in A\cup C$ are correct and $p_a$'s operation returns before $p_c$'s operation is invoked, condition~(\ref{cond:byz-timestamp-ordering}) implies $r_a<r_c$, a contradiction.
\end{proof}

\subsection{Latency Upper Bounds}\label{sec:byz2delay-sig}
When $n>3f$, \ts{} can be implemented with a worst-case message delay of three.
In fact, \Cref{thm:byz-upper-bound-ngt3f}, proved in \Cref{sec:byz:upper-bound}, gives such an implementation without digital signatures.
Since it also works in systems \emph{with} digital signatures, we immediately obtain the following.

\begin{observation}\label{thm:byz-sig-upper-bound}
  In a system with Byzantine processes and digital signatures, if $n>3f$, then there exists an implementation of~\ts{} whose worst-case message delay is three.
\end{observation}

For the special case $f=1$, \ts{} can be implemented with only two message delays using digital signatures.

\begin{theorem}\label{thm:byz-sig-upper-bound-f1}
  In a system with Byzantine processes and digital signatures, if $n>3f$ and $f=1$, then there exists an implementation of~\ts{} (\Cref{alg:byz-two-delays-sig}) whose worst-case message delay is two.
\end{theorem}
\begin{proof}
  \proofqed{thm:byz-sig-upper-bound-f1}
  We prove the theorem in five parts.
  \Cref{part:byz-two-delays-sig:definitions,part:byz-two-delays-sig:algorithm} introduce the necessary definitions and present the algorithm.
  \Cref{part:byz-two-delays-sig:preliminary-properties} establishes preliminary properties, which we use in \Cref{part:byz-two-delays-sig:safety} to prove that \nextOpParens{} returns positive integers and satisfies ordering and uniqueness (\Cref{thm:byz-two-delays-sig:return-values}).
  Finally, \Cref{part:byz-two-delays-sig:liveness-latency} proves liveness (\Cref{thm:byz-two-delays-sig:liveness}) and shows that the worst-case message delay is two (\Cref{thm:byz-two-delays-sig:latency}).

\proofpart{Definitions}\label{part:byz-two-delays-sig:definitions}

We start by defining some concepts used in the algorithm.
We name the processes $p_0,\ldots,p_{n-1}$.
The notation $\signed{m}{i}$ denotes statement~$m$ carrying a digital signature by process~$p_i$.

\begin{definition}
A \emph{label} by process~$p_i$ has the form
\[
  \ell=\signed{\Label~i,k,r}{i},
  \qquad
  k,r\in\mathbb{Z}_{\ge 1}.
\]
\end{definition}
We write
\begin{itemize}
  \item $\origin(\ell) \deq i$ for the \emph{origin} of~$\ell$,
  \item $\seq(\ell) \deq k$ for its \emph{sequence number},
  \item $\rank(\ell) \deq r$ for its \emph{rank}.
\end{itemize}
Successive invocations by a correct process~$p_i$ are identified by labels of origin~$i$ with consecutive sequence numbers.

\begin{definition}
  An \emph{acknowledgment} of label~$\ell$ by process~$p_i$ is $\signed{\ell}{i}$.
\end{definition}

\begin{definition}
  A \emph{quorum certificate} for label~$\ell$ is a set of acknowledgments
  \[
    C = \{\signed{\ell}{i}:i\in Q\}
  \]
  where $Q\subseteq \{0,\ldots,n-1\}$ and $|Q|=n-1$.
\end{definition}

A quorum certificate for~$\ell$ contains acknowledgments from $n-1$ distinct processes.
Since at most one process is Byzantine, it proves that at least $n-2$ correct processes have handled a request carrying label~$\ell$.

\begin{definition}
  A \emph{certified label} is either $\bot$ or a pair
  \[
    \widehat{\ell}=(\ell,C)
  \]
  where $\ell$ is a label and
  \begin{itemize}
    \item if $\seq(\ell)=1$, then $C=\emptyset$;
    \item if $\seq(\ell)>1$, then $C$ is a quorum certificate for some label $\ell^-$ such that $\origin(\ell^-)=\origin(\ell)$ and $\seq(\ell^-)=\seq(\ell)-1$.
  \end{itemize}
  We say that $\widehat{\ell}$ \emph{contains} label~$\ell$.
\end{definition}
For a non-$\bot$ certified label $\widehat{\ell}=(\ell,C)$, define
\begin{align*}
  \labelof(\widehat{\ell})&\deq\ell,
  &\origin(\widehat{\ell})&\deq\origin(\ell),\\
  \seq(\widehat{\ell})&\deq\seq(\ell),
  &\rank(\widehat{\ell})&\deq\rank(\ell).
\end{align*}
Set $\seq(\bot)\deq 0$ and $\rank(\bot)\deq 0$.

\begin{definition}[Frontier]
  Let $F$ be a sequence of $n$~certified labels.
  We call $F$ a \emph{frontier} if the following conditions hold for every $i=0,\ldots,n-1$:
  \begin{itemize}
    \item $\origin(F[i])=i$ whenever $F[i]\ne\bot$;
    \item $\rank(F[i])\le N(F)$, where
      \[
        N(F)\Deq\sum_{k=0}^{n-1} \seq(F[k]).
      \]
  \end{itemize}
  The \emph{empty frontier} is $\bot^n \Deq (\bot,\ldots,\bot)$.
\end{definition}

\begin{definition}[Merge]
For frontiers $F$ and $G$, the \emph{merge of $G$ into $F$}, denoted $F\merge G$, is the sequence of $n$~certified labels defined by
\[
  (F\merge G)[i]
  \Deq
  \begin{cases}
    G[i] & \text{if $\seq(G[i])>\seq(F[i])$},\\
    F[i] & \text{otherwise}.
  \end{cases}
\]
\end{definition}
This definition directly implies
\[
  \seq((F\merge G)[i])=\max\{\seq(F[i]),\seq(G[i])\}
  \qquad\text{for every $i=0,\ldots,n-1$}.
\]
\begin{proposition}\label{thm:byz-two-delays-sig:merge-frontier}
If $F$ and $G$ are frontiers, then $F\merge G$ is a frontier.
\end{proposition}
\begin{proof}
  \proofqed{thm:byz-two-delays-sig:merge-frontier}
  Write $H\deq F\merge G$.
  For each~$i=0,\ldots,n-1$, the certified label $H[i]$ is either $F[i]$ or $G[i]$, hence is $\bot$ or has origin~$i$.
  Furthermore,
  \begin{align*}
    \rank(H[i]) &\le\max\{\rank(F[i]),\rank(G[i])\} \\
    &\le\max\{N(F),N(G)\} & (\text{$F$, $G$ are frontiers}) \\
    &= \max\left\{\sum_{k=0}^{n-1} \seq(F[k]), \sum_{k=0}^{n-1} \seq(G[k])\right\} \\
    &\le \sum_{k=0}^{n-1} \max\{\seq(F[k]),\seq(G[k])\} \\
    &= \sum_{k=0}^{n-1} \seq(H[k]) & (\text{by the definition of $\merge$}) \\
    &= N(H).
  \end{align*}
  Hence, $H$ is a frontier.
\end{proof}

\begin{definition}
For frontiers $F$ and $G$, we write $F\preceq G$ if
\[
  \seq(F[i])\le\seq(G[i])
  \qquad\text{for every $i=0,\ldots,n-1$}.
\]
\end{definition}

\begin{proposition}\label{thm:byz-two-delays-sig:frontier-preorder}
Let $F$ and $G$ be frontiers.
\begin{enumerate}[label=(\alph*),ref=(\alph*)]
  \item\label{fo:byz-two-delays-sig:preorder} The relation~$\preceq$ is a preorder on frontiers.%
  \footnote{The relation~$\preceq$ is not a partial order because it is not antisymmetric.
  For example, two frontiers can differ in the ranks of corresponding labels while having the same sequence number at every coordinate.
  They are distinct, but each is~$\preceq$ the other.}
  \item\label{fo:byz-two-delays-sig:upper} $F\preceq F\merge G$ and $G\preceq F\merge G$.
  \item\label{fo:byz-two-delays-sig:monotone} If $F\preceq G$, then $N(F)\le N(G)$.
\end{enumerate}
\end{proposition}

\begin{definition}[Messages]
  There are two kinds of valid messages:
  \begin{itemize}
    \item A valid \emph{request} from process~$p_i$ has the form $\langle \Req~F\rangle$, where $F$ is a frontier with $F[i]\ne\bot$.
    \item A valid \emph{reply} from process~$p_j$ has the form $\langle \Reply~\signed{\ell}{j},F\rangle$, where $\ell$ is a label and $F$ is a frontier.
  \end{itemize}
\end{definition}

\proofpart{Algorithm and timestamp encoding}\label{part:byz-two-delays-sig:algorithm}

\Cref{alg:byz-two-delays-sig} gives the implementation.
It computes base timestamps using the two functions $\mathcal{T}_i$ and~$\enc$ from lines~\ref{line:byz-two-delays-sig:def-triple} and~\ref{line:byz-two-delays-sig:def-enc}; we explain the idea behind them here.

\begin{algorithm}
\caption{A worst-case two-delay implementation of~$\ts{}$ in a system with Byzantine processes and digital signatures where $n>3f$ and $f=1$.
Each \gtagtext{gray tag} in the right margin records how the values immediately after its line correspond to the notation in \Cref{part:byz-two-delays-sig:preliminary-properties}.}\label{alg:byz-two-delays-sig}
\begin{algorithmic}[1]
  \State $\begin{aligned}[t]
      \mathcal{T}_i(F,o)\Deq{}&\bigl(N(F),\min\{r_i,r_o\},\ind{r_i>r_o}\bigr)\\
      &\text{where $r_i=\rank(F[i])$, $r_o=\rank(F[o])$}
    \end{aligned}$ \label{line:byz-two-delays-sig:def-triple}
    \Comment{See \Cref{def:byz-two-delays-sig:timestamp-functions}.}
  \State $\enc((s,m,\delta))\Deq s^2+2m+\delta$ \label{line:byz-two-delays-sig:def-enc}
    \Comment{See \Cref{def:byz-two-delays-sig:encoding}.}
  \Statex

  \Statex {\bfseries\boldmath Code for process~$p_i$} ($0 \leq i < n$):
  \State $\varfrontier{i}\gets\bot^n$
  \State $\varlabel{i}\gets\bot$
    \Comment{Label for ongoing invocation, $\bot$ while idle.}
  \State $\varacks{i}\gets\emptyset$ \Comment{Acknowledgments for ongoing invocation; latest label's quorum certificate while idle.}
  \Statex

  \Upon{\textbf{invocation} of \nextOpParens{}} \gtag{\varfrontier{i}=L_\gamma}\label{line:byz-two-delays-sig:invoke-upon}
    \State $\mathit{seq}\gets \seq(\varfrontier{i}[i])+1$ \label{line:byz-two-delays-sig:choose-seq}
    \State $\mathit{rank}\gets N(\varfrontier{i})+1$ \label{line:byz-two-delays-sig:choose-rank}
    \State $\ell\gets\signed{\Label~i,\mathit{seq},\mathit{rank}}{i}$ \gtag{\ell=\ell_\gamma}\label{line:byz-two-delays-sig:create-label}
    \State $\varlabel{i}\gets \ell$
    \State $\varfrontier{i}[i]\gets(\ell,\varacks{i})$ \gtag{\varfrontier{i}=F_\gamma^0}\label{line:byz-two-delays-sig:invoke-frontier}
    \State $\varacks{i}\gets\emptyset$
    \State \textbf{send} $\langle \Req~\varfrontier{i} \rangle$ \textbf{to} every process \label{line:byz-two-delays-sig:send-request}
  \EndUpon

  \Statex

  \Upon{\textbf{receive} valid~$\langle \Req~F \rangle$ \textbf{from} $p_j$}
    \State $\varfrontier{i}\gets\varfrontier{i}\merge F$ \label{line:byz-two-delays-sig:respond-merge}
    \State $\ell\gets\labelof(F[j])$
    \State \textbf{send} $\langle \Reply~\signed{\ell}{i},\varfrontier{i} \rangle$ \textbf{to}~$p_j$ \label{line:byz-two-delays-sig:send-reply}
  \EndUpon

  \Statex

  \Upon{\textbf{receive} valid $\langle \Reply~\signed{\ell}{j},F \rangle$ \textbf{from} $p_j$} \label{line:byz-two-delays-sig:receive-reply}
    \If{$\ell=\varlabel{i}\neq\bot$ \textbf{and} $\signed{\ell}{j}\notin\varacks{i}$} \gtag{F=F_{\gamma,j}}\label{line:byz-two-delays-sig:accept-check}
      \State $\varacks{i}\gets\varacks{i}\cup\{\signed{\ell}{j}\}$
      \State $\varfrontier{i}\gets\varfrontier{i}\merge F$ \label{line:byz-two-delays-sig:reply-frontier-merge}
      \If{$|\varacks{i}|=n-1$} \gtag{\varfrontier{i}=F_\gamma}\label{line:byz-two-delays-sig:completion-check}
        \State $Q\gets\{k\in\{0,\ldots,n-1\}:\signed{\ell}{k}\in\varacks{i}\}$ \gtag{Q=Q_\gamma}\label{line:byz-two-delays-sig:define-quorum}
        \State $o\gets$ the unique index in $\{0,\ldots,n-1\}\setminus Q$ \gtag{o=o_\gamma}\label{line:byz-two-delays-sig:choose-omitted}
        \State $\tau\gets\mathcal{T}_i(\varfrontier{i},o)$ \gtag{\tau=\tau_\gamma}\label{line:byz-two-delays-sig:compute-triple}
        \State $t\gets\enc(\tau)$ \Comment{\emph{Base timestamp}.}\quad\gtagtext{$t=t_\gamma$}\label{line:byz-two-delays-sig:compute-timestamp}
        \State $\varlabel{i}\gets\bot$
        \State \Return $n\cdot t+i$ \Comment{Completes the \nextOpParens{} invocation.}\label{line:byz-two-delays-sig:return-timestamp}
      \EndIf
    \EndIf
  \EndUpon
\end{algorithmic}
\end{algorithm}

For each operation, the algorithm computes a base-timestamp triple~$\tau$ of integers by applying~$\mathcal{T}_i$ (line~\ref{line:byz-two-delays-sig:compute-triple}).
The function~$\mathcal{T}_i$ is defined so that the lexicographic order on triples respects the precedence of operations (proved in \Cref{part:byz-two-delays-sig:safety}).
Because \nextOpParens{} must return a single positive integer, the algorithm then maps~$\tau$ to a positive integer~$t$ (line~\ref{line:byz-two-delays-sig:compute-timestamp}) using the encoding~$\enc$, which preserves the lexicographic order (\Cref{thm:byz-two-delays-sig:encoding}).
This integer~$t$ is the operation's \emph{base timestamp}; the invoker~$p_i$ then returns $n\cdot t+i$ (line~\ref{line:byz-two-delays-sig:return-timestamp}), using its index~$i$ to break ties between operations with the same base timestamp.

\begin{definition}[Base-timestamp triple]\label{def:byz-two-delays-sig:timestamp-functions}
For a frontier~$F$ and process indices $i,o\in \{0,\ldots,n-1\}$, the \emph{base-timestamp triple} is the triple of integers defined by
\[
  \mathcal{T}_i(F,o)
  \Deq
  \bigl(N(F),\min\{r_i,r_o\},\ind{r_i>r_o}\bigr)
\]
where $r_i=\rank(F[i])$ and $r_o=\rank(F[o])$, and where $\ind{P}$ denotes the indicator of a condition~$P$, equal to~$1$ if $P$ holds and to~$0$ otherwise.
\end{definition}

We compare these triples lexicographically, writing $\tau\ltlex\tau'$ for the resulting strict total order.

\begin{definition}[Encoding]\label{def:byz-two-delays-sig:encoding}
The \emph{encoding} of a base-timestamp triple~$\tau=(s,m,\delta)$ is the integer
\[
  \enc(\tau) \Deq s^2+2m+\delta.
\]
\end{definition}

Intuitively, a triple's first component~$s$ confines its encoding to the interval $[s^2,(s+1)^2)$.
These intervals are disjoint and increasing, so the first component dominates.
The remaining two components break ties within an interval.

\begin{proposition}\label{thm:byz-two-delays-sig:encoding-block}
Every base-timestamp triple $\tau=\mathcal{T}_i(F,o)$ satisfies
\[
  N(F)^2\le \enc(\tau)<\left(N(F)+1\right)^2.
\]
\end{proposition}
\begin{proof}
  \proofqed{thm:byz-two-delays-sig:encoding-block}
Let $r_i=\rank(F[i])$ and $r_o=\rank(F[o])$, so that $\tau=(s,m,\delta)$ with $s=N(F)$, $m=\min\{r_i,r_o\}$, and $\delta=\ind{r_i>r_o}\in\{0,1\}$.
Because ranks are nonnegative and $F$ is a frontier, $0\leq r_i,r_o\leq s$, and so $0\le m\le s$.

If $r_i\le r_o$, then $2m+\delta=2r_i\le 2s$.
If $r_i>r_o$, then $2m+\delta=2r_o+1<2(r_o+1)\le 2r_i\le 2s$.
Thus $0\le 2m+\delta\le 2s$ in both cases, and therefore
\[
  s^2\le \enc(\tau)=s^2+2m+\delta\le s^2+2s<(s+1)^2,
\]
which proves the claim.
\end{proof}

\begin{lemma}[Encoding order]\label{thm:byz-two-delays-sig:encoding}
For any base-timestamp triples $\tau$ and $\tau'$,
\[
  \tau\ltlex \tau'
  \quad\Longleftrightarrow\quad
  \enc(\tau)<\enc(\tau').
\]
\end{lemma}
\begin{proof}
  \proofqed{thm:byz-two-delays-sig:encoding}
Since $\ltlex $ is a strict total order, it suffices to prove the forward implication.
Suppose $\tau=(s,m,\delta)\ltlex (s',m',\delta')=\tau'$.
If $s<s'$, then $s+1\le s'$ and \Cref{thm:byz-two-delays-sig:encoding-block} gives
\[
  \enc(\tau)<(s+1)^2\le (s')^2\le\enc(\tau').
\]
Otherwise $s=s'$ and $(m,\delta)\ltlex (m',\delta')$, so $2m+\delta<2m'+\delta'$ (as $m,m'$ are integers and $\delta,\delta'\in\{0,1\}$), and adding~$s^2$ gives $\enc(\tau)<\enc(\tau')$.
\end{proof}

\proofpart{Preliminary properties}\label{part:byz-two-delays-sig:preliminary-properties}

The following are some simple invariants used throughout the safety proof.
\begin{proposition}\label{thm:byz-two-delays-sig:local-invariants}
  The following hold at every correct process~$p_i$:
  \begin{enumerate}[label=(\alph*),ref=(\alph*)]
    \item \label{inv:byz-two-delays-sig:frontier} $\varfrontier{i}$ is a frontier and is non-decreasing under~$\preceq$.
    \item \label{inv:byz-two-delays-sig:correct-labels} The sequence numbers and ranks of labels created by~$p_i$ strictly increase with creation order.
    \item \label{inv:byz-two-delays-sig:valid-messages} Messages sent by~$p_i$ on lines~\ref{line:byz-two-delays-sig:send-request} and~\ref{line:byz-two-delays-sig:send-reply} are valid requests and valid replies, respectively.
\end{enumerate}
\end{proposition}

\Crefpart{thm:byz-two-delays-sig:local-invariants}{inv:byz-two-delays-sig:correct-labels} immediately implies the following.

\begin{corollary}\label{thm:byz-two-delays-sig:correct-label-order}
Let $\ell$ and $\ell'$ be labels created by the same correct process.
If $\seq(\ell)\ge\seq(\ell')$, then $\rank(\ell)\ge\rank(\ell')$.
\end{corollary}

Next, we state some properties of a complete operation invoked by a correct process.
Throughout, we denote by $\gamma$ such an operation and $p_i$ its invoker, and use the following notation (which also appears as right-margin tags in \Cref{alg:byz-two-delays-sig}):
\begin{itemize}
  \item $\ell_\gamma$ is the label created for~$\gamma$ on line~\ref{line:byz-two-delays-sig:create-label}.
  \item $L_\gamma$ is the invoker's frontier at the start of~$\gamma$ (line~\ref{line:byz-two-delays-sig:invoke-upon}).
  \item $F_\gamma^0$ is the frontier carried by $\gamma$'s request, namely the value of $\varfrontier{i}$ after line~\ref{line:byz-two-delays-sig:invoke-frontier}.
  \item $Q_\gamma$ is the value of~$Q$ on line~\ref{line:byz-two-delays-sig:define-quorum}.
  \item $o_\gamma$ is the value of~$o$ from line~\ref{line:byz-two-delays-sig:choose-omitted}.  (So $Q_\gamma=\{0,\ldots,n-1\}\setminus\{o_\gamma\}$.)
  \item For $j\in Q_\gamma$, $F_{\gamma,j}$ is the frontier in the unique reply from~$p_j$ that $p_i$ receives during~$\gamma$ and that passes the check on line~\ref{line:byz-two-delays-sig:accept-check}.
  \item $F_\gamma$ is the value of $\varfrontier{i}$ when the check on line~\ref{line:byz-two-delays-sig:completion-check} succeeds.
  \item $\tau_\gamma=\mathcal{T}_i(F_\gamma,o_\gamma)$ is the base-timestamp triple computed on line~\ref{line:byz-two-delays-sig:compute-triple}.
  \item $t_\gamma=\enc(\tau_\gamma)$ is the base timestamp computed on line~\ref{line:byz-two-delays-sig:compute-timestamp}.
\end{itemize}

\begin{proposition}\label{thm:byz-two-delays-sig:reply-persistence}
Let $p_j$ be a correct process such that $j\in Q_\gamma$, and let $F$ be a frontier.
If $F\preceq\varfrontier{j}$ holds at some point before~$\gamma$ is invoked, then $F\preceq F_\gamma$.
\end{proposition}
\begin{proof}
  \proofqed{thm:byz-two-delays-sig:reply-persistence}
Because $p_j$ is correct, $F_{\gamma,j}$ is the value of $\varfrontier{j}$ when $p_j$ replies to $\gamma$'s request on line~\ref{line:byz-two-delays-sig:send-reply}, after $\gamma$ is invoked; frontier monotonicity (\Crefpart{thm:byz-two-delays-sig:local-invariants}{inv:byz-two-delays-sig:frontier}) therefore gives $F\preceq F_{\gamma,j}$.
Since $j\in Q_\gamma$, process~$p_i$ receives $p_j$'s reply before $\gamma$ completes, merging $F_{\gamma,j}$ into $\varfrontier{i}$ on line~\ref{line:byz-two-delays-sig:reply-frontier-merge}, after which $F_{\gamma,j}\preceq\varfrontier{i}$.
Frontier monotonicity preserves this until the check on line~\ref{line:byz-two-delays-sig:completion-check} succeeds, at which point $\varfrontier{i}=F_\gamma$, so $F_{\gamma,j}\preceq F_\gamma$.
By the transitivity of~$\preceq$ (\Crefpart{thm:byz-two-delays-sig:frontier-preorder}{fo:byz-two-delays-sig:preorder}), we have $F\preceq F_\gamma$.
\end{proof}

\begin{corollary}\label{thm:byz-two-delays-sig:coordinate-persistence}
Let $p_j$ be a correct process with $j\in Q_\gamma$.
For any $k\in \{0,\ldots,n-1\}$ and $s\in\mathbb{Z}_{\ge 0}$, suppose that $\seq(\varfrontier{j}[k])\ge s$ at some point before~$\gamma$ is invoked.
Then $\seq(F_\gamma[k])\ge s$.
\end{corollary}

\begin{proposition}\label{thm:byz-two-delays-sig:own-label-persists}
After $p_i$ executes line~\ref{line:byz-two-delays-sig:invoke-frontier} in invocation $\gamma$, $\labelof(\varfrontier{i}[i])=\ell_\gamma$ holds until $\gamma$ returns.
\end{proposition}
\begin{proof}
  \proofqed{thm:byz-two-delays-sig:own-label-persists}
Line~\ref{line:byz-two-delays-sig:invoke-frontier} sets $\varfrontier{i}[i]$ to the certified label $(\ell_\gamma,\varacks{i})$, so $F_\gamma^0[i]\ne\bot$ and $\labelof(F_\gamma^0[i])=\ell_\gamma$ immediately after.
Fix a point between then and $\gamma$'s return, and let $F$ be the value of $\varfrontier{i}$ at that point.
Frontier monotonicity (\Crefpart{thm:byz-two-delays-sig:local-invariants}{inv:byz-two-delays-sig:frontier}) gives $F_\gamma^0\preceq F$, so
\[
  \seq(F[i])\ge\seq(F_\gamma^0[i])=\seq(\ell_\gamma)\ge1.
\]
In particular $F[i]\ne\bot$, and $\origin(F[i])=i$ because $F$ is a frontier.
So $\labelof(F[i])$ carries correct process~$p_i$'s (unforgeable) signature and was therefore created by~$p_i$.
Since labels are created only on line~\ref{line:byz-two-delays-sig:create-label} and $p_i$ has at most one outstanding invocation, $p_i$ creates no label after~$\ell_\gamma$ before $\gamma$ returns.
Every label that $p_i$ created before~$\ell_\gamma$ has a smaller sequence number by \Crefpart{thm:byz-two-delays-sig:local-invariants}{inv:byz-two-delays-sig:correct-labels}, so $\labelof(F[i])=\ell_\gamma$.
\end{proof}

\begin{proposition}\label{thm:byz-two-delays-sig:invoker-coordinate}
The following equalities hold:
\begin{enumerate}[label=(\alph*),ref=(\alph*)]
  \item\label{ic:byz-two-delays-sig:label} $\ell_\gamma=\labelof(F_\gamma^0[i])=\labelof(F_\gamma[i])$.
  \item\label{ic:byz-two-delays-sig:sequence} $\seq(\ell_\gamma)=\seq(F_\gamma^0[i])=\seq(F_\gamma[i])$.
  \item\label{ic:byz-two-delays-sig:rank} $\rank(\ell_\gamma)=\rank(F_\gamma^0[i])=\rank(F_\gamma[i])$.
  \item\label{ic:byz-two-delays-sig:size} $\rank(\ell_\gamma)=N(F_\gamma^0)$.
\end{enumerate}
\end{proposition}
\begin{proof}
  \proofqed{thm:byz-two-delays-sig:invoker-coordinate}
During the invocation of~$\gamma$, line~\ref{line:byz-two-delays-sig:invoke-frontier} sets $\varfrontier{i}[i]$ to a certified label containing~$\ell_\gamma$.
Hence, $\ell_\gamma = \labelof(F_\gamma^0[i])$.
Moreover, $F_\gamma$ is the value of $\varfrontier{i}$ when the check on line~\ref{line:byz-two-delays-sig:completion-check} succeeds, which precedes $\gamma$'s return, so \Cref{thm:byz-two-delays-sig:own-label-persists} gives $\labelof(F_\gamma[i])=\ell_\gamma$.
We have thus proved part~\ref{ic:byz-two-delays-sig:label}.

Applying $\seq$ and $\rank$ to part~\ref{ic:byz-two-delays-sig:label} immediately yields parts~\labelcref{ic:byz-two-delays-sig:sequence,ic:byz-two-delays-sig:rank}.

Finally, lines~\labelcref{line:byz-two-delays-sig:choose-seq,line:byz-two-delays-sig:invoke-frontier} give
\[
  F_\gamma^0[k]=L_\gamma[k] \quad (k\ne i),
  \qquad
  \seq(F_\gamma^0[i])=\seq(L_\gamma[i])+1.
\]
Therefore, by line~\ref{line:byz-two-delays-sig:choose-rank},
$N(F_\gamma^0)=N(L_\gamma)+1=\rank(\ell_\gamma)$,
proving part~\ref{ic:byz-two-delays-sig:size}.
\end{proof}

\proofpart{Safety}\label{part:byz-two-delays-sig:safety}

Fix complete operations $\alpha$ and $\beta$, invoked by correct processes $p_a$ and $p_b$, respectively, such that $\alpha$ returns before $\beta$ is invoked.
We first establish several relations between the two operations' frontiers and ranks, then use them to prove that \Cref{alg:byz-two-delays-sig} satisfies ordering~(\ref{cond:byz-timestamp-ordering}) and uniqueness~(\ref{cond:byz-timestamp-uniqueness}).

\begin{proposition}\label{thm:byz-two-delays-sig:base-persists}
$F_\alpha^0\preceq F_\beta$.
\end{proposition}
\begin{proof}
  \proofqed{thm:byz-two-delays-sig:base-persists}
Since $|Q_\alpha|=|Q_\beta|=n-1$, we have $|Q_\alpha\cap Q_\beta|
  \ge 2(n-1)-n
  =n-2
  \ge2$.
Since at most one process is Byzantine, $Q_\alpha\cap Q_\beta$ contains a correct process~$p_i$.

Because $i\in Q_\alpha$, process~$p_i$ merges $F_\alpha^0$ into $\varfrontier{i}$ on line~\ref{line:byz-two-delays-sig:respond-merge}, before $\alpha$ returns.
Thus $F_\alpha^0\preceq\varfrontier{i}$ at some point before~$\beta$ is invoked.
Since $i\in Q_\beta$, \Cref{thm:byz-two-delays-sig:reply-persistence} gives $F_\alpha^0\preceq F_\beta$.
\end{proof}

\begin{proposition}\label{thm:byz-two-delays-sig:invoker-rank-persistence}
$\seq(F_\beta[a])\ge\seq(F_\alpha[a])$ and $\rank(F_\beta[a])\ge\rank(F_\alpha[a])$.
\end{proposition}
\begin{proof}
  \proofqed{thm:byz-two-delays-sig:invoker-rank-persistence}
By \Cref{thm:byz-two-delays-sig:base-persists,thm:byz-two-delays-sig:invoker-coordinate},
\[
  \seq(F_\beta[a])
  \ge\seq(F_\alpha^0[a])
  =\seq(F_\alpha[a])
  =\seq(\ell_\alpha)
  \ge1,
\]
which is the first inequality.
It also implies that $F_\alpha[a],F_\beta[a]\neq\bot$, which means they contain labels that carry correct process~$p_a$'s (unforgeable) signature, and so were created by~$p_a$.
The second inequality then follows from \Cref{thm:byz-two-delays-sig:correct-label-order}.
\end{proof}

\begin{proposition}\label{thm:byz-two-delays-sig:invoker-included}
If $a\in Q_\beta$, then $F_\alpha\preceq F_\beta$.
\end{proposition}
\begin{proof}
  \proofqed{thm:byz-two-delays-sig:invoker-included}
When the check on line~\ref{line:byz-two-delays-sig:completion-check} succeeds during~$\alpha$, $\varfrontier{a}=F_\alpha$; this point precedes~$\beta$'s invocation.
Since $p_a$ is correct and $a\in Q_\beta$, \Cref{thm:byz-two-delays-sig:reply-persistence} gives $F_\alpha\preceq F_\beta$.
\end{proof}

\begin{lemma}\label{thm:byz-two-delays-sig:one-loss}
For any $j\in \{0,\ldots,n-1\}$, if $\seq(F_\beta[j])<\seq(F_\alpha[j])$, then process~$p_j$ is Byzantine and $\seq(F_\beta[j])\ge \seq(F_\alpha[j])-1$.
\end{lemma}
\begin{proof}
  \proofqed{thm:byz-two-delays-sig:one-loss}
Fix $j\in \{0,\ldots,n-1\}$ and suppose
\begin{equation}\label{eq:byz-two-delays-sig:assumed-decrease}
  \seq(F_\beta[j])<\seq(F_\alpha[j]).
\end{equation}
This implies $\seq(F_\alpha[j])\ge1$; in particular, $F_\alpha[j]\neq\bot$.
We proceed by proving four claims.

\begin{claim}\label{claim:byz-two-delays-sig:one-loss-omitted}
$o_\beta=a$.
\end{claim}
\begin{proof}
  \proofqed{claim:byz-two-delays-sig:one-loss-omitted}
By \eqref{eq:byz-two-delays-sig:assumed-decrease}, $F_\alpha\not\preceq F_\beta$.
The contrapositive of \Cref{thm:byz-two-delays-sig:invoker-included} therefore gives $a\notin Q_\beta$.
By line~\ref{line:byz-two-delays-sig:choose-omitted}, $Q_\beta$ omits exactly~$o_\beta$, so $o_\beta=a$ and $Q_\beta=\{0,\ldots,n-1\}\setminus\{a\}$.
\end{proof}

\begin{claim}\label{claim:byz-two-delays-sig:one-loss-quorum}
$j\in Q_\beta$.
\end{claim}
\begin{proof}
  \proofqed{claim:byz-two-delays-sig:one-loss-quorum}
If $j=a$, then \Cref{thm:byz-two-delays-sig:invoker-rank-persistence} gives $\seq(F_\beta[a])\ge\seq(F_\alpha[a])$, contradicting \eqref{eq:byz-two-delays-sig:assumed-decrease}.
Hence $j\ne a$, and \Cref{claim:byz-two-delays-sig:one-loss-omitted} gives $j\in Q_\beta$.
\end{proof}

\begin{claim}\label{claim:byz-two-delays-sig:one-loss-byzantine}
Process~$p_j$ is Byzantine.
\end{claim}
\begin{proof}
  \proofqed{claim:byz-two-delays-sig:one-loss-byzantine}
Suppose for contradiction that $p_j$ is correct.
Let $\ell=\labelof(F_\alpha[j])$.
Since $\origin(\ell)=j$, the label~$\ell$ carries correct process~$p_j$'s (unforgeable) signature; hence $p_j$ created~$\ell$ on line~\ref{line:byz-two-delays-sig:create-label}.
Line~\ref{line:byz-two-delays-sig:invoke-frontier} then assigns to $\varfrontier{j}[j]$ a certified label containing~$\ell$.
Because $\ell$ appears in~$F_\alpha$, this assignment precedes $\alpha$'s return, and so it precedes $\beta$'s invocation.

Immediately after the assignment,
\[
  \seq(\varfrontier{j}[j])
  =\seq(\ell)
  =\seq(F_\alpha[j]).
\]
Since $j\in Q_\beta$ by \Cref{claim:byz-two-delays-sig:one-loss-quorum}, \Cref{thm:byz-two-delays-sig:coordinate-persistence} gives
\[
  \seq(F_\beta[j])
  \ge\seq(\ell)
  =\seq(F_\alpha[j]),
\]
contradicting \eqref{eq:byz-two-delays-sig:assumed-decrease}.
Hence $p_j$ is Byzantine.
\end{proof}

\begin{claim}\label{claim:byz-two-delays-sig:one-loss-bound}
$\seq(F_\beta[j])\ge \seq(F_\alpha[j])-1$.
\end{claim}
\begin{proof}
  \proofqed{claim:byz-two-delays-sig:one-loss-bound}
The claim is immediate if $\seq(F_\alpha[j])=1$.
Assume henceforth that $\seq(F_\alpha[j])>1$.
The certified label $F_\alpha[j]$ carries a quorum certificate for a label~$\ell^-$ satisfying
\[
  \origin(\ell^-)=j,
  \qquad
  \seq(\ell^-)=\seq(F_\alpha[j])-1.
\]

The certificate contains acknowledgments of~$\ell^-$ (of the form $\signed{\ell^-}{k}$) from $n-1$ processes, at least $n-2$ of which are correct.
Since $Q_\beta$ omits only one process, at least $n-3\ge1$ of these correct processes belong to~$Q_\beta$.
Fix one such process~$p_i$.

The quorum certificate in~$F_\alpha[j]$ contains an acknowledgment carrying correct process~$p_i$'s (unforgeable) signature, so $p_i$ sent that acknowledgment on line~\ref{line:byz-two-delays-sig:send-reply} before $\alpha$ returns.
Before sending the acknowledgment, on line~\ref{line:byz-two-delays-sig:respond-merge} $p_i$ merges into $\varfrontier{i}$ a frontier~$F$ with $\labelof(F[j])=\ell^-$.
Immediately after,
\[
  \seq(\varfrontier{i}[j])
  \ge\seq(F[j])
  =\seq(\ell^-).
\]
This holds at some point before~$\beta$ is invoked.

Since $i\in Q_\beta$, \Cref{thm:byz-two-delays-sig:coordinate-persistence} now gives
\[
  \seq(F_\beta[j])
  \ge\seq(\ell^-)
  =\seq(F_\alpha[j])-1.
  \qedhere
\]
\end{proof}

\Cref{claim:byz-two-delays-sig:one-loss-byzantine,claim:byz-two-delays-sig:one-loss-bound} establish the two conclusions.
\end{proof}

\begin{proposition}\label{thm:byz-two-delays-sig:own-coordinate-advance}
$\seq(\ell_\beta)>\seq(F_\alpha[b])$ and $\rank(\ell_\beta)>\rank(F_\alpha[b])$.
\end{proposition}
\begin{proof}
  \proofqed{thm:byz-two-delays-sig:own-coordinate-advance}
If $F_\alpha[b]=\bot$, the claim is immediate.
Otherwise, $\labelof(F_\alpha[b])$ carries correct process~$p_b$'s (unforgeable) signature, so $p_b$ created it before~$\alpha$ returned, whereas it created $\ell_\beta$ after~$\beta$ was invoked.
Thus, $\ell_\beta$ comes after $\labelof(F_\alpha[b])$ in $p_b$'s label-creation order and therefore has a larger sequence number and rank by \Crefpart{thm:byz-two-delays-sig:local-invariants}{inv:byz-two-delays-sig:correct-labels}.
\end{proof}

\begin{proposition}\label{thm:byz-two-delays-sig:replier-rank}
If $b\in Q_\alpha$, then $\rank(\ell_\beta)>\rank(\ell_\alpha)$.
\end{proposition}
\begin{proof}
  \proofqed{thm:byz-two-delays-sig:replier-rank}
By \Cref{thm:byz-two-delays-sig:invoker-coordinate}, $N(F_\alpha^0)=\rank(\ell_\alpha)$.
Because $b\in Q_\alpha$, process~$p_b$ merges $F_\alpha^0$ into $\varfrontier{b}$ on line~\ref{line:byz-two-delays-sig:respond-merge} upon receiving a request for~$\alpha$---so before $\alpha$ returns.
Since $\alpha$ returns before~$\beta$ is invoked, frontier monotonicity (\Crefpart{thm:byz-two-delays-sig:local-invariants}{inv:byz-two-delays-sig:frontier}) gives $F_\alpha^0\preceq L_\beta$.
Thus, by line~\ref{line:byz-two-delays-sig:choose-rank} during~$\beta$ and \Crefpart{thm:byz-two-delays-sig:frontier-preorder}{fo:byz-two-delays-sig:monotone},
\[
  \rank(\ell_\beta)
  =N(L_\beta)+1
  \ge N(F_\alpha^0)+1
  =\rank(\ell_\alpha)+1
  >\rank(\ell_\alpha).
\]
\end{proof}

\begin{proposition}[Ordering for base timestamps]\label{thm:byz-two-delays-sig:ordering}
In \Cref{alg:byz-two-delays-sig}, if $\alpha$ and $\beta$ are complete operations invoked by correct processes and $\alpha$ returns before $\beta$ is invoked, then their base timestamps satisfy $t_\beta>t_\alpha$.
\end{proposition}
\begin{proof}
  \proofqed{thm:byz-two-delays-sig:ordering}
Write
\begin{align*}
  x_\alpha&\Deq\rank(F_\alpha[a])=\rank(\ell_\alpha),
    &\text{(``$=$'' by \Cref{thm:byz-two-delays-sig:invoker-coordinate})}\\
  y_\alpha&\Deq\rank(F_\alpha[o_\alpha]),\\
  x_\beta&\Deq\rank(F_\beta[b])=\rank(\ell_\beta),
    &\text{(``$=$'' by \Cref{thm:byz-two-delays-sig:invoker-coordinate})}\\
  y_\beta&\Deq\rank(F_\beta[o_\beta]),
\end{align*}
so that
\begin{alignat*}{4}
  \tau_\alpha
  &= \bigl(N(F_\alpha)
  &,\,& \min\{x_\alpha,y_\alpha\}
  &,\,& \ind{x_\alpha>y_\alpha}
  &\bigr),\\
  \tau_\beta
  &= \bigl(N(F_\beta)
  &,\,& \min\{x_\beta,y_\beta\}
  &,\,& \ind{x_\beta>y_\beta}
  &\bigr).
\end{alignat*}
Since $\enc$ respects the lexicographic order of these triples (\Cref{thm:byz-two-delays-sig:encoding}), it suffices to prove that $\tau_\beta\gtlex \tau_\alpha$.
We prove this via three claims, each comparing one component of $\tau_\alpha$ and $\tau_\beta$ under the assumption that the preceding components are equal.

\begin{claim}\label{claim:byz-two-delays-sig:safety-first}
$N(F_\beta)\ge N(F_\alpha)$.
\end{claim}
\begin{proof}
  \proofqed{claim:byz-two-delays-sig:safety-first}
First, by \Cref{thm:byz-two-delays-sig:invoker-coordinate,thm:byz-two-delays-sig:own-coordinate-advance},
\begin{equation}\label{eq:byz-two-delays-sig:b-advance}
  \seq(F_\beta[b])=\seq(\ell_\beta)\ge \seq(F_\alpha[b])+1.
\end{equation}

Next, since at most one process is Byzantine and $p_b$ is correct, \Cref{thm:byz-two-delays-sig:one-loss} ensures that there exists at most one index $j\ne b$ such that $\seq(F_\beta[j])<\seq(F_\alpha[j])$, and that if such~$j$ exists, then $\seq(F_\beta[j])\ge \seq(F_\alpha[j])-1$.
Summing over $i\ne b$ therefore gives
\begin{equation}\label{eq:byz-two-delays-sig:others-persist}
  \sum_{i\ne b} \seq(F_\beta[i])
  \ge
  \sum_{i\ne b} \seq(F_\alpha[i])-1.
\end{equation}
Putting together \eqref{eq:byz-two-delays-sig:b-advance} and \eqref{eq:byz-two-delays-sig:others-persist}, we get
\begin{align*}
  N(F_\beta)
  &= \seq(F_\beta[b]) + \sum_{i\ne b} \seq(F_\beta[i]) \\
  &\ge \bigl(\seq(F_\alpha[b])+1\bigr)
    + \left(\sum_{i\ne b} \seq(F_\alpha[i])-1\right) \\
  &= N(F_\alpha),
\end{align*}
as claimed.
\end{proof}

\begin{claim}\label{claim:byz-two-delays-sig:safety-second}
If $N(F_\beta)=N(F_\alpha)$, then
\begin{align}
  y_\beta\ge x_\alpha&\ge \min\{x_\alpha,y_\alpha\},\label{eq:byz-two-delays-sig:rank-persist}\\
  x_\beta&>\min\{x_\alpha,y_\alpha\}.\label{eq:byz-two-delays-sig:rank-advance}
\end{align}
Consequently,
\[
  \min\{x_\beta,y_\beta\}\ge\min\{x_\alpha,y_\alpha\}.
\]
\end{claim}
\begin{proof}
  \proofqed{claim:byz-two-delays-sig:safety-second}
Assume that $N(F_\beta)=N(F_\alpha)$.
We first show that $a\notin Q_\beta$.
Suppose instead that $a\in Q_\beta$.
Then \Cref{thm:byz-two-delays-sig:invoker-included} gives $F_\alpha\preceq F_\beta$, so $\seq(F_\beta[i])\ge\seq(F_\alpha[i])$ for every~$i$.
But $\seq(F_\beta[b])>\seq(F_\alpha[b])$ by \eqref{eq:byz-two-delays-sig:b-advance}.
Summing over all~$i$ then yields $N(F_\beta)>N(F_\alpha)$, contradicting our assumption.

Since $Q_\beta$ omits only one index, we conclude that $o_\beta=a$, so $y_\beta=\rank(F_\beta[a])$.
\Cref{thm:byz-two-delays-sig:invoker-rank-persistence} immediately gives $y_\beta\ge x_\alpha$, proving \eqref{eq:byz-two-delays-sig:rank-persist}.

For~\eqref{eq:byz-two-delays-sig:rank-advance}, if $b\in Q_\alpha$, then \Cref{thm:byz-two-delays-sig:replier-rank} gives $x_\beta>x_\alpha\ge\min\{x_\alpha,y_\alpha\}$;
otherwise, $o_\alpha=b$, so $y_\alpha=\rank(F_\alpha[b])$ and \Cref{thm:byz-two-delays-sig:own-coordinate-advance} gives $x_\beta>y_\alpha\ge\min\{x_\alpha,y_\alpha\}$.
So in both cases, $x_\beta>\min\{x_\alpha,y_\alpha\}$.

Inequalities~\eqref{eq:byz-two-delays-sig:rank-persist} and~\eqref{eq:byz-two-delays-sig:rank-advance} together imply that $\min\{x_\beta,y_\beta\}\ge\min\{x_\alpha,y_\alpha\}$.
\end{proof}

\begin{claim}\label{claim:byz-two-delays-sig:safety-third}
If $N(F_\beta)=N(F_\alpha)$ and $\min\{x_\beta,y_\beta\}=\min\{x_\alpha,y_\alpha\}$, then
\[
  \ind{x_\beta>y_\beta}>\ind{x_\alpha>y_\alpha}.
\]
\end{claim}
\begin{proof}
  \proofqed{claim:byz-two-delays-sig:safety-third}
Suppose $N(F_\beta)=N(F_\alpha)$ and $\min\{x_\beta,y_\beta\}=\min\{x_\alpha,y_\alpha\}$.
By \Cref{claim:byz-two-delays-sig:safety-second}, inequalities~\eqref{eq:byz-two-delays-sig:rank-persist} and~\eqref{eq:byz-two-delays-sig:rank-advance} hold.
By \eqref{eq:byz-two-delays-sig:rank-advance}, $x_\beta>\min\{x_\beta,y_\beta\}$, so $y_\beta=\min\{x_\beta,y_\beta\}=\min\{x_\alpha,y_\alpha\}$.
Combined with \eqref{eq:byz-two-delays-sig:rank-persist}, this gives $y_\beta=x_\alpha=\min\{x_\alpha,y_\alpha\}$.
Therefore,
\begin{alignat*}{3}
  x_\alpha&=\min\{x_\alpha,y_\alpha\}\le y_\alpha,
    &&\qquad\text{so }&\ind{x_\alpha>y_\alpha}&=0;\\
  x_\beta&>\min\{x_\beta,y_\beta\}=y_\beta,
    &&\qquad\text{so }&\ind{x_\beta>y_\beta}&=1.
\end{alignat*}
This proves $\ind{x_\beta>y_\beta}>\ind{x_\alpha>y_\alpha}$.
\end{proof}

\Cref{claim:byz-two-delays-sig:safety-first,claim:byz-two-delays-sig:safety-second,claim:byz-two-delays-sig:safety-third} together show that $\tau_\beta\gtlex \tau_\alpha$.
Therefore, $t_\beta=\enc(\tau_\beta)>\enc(\tau_\alpha)=t_\alpha$.
\end{proof}

\begin{proposition}[Return values]\label{thm:byz-two-delays-sig:return-values}
  \Cref{alg:byz-two-delays-sig} returns positive integers from \nextOpParens{} and satisfies ordering~(\ref{cond:byz-timestamp-ordering}) and uniqueness~(\ref{cond:byz-timestamp-uniqueness}).
\end{proposition}
\begin{proof}
  \proofqed{thm:byz-two-delays-sig:return-values}
  Consider a complete operation~$\gamma$ at a correct process~$p_i$.
  By \Crefpart{thm:byz-two-delays-sig:invoker-coordinate}{ic:byz-two-delays-sig:sequence}, $N(F_\gamma)\geq\seq(F_\gamma[i])=\seq(\ell_\gamma)\geq 1$.
  Thus, \Cref{thm:byz-two-delays-sig:encoding-block} gives $t_\gamma\geq N(F_\gamma)^2\geq 1$, so both $t_\gamma$ and $n\cdot t_\gamma+i$ are positive integers.

  For~(\ref{cond:byz-timestamp-ordering}), suppose a complete operation~$\alpha$ at a correct process~$p_a$ precedes a complete operation~$\beta$ at a correct process~$p_b$.
  By \Cref{thm:byz-two-delays-sig:ordering}, $t_\alpha<t_\beta$, so $n\cdot t_\alpha+a<n\cdot(t_\alpha+1)\leq n\cdot t_\beta+b$.

  For~(\ref{cond:byz-timestamp-uniqueness}), timestamps returned by different correct processes differ modulo~$n$, and those returned by the same correct process are distinct by~(\ref{cond:byz-timestamp-ordering}) since its operations are sequential.
\end{proof}

\proofpart{Liveness and latency}\label{part:byz-two-delays-sig:liveness-latency}

\begin{proposition}[Liveness]\label{thm:byz-two-delays-sig:liveness}
In \Cref{alg:byz-two-delays-sig}, every invocation of \nextOpParens{} by a correct process eventually returns.
\end{proposition}
\begin{proof}
  \proofqed{thm:byz-two-delays-sig:liveness}
On line~\ref{line:byz-two-delays-sig:send-request}, a correct invoker sends every process a request that is valid by \Crefpart{thm:byz-two-delays-sig:local-invariants}{inv:byz-two-delays-sig:valid-messages}; hence every correct process eventually receives it.
By line~\ref{line:byz-two-delays-sig:invoke-frontier}, the request frontier contains the invoker's current label at the invoker's coordinate, so each correct recipient sends a valid reply acknowledging that label on line~\ref{line:byz-two-delays-sig:send-reply}.
These replies have distinct signers and therefore pass the check on line~\ref{line:byz-two-delays-sig:accept-check}.
At least $n-1$ processes, including the invoker itself, are correct, so eventually the invoker accepts $n-1$ replies, the check on line~\ref{line:byz-two-delays-sig:completion-check} succeeds, and the invoker returns on line~\ref{line:byz-two-delays-sig:return-timestamp}.
\end{proof}

\begin{proposition}[Latency]\label{thm:byz-two-delays-sig:latency}
  The worst-case message delay of \Cref{alg:byz-two-delays-sig} is two.
\end{proposition}
\begin{proof}[Proof sketch]
  \proofqed{thm:byz-two-delays-sig:latency}
  An invocation follows one \Req{} delay and one \Reply{} delay, as illustrated below:
  \[
    \begin{tikzpicture}[
      baseline=(current bounding box.south),
      event/.style={circle, fill, inner sep=2pt},
      msg/.style={-{Stealth[length=2mm]}, thin},
      every node/.style={font=\small}
    ]
      \node[event, label={[align=right, xshift=-1em, yshift=-2pt]below:{invoke\\\nextOpParens{}}}] (invoke) at (0,0.35) {};
      \node[event] (req1) at (2.4,1.05) {};
      \node[event] (req2) at (2.4,0.35) {};
      \node[event] (req3) at (2.4,-0.35) {};
      \node[event] (req4) at (2.4,-1.05) {};
      \node[event, label={[align=left, xshift=1em, yshift=-2pt]below:{return}}] (return) at (4.8,0.35) {};

      \node[fill=white, inner sep=1pt] at (1.2,1.45) {\Req};
      \node[fill=white, inner sep=1pt] at (3.6,1.45) {\Reply};

      \foreach \a in {req1,req2,req3,req4} {
        \draw[msg] (invoke) -- (\a);
        \draw[msg] (\a) -- (return);
      }
    \end{tikzpicture}\qedhere
  \]
\end{proof}

  By \Cref{thm:byz-two-delays-sig:return-values,thm:byz-two-delays-sig:liveness}, \Cref{alg:byz-two-delays-sig} satisfies the conditions of \Cref{thm:byz-timestamp-characterization} and hence implements~\ts{}.
  By \Cref{thm:byz-two-delays-sig:latency}, its worst-case message delay is two.
\end{proof}

\subsection{Latency Lower Bounds}
The latency lower bounds for systems with crash failures (\Cref{sec:crash:lower-bounds}) also apply to systems with Byzantine processes and digital signatures.
Thus, \Cref{thm:crash-two-delay-lower-bound} immediately implies the following.

\begin{observation}\label{thm:byz-sig-lower-bound-f1}
  In a system with Byzantine processes and digital signatures, any implementation of~\ts{} has worst-case message delay at least two.
\end{observation}

And \Cref{thm:crash-lower-bound-fgeq2} similarly implies the following.

\begin{observation}\label{thm:byz-sig-lower-bound-fgeq2}
  In a system with Byzantine processes and digital signatures, if $f\geq 2$, then any implementation of~\ts{} has worst-case message delay at least three.
\end{observation}

\subsection{Summary}

\Cref{thm:byzantine-resilience-with-signatures} and \Cref{thm:byz-sig-upper-bound} together give the following necessary and sufficient condition for implementing~\ts{} under Byzantine failures with digital signatures.

\begin{corollary}
  In a system with Byzantine processes and digital signatures, there exists an implementation of~\ts{} if and only if $n>3f$.
\end{corollary}

We next state the optimal latency for implementations of~\ts{} in systems with Byzantine processes and digital signatures.
When $n>3f$ and $f\geq 2$, the latency upper bound in \Cref{thm:byz-sig-upper-bound} matches the latency lower bound in \Cref{thm:byz-sig-lower-bound-fgeq2}, yielding the following.

\begin{corollary}
  In a system with Byzantine processes and digital signatures, if $n>3f$ and $f\geq 2$, then the optimal latency of implementing~\ts{} is three message delays.
\end{corollary}

When $n>3f$ and $f=1$, the latency upper bound in \Cref{thm:byz-sig-upper-bound-f1} matches the latency lower bound in \Cref{thm:byz-sig-lower-bound-f1}, yielding the following.

\begin{corollary}
  In a system with Byzantine processes and digital signatures, if $n>3f$ and $f=1$, then the optimal latency of implementing~\ts{} is two message delays.
\end{corollary}

Thus, the optimal latency of implementing~\ts{} under Byzantine failures with digital signatures is settled in all cases where an implementation exists.

\section{Results for Byzantine Failures without Digital Signatures}\label{sec:byz}
In this section, we consider implementations of \ts{} in systems with Byzantine processes, but without digital signatures.
Again, we assume there are $n \geq 2$ processes, of which at most $f \geq 1$ may be Byzantine.

\subsection{Maximum Resilience}
The resilience lower bound for systems with Byzantine processes and digital signatures, given by \Cref{thm:byzantine-resilience-with-signatures}, holds also for systems \emph{without} digital signatures:

\begin{observation}\label{thm:byzantine-resilience-no-signatures}
  In a system with Byzantine processes, there exists an implementation of~\ts{} only if $n > 3f$.
\end{observation}

\subsection{Latency Upper Bound}\label{sec:byz:upper-bound}

\begin{theorem}\label{thm:byz-upper-bound-ngt3f}
  In a system with Byzantine processes, if $n>3f$, then there exists an implementation of~\ts{} (\Cref{alg:byz-three-delays}) whose worst-case message delay is three.
\end{theorem}

\begin{proof}
  \proofqed{thm:byz-upper-bound-ngt3f}
The algorithm uses three types of messages.
We call a message \emph{valid} if it has one of the following forms:
\begin{itemize}
  \item a request $\langle \Req~q\rangle$, where $q$ is an integer;
  \item a proposal $\langle \Prop~\langle k,q\rangle,v\rangle$, where $k\in\{0,\ldots,n-1\}$ and $q$ and~$v$ are integers;
  \item an echo $\langle \Echo~q,j,v\rangle$, where $j\in\{0,\ldots,n-1\}$ and $q$ and~$v$ are integers.
\end{itemize}
We first establish some supporting properties of the algorithm.
We then show that the algorithm implements~$\ts{}$ by proving that it returns positive integers and satisfies ordering~(\ref{cond:byz-timestamp-ordering}), uniqueness~(\ref{cond:byz-timestamp-uniqueness}), and liveness~(\ref{cond:byz-timestamp-liveness}).
Finally, we show that its worst-case message delay is three.

\begin{algorithm}
  \caption{A worst-case three-delay implementation of~$\ts{}$ in a system with Byzantine processes where $n>3f$.}\label{alg:byz-three-delays}
  \begin{algorithmic}[1]
    \Statex {\bfseries\boldmath Code for process~$p_i$} ($0 \leq i < n$):
    \State $\xx_i \gets 0$ \label{line:byz-three-delays-clock-init}\Comment{$p_i$'s \emph{local clock}.}
    \State $\texttt{seq}_i \gets 0$ \Comment{Sequence number of $p_i$'s ongoing invocation, or of its next invocation when idle.}\label{line:byz-three-delays-seq-init}
    \State $\texttt{reqSeq}_i[k] \gets -1$ for each $k$ \Comment{Largest $q$ in a $\langle \Req~q\rangle$ received from $p_k$.}\label{line:byz-three-delays-reqseq-init}
    \State $\texttt{propSeq}_i[k,j] \gets -1$ for each $k,j$ \Comment{Largest $q$ in a $\langle \Prop~\langle k, q\rangle, - \rangle$ received from $p_j$.}\label{line:byz-three-delays-propseq-init}
    \Statex

    \State $\texttt{phase}_i \gets \textsf{idle}$ \Comment{Whether $p_i$ is \textsf{idle} or has an \textsf{ongoing} invocation.}
    \State $\texttt{echoVal}_i[j,k] \gets \bot$ for each $j,k$ \Comment{$p_j$'s proposal for $p_i$'s ongoing invocation as echoed by $p_k$.}
    \Statex

    \Upon{\textbf{invocation} of \nextOpParens{}}
      \State $\texttt{phase}_i \gets \textsf{ongoing}$
      \State $\texttt{echoVal}_i[j,k] \gets \bot$ for each $j,k$
      \State \textbf{send} $\langle \Req~\texttt{seq}_i \rangle$ \textbf{to} every process \label{line:byz-three-delays-req-send}
    \EndUpon
    \Statex

    \Upon{\textbf{receive} valid~$\langle \Req~q \rangle$ \textbf{from} $p_k$}
      \If{$q > \texttt{reqSeq}_i[k]$} \label{line:byz-three-delays-req-guard}
        \State $\texttt{reqSeq}_i[k] \gets q$ \label{line:byz-three-delays-reqseq-update}
        \State $\xx_i \gets \xx_i + 1$ \label{line:byz-three-delays-clock-increment}
        \State \textbf{send} $\langle \Prop~\langle k, q\rangle, \xx_i \rangle$ \textbf{to} every process \label{line:byz-three-delays-prop-send}
      \EndIf
    \EndUpon
    \Statex

    \Upon{\textbf{receive} valid~$\langle \Prop~\langle k, q \rangle, v \rangle$ \textbf{from} $p_j$}
      \If{$q > \texttt{propSeq}_i[k,j]$} \label{line:byz-three-delays-prop-guard}
        \State $\texttt{propSeq}_i[k,j] \gets q$ \label{line:byz-three-delays-propseq-update}
        \State $\xx_i \gets \max(\xx_i, v)$ \label{line:byz-three-delays-clock-max}
        \State \textbf{send} $\langle \Echo~q, j, v \rangle$ \textbf{to} $p_k$ \label{line:byz-three-delays-echo-send}\Comment{$p_i$ says to $p_k$: ``$p_j$ proposed $v$ for your invocation~$q$.''}
      \EndIf
    \EndUpon
    \Statex

    \Upon{\textbf{receive} valid~$\langle \Echo~q, j, v \rangle$ \textbf{from} $p_k$}
      \If{$\texttt{phase}_i = \textsf{ongoing}$ \textbf{and} $q = \texttt{seq}_i$} \label{line:byz-three-delays-current-echo}
        \If{$\texttt{echoVal}_i[j,k] = \bot$} \label{line:byz-three-delays-echo-record-test}
          \State $\texttt{echoVal}_i[j,k] \gets v$ \label{line:byz-three-delays-echo-record}
        \EndIf
        \State $\mathit{certified} \gets \left\{ \langle j', v' \rangle : v'\neq\bot,\, |\{ k' : \texttt{echoVal}_i[j',k'] = v' \}| \geq n-f \right\}$ \label{line:byz-three-delays-certified}
        \Statex \Comment{The proposal by $p_{j'}$ of value~$v'$ is \emph{certified} once it is echoed by $n-f$ distinct processes~$p_{k'}$.}
        \If{$|\mathit{certified}| \geq n-f$} \label{line:byz-three-delays-certified-full}
          \State $\texttt{phase}_i \gets \textsf{idle}$
          \State $\texttt{seq}_i \gets \texttt{seq}_i + 1$ \label{line:byz-three-delays-seq-increment}
          \State $t \gets \max\{v' : \langle j',v'\rangle\in\mathit{certified}\}$ \Comment{\emph{Base timestamp}.}\label{line:byz-three-delays-return-compute}
          \State \Return $n\cdot t+i$ \Comment{Completes the \nextOpParens{} invocation.}\label{line:byz-three-delays-return}
        \EndIf
      \EndIf
    \EndUpon
  \end{algorithmic}
\end{algorithm}

\begin{claim}\label{claim:byz-three-delays-clock-monotone}
  For every correct process~$p_i$, its local clock~$\xx_i$ is non-decreasing.
\end{claim}
\begin{proof}
  \proofqed{claim:byz-three-delays-clock-monotone}
  After initialization, the only assignments to~$\xx_i$ are on lines~\ref{line:byz-three-delays-clock-increment} and~\ref{line:byz-three-delays-clock-max}, neither of which decreases~$\xx_i$.
\end{proof}

\begin{claim}\label{claim:byz-three-delays-certified-support}
  If the set~$\mathit{certified}$ computed by a correct process~$p_i$ on line~\ref{line:byz-three-delays-certified} contains $\langle j,v\rangle$, then at that time, at least $n-2f$ correct processes~$p_k$ satisfy $\xx_k\geq v$.
\end{claim}
\begin{proof}
  \proofqed{claim:byz-three-delays-certified-support}
  Let $S=\{k:\texttt{echoVal}_i[j,k]=v\}$ when $p_i$ executes line~\ref{line:byz-three-delays-certified}.
  Then $|S|\geq n-f$, so at least $n-2f$ indices in~$S$ belong to correct processes.
  For each such correct index~$k$, process~$p_k$ sent its \Echo{} message only after line~\ref{line:byz-three-delays-clock-max} made $\xx_k\geq v$.
  By \Cref{claim:byz-three-delays-clock-monotone}, $\xx_k\geq v$ still holds when $p_i$ executes line~\ref{line:byz-three-delays-certified}.
\end{proof}

\begin{claim}\label{claim:byz-three-delays-return-support}
  If an invocation of~\nextOpParens{} by a correct process computes base timestamp~$t$ on line~\ref{line:byz-three-delays-return-compute}, then at the time of return, at least $n-2f$ correct processes~$p_k$ satisfy $\xx_k\geq t$.
\end{claim}
\begin{proof}
  \proofqed{claim:byz-three-delays-return-support}
  By line~\ref{line:byz-three-delays-return-compute}, $\langle j,t\rangle\in\mathit{certified}$ for some~$j$.
  By \Cref{claim:byz-three-delays-certified-support}, $\xx_k\geq t$ for at least $n-2f$ correct processes~$p_k$ when $\mathit{certified}$ is computed, and \Cref{claim:byz-three-delays-clock-monotone} implies that these inequalities still hold when the invocation returns.
\end{proof}

\begin{claim}\label{claim:byz-three-delays-certified-unique}
  If the set~$\mathit{certified}$ computed by a correct process~$p_i$ on line~\ref{line:byz-three-delays-certified} contains both $\langle j,v\rangle$ and $\langle j,v'\rangle$, then $v=v'$.
\end{claim}
\begin{proof}
  \proofqed{claim:byz-three-delays-certified-unique}
  Let $S_u=\{k:\texttt{echoVal}_i[j,k]=u\}$ when $p_i$ executes line~\ref{line:byz-three-delays-certified}.
  Suppose $v\neq v'$.
  Since both pairs belong to $\mathit{certified}$, we have $|S_v|,|S_{v'}|\geq n-f$.
  Since each entry $\texttt{echoVal}_i[j,k]$ records at most one non-$\bot$ value, $S_v\cap S_{v'}=\emptyset$.
  But $|S_v \cap S_{v'}|=|S_v|+|S_{v'}|-|S_v\cup S_{v'}|\geq 2(n-f)-n=n-2f>0$, a contradiction.
\end{proof}

\begin{claim}[Ordering for base timestamps]\label{claim:byz-three-delays:ordering}
  Suppose two \nextOpParens{} operations invoked by correct processes compute base timestamps~$t$ and~$t'$ on line~\ref{line:byz-three-delays-return-compute}, respectively.
  If the first returns before the second is invoked, then $t<t'$.
\end{claim}
\begin{proof}
  \proofqed{claim:byz-three-delays:ordering}
  By \Cref{claim:byz-three-delays-return-support}, when the first operation returns, there is a set~$C$ of process indices such that $|C|\geq n-2f$ and for every $k\in C$, $p_k$ is correct and $\xx_k\geq t$.
  By \Cref{claim:byz-three-delays-clock-monotone}, these inequalities still hold when the second operation is invoked and thereafter.

  Let $p_j$ be the process performing the second operation, let $q_j$ be the value of~$\texttt{seq}_j$ when it is invoked, let $\mathit{certified}'$ be $p_j$'s certified set when it returns on line~\ref{line:byz-three-delays-return}, and let $P=\{k:\langle k,-\rangle\in\mathit{certified}'\}$.
  The check on line~\ref{line:byz-three-delays-certified-full} ensures that $|\mathit{certified}'|\geq n-f$.
  By \Cref{claim:byz-three-delays-certified-unique}, distinct pairs in $\mathit{certified}'$ have distinct first components, so $|P|=|\mathit{certified}'|\geq n-f$.
  Hence,
  \[
    |C\cap P| = |C|+|P|-|C \cup P| \geq (n-2f)+(n-f)-n = n-3f > 0,
  \]
  so we may fix some $\ell\in C\cap P$; process~$p_\ell$ is correct.

  Since $\ell\in P$, there is a value~$v_\ell$ such that $\langle \ell,v_\ell\rangle\in\mathit{certified}'$.
  By line~\ref{line:byz-three-delays-certified}, process~$p_j$ received an $\langle\Echo~q_j, \ell, v_\ell\rangle$ message from at least $n-f$ distinct processes;
  at least $n-2f>0$ of them are correct, so fix one such correct process~$p_m$.
  By line~\ref{line:byz-three-delays-echo-send}, $p_m$ sent that \Echo{} only after receiving
  from~$p_\ell$ the message $\langle\Prop~\langle j,q_j\rangle, v_\ell\rangle$.
  As $p_\ell$ is correct, it sent that \Prop{} to every process on line~\ref{line:byz-three-delays-prop-send} immediately after executing $\xx_\ell\gets\xx_\ell+1$ on line~\ref{line:byz-three-delays-clock-increment}, upon receiving the \Req{} message that $p_j$ sent at the start of the second operation.

  Because $\ell\in C$, $\xx_\ell\geq t$ immediately before this assignment, so $v_\ell\geq t+1>t$.
  Since line~\ref{line:byz-three-delays-return-compute} computes the base timestamp as the maximum value in $\mathit{certified}'$, we have $t'\geq v_\ell>t$.
\end{proof}

\begin{claim}[Return values]\label{claim:byz-three-delays:return-values}
  \Cref{alg:byz-three-delays} returns positive integers from \nextOpParens{} and satisfies ordering~(\ref{cond:byz-timestamp-ordering}) and uniqueness~(\ref{cond:byz-timestamp-uniqueness}).
\end{claim}
\begin{proof}
  \proofqed{claim:byz-three-delays:return-values}
  Consider a complete operation at a correct process~$p_i$ with base timestamp~$t$, and let $\mathit{certified}$ be the set from which $t$ is computed.
  By the check on line~\ref{line:byz-three-delays-certified-full}, this set contains at least $n-f$ pairs, whose first components are distinct by \Cref{claim:byz-three-delays-certified-unique}.
  Since $n-f>f$, $\mathit{certified}$ contains a pair $\langle\ell,v_\ell\rangle$ such that $p_\ell$ is correct.
  Some correct process echoed~$v_\ell$ after receiving it in a \Prop{} message from~$p_\ell$, which sent that message right after incrementing~$\xx_\ell$ on line~\ref{line:byz-three-delays-clock-increment}; hence $v_\ell\geq 1$.
  Moreover, every value in~$\mathit{certified}$ is an integer, since it comes from a valid \Echo{} message.
  Thus, $t$, the maximum of these values, is a positive integer, and so is $n\cdot t+i$.

  For~(\ref{cond:byz-timestamp-ordering}), suppose a complete operation at a correct process~$p_i$ with base timestamp~$t$ precedes one at a correct process~$p_j$ with base timestamp~$t'$.
  By \Cref{claim:byz-three-delays:ordering}, $t<t'$, so $n\cdot t+i<n\cdot(t+1)\leq n\cdot t'+j$.

  For~(\ref{cond:byz-timestamp-uniqueness}), timestamps returned by different correct processes differ modulo~$n$, and those returned by the same correct process are distinct by~(\ref{cond:byz-timestamp-ordering}) since its operations are sequential.
\end{proof}

\begin{claim}[Liveness]\label{claim:byz-three-delays:liveness}
  In \Cref{alg:byz-three-delays}, every invocation of~\nextOpParens{} by a correct process eventually returns.
\end{claim}
\begin{proof}
  \proofqed{claim:byz-three-delays:liveness}
  Every message sent by a correct process is valid, since $\texttt{seq}_i$ and~$\xx_i$ always hold integers and every other field is either a process index or copied from a valid message.
  Let $p_i$ be a correct process that invokes~\nextOpParens{}, and let $q$ be the value of~$\texttt{seq}_i$ at that invocation.
  By initialization (lines~\ref{line:byz-three-delays-seq-init} and~\ref{line:byz-three-delays-reqseq-init}) and the increment of~$\texttt{seq}_i$ after each return (line~\ref{line:byz-three-delays-seq-increment}), the \Req{} message for~$q$ is fresh at every correct process.
  Hence every correct process eventually receives it, passes the check on line~\ref{line:byz-three-delays-req-guard}, and sends a corresponding \Prop{} message to every process.

  Until the invocation returns, $p_i$ remains in phase~$\textsf{ongoing}$ with $\texttt{seq}_i=q$ and sends no request with a higher sequence number.
  Likewise, by initialization (lines~\ref{line:byz-three-delays-seq-init} and~\ref{line:byz-three-delays-propseq-init}) and the increment of~$\texttt{seq}_i$ after each return (line~\ref{line:byz-three-delays-seq-increment}), every \Prop{} message for~$q$ from a correct proposer is fresh at every correct process, which eventually passes the check on line~\ref{line:byz-three-delays-prop-guard} and sends the corresponding \Echo{} to~$p_i$.
  Thus, for each of at least $n-f$ correct proposers, $p_i$ eventually receives $n-f$ matching \Echo{} messages.

  Since the corresponding $\texttt{echoVal}_i$ entries were reset to~$\bot$ when the invocation began, $p_i$ records all $n-f$ echoes on line~\ref{line:byz-three-delays-echo-record}; together, the resulting entries certify the proposal.
  Consequently, the condition $|\mathit{certified}|\geq n-f$ eventually holds, and $p_i$ returns on line~\ref{line:byz-three-delays-return}.
\end{proof}

\begin{claim}[Latency]\label{claim:byz-three-delays:latency}
  The worst-case message delay of \Cref{alg:byz-three-delays} is three.
\end{claim}
\begin{proof}[Proof sketch]
  \proofqed{claim:byz-three-delays:latency}
  An invocation follows one \Req{} delay, one \Prop{} delay, and one \Echo{} delay, as illustrated below (where a thin arrow denotes a single message and a thick arrow denotes multiple messages):
  \[
    \begin{tikzpicture}[
      baseline=(current bounding box.south),
      event/.style={circle, fill, inner sep=2pt},
      msg/.style={-{Stealth[length=2mm]}, thin},
      echomsg/.style={-{Stealth[length=2mm]}, very thick},
      manymsg/.style={-{Stealth[length=2mm]}, thin},
      every node/.style={font=\small}
    ]
      \node[event, label={[align=right, xshift=-1em, yshift=-2pt]below:{invoke\\\nextOpParens{}}}] (invoke) at (0,0.35) {};
      \node[event] (req1) at (2.4,1.05) {};
      \node[event] (req2) at (2.4,0.35) {};
      \node[event] (req3) at (2.4,-0.35) {};
      \node[event] (req4) at (2.4,-1.05) {};
      \node[event] (prop1) at (5.0,1.05) {};
      \node[event] (prop2) at (5.0,0.35) {};
      \node[event] (prop3) at (5.0,-0.35) {};
      \node[event] (prop4) at (5.0,-1.05) {};
      \node[event, label={[align=left, xshift=1em, yshift=-2pt]below:{return}}] (return) at (7.4,0.35) {};

      \node[fill=white, inner sep=1pt] at (1.2,1.45) {\Req};
      \node[fill=white, inner sep=1pt] at (3.7,1.45) {\Prop};
      \node[fill=white, inner sep=1pt] at (6.2,1.45) {\Echo};

      \draw[msg] (invoke) -- (req1);
      \draw[msg] (invoke) -- (req2);
      \draw[msg] (invoke) -- (req3);
      \draw[msg] (invoke) -- (req4);

      \foreach \a in {req1,req2,req3,req4} {
        \foreach \b in {prop1,prop2,prop3,prop4} {
          \draw[manymsg] (\a) -- (\b);
        }
      }

      \draw[echomsg] (prop1) -- (return);
      \draw[echomsg] (prop2) -- (return);
      \draw[echomsg] (prop3) -- (return);
      \draw[echomsg] (prop4) -- (return);
    \end{tikzpicture}\qedhere
  \]
\end{proof}

  By \Cref{claim:byz-three-delays:return-values,claim:byz-three-delays:liveness}, \Cref{alg:byz-three-delays} satisfies the conditions of \Cref{thm:byz-timestamp-characterization} and hence implements~\ts{}.
  By \Cref{claim:byz-three-delays:latency}, its worst-case message delay is three.
\end{proof}

\subsection{Latency Lower Bounds}

\subsubsection{Auxiliary Lemmas}\label{sec:byz:auxiliary-lemmas}

We begin with two lemmas used in the proofs below.

First, the notion of a loquacious algorithm (\Cref{def:loquacious}) carries over unchanged.
A Byzantine process that deviates from~$\alg$ need not behave loquaciously.
The loquacious-transformation lemma has the following Byzantine analogue.

\begin{restatable}[Byzantine version of \Cref{thm:loquacious-transformation}]{lemma}{loquaciousTransformationByz}\label{thm:loquacious-transformation-byz}
  In a system with Byzantine processes, let~$\alg$ be an implementation of~\ts{}.
  Then there exists a loquacious implementation~$\alg'$ of~\ts{} such that $\latency{\alg'}\leq\latency{\alg}$.
\end{restatable}

We defer its proof to \Cref{sec:loquacious-transformation-proofs:byz}.

Our second auxiliary lemma adapts the completion lemma (\Cref{thm:two-delay-completion}) to the Byzantine setting by adding one hypothesis:
Every process in~$Q$ must follow~$\alg$ throughout the prefix.
Without this requirement, $p$~may deviate from the algorithm and fail to return; alternatively, another process $q\in Q$ may send a reply that $p$~recognizes as invalid, in which case $p$~may, without violating the two-delay bound, disregard it and wait for a reply from a process outside~$Q$ before returning.
We state the resulting Byzantine variant below, with \underline{the differences from~\Cref{thm:two-delay-completion} underlined}.
(We index sequences from~$0$.
For a sequence~$\mu$, we write $\mu[i..j]$ for the contiguous subsequence $\mu[i],\dots,\mu[j]$.)

\begin{restatable}[Byzantine version of \Cref{thm:two-delay-completion}]{lemma}{twoDelayCompletionByz}\label{thm:two-delay-completion-byz}
  In a system with \underline{Byzantine processes}, let $\alg$~be a loquacious algorithm with $\latency{\alg} \leq 2$.
  Let $p$~be a process, and let $Q$ be a set of processes such that $|Q|=n-f$ and $p\in Q$.

  Suppose that $\sigma$ is a schedule prefix of~$\alg$ in which the following hold:
  \begin{enumerate}
    \item Process~$p$ invokes an operation and broadcasts a message~$m$ in step~$\sigma[a]$.
    \item For each $q \in Q$, process~$q$ receives~$m$ and broadcasts a message~$m_q$ in step~$\sigma[j_q]$.
    \item Process~$p$ receives the set of messages $\{m_q : q\in Q\}$ in step~$\sigma[b]$.
    \item Between steps $\sigma[a]$ and $\sigma[b]$, inclusive, the only steps are those specified above.
    \item \underline{Processes in~$Q$ follow the algorithm~$\alg$ throughout~$\sigma[0..b]$.}
  \end{enumerate}
  Then the operation invoked by~$p$ in step~$\sigma[a]$ returns no later than step~$\sigma[b]$.
\end{restatable}

Its proof is very close to that of \Cref{thm:two-delay-completion}; we defer it to \Cref{sec:completion-lemma-proofs:byz}.

\subsubsection{The Bounds}

The latency lower bounds for systems with crash failures (\Cref{sec:crash:lower-bounds}) also apply to systems with Byzantine processes, since any implementation that tolerates Byzantine processes also tolerates crashes.
Thus, \Cref{thm:crash-two-delay-lower-bound} immediately implies the following.

\begin{observation}\label{thm:byz-lower-bound-f1}
  In a system with Byzantine processes, any implementation of~\ts{} has worst-case message delay at least two.
\end{observation}

Similarly, \Cref{thm:crash-lower-bound-fgeq2} immediately implies the following.

\begin{observation}\label{thm:byz-lower-bound-fgeq2}
  In a system with Byzantine processes, if $f\geq 2$, then any implementation of~\ts{} has worst-case message delay at least three.
\end{observation}

When $f=1$, the best lower bound we obtain for $n>4$ is two delays (from \Cref{thm:byz-lower-bound-f1}).
However, in the boundary case $n=4$, we prove a stronger lower bound of three delays.

\begin{theorem}\label{thm:byz-lower-bound-n4f1}
  In a system with Byzantine processes, if $n=4$ and $f=1$, then any implementation of~\ts{} has worst-case message delay at least three.
\end{theorem}
\begin{figure}
  \centering
  \begin{subfigure}{\textwidth}
    \centering
    \begin{tikzpicture}[x=1cm,y=1cm]
  \tikzset{
    process line/.style={draw=black!35, line width=0.55pt, densely dotted},
    proc label/.style={anchor=east, font=\scriptsize},
    operation/.style={line width=0.8pt},
    operation tick/.style={line width=0.8pt},
    message/.style={-{Latex}, line width=0.55pt, shorten >=1.4pt},
    byzantine message/.style={-{Latex}, line width=0.55pt, shorten >=1.4pt, decorate, decoration={snake, amplitude=.35mm, segment length=2mm}},
    event/.style={circle, draw=none, inner sep=1.45pt},
    return label/.style={anchor=west, font=\scriptsize, xshift=2pt},
  }
  \node[proc label] at (-0.18,-1.120) {\phantom{$\byzprocess{p_2}$}};
\node[font=\small, anchor=north] at (7.350,-2.260) {Auxiliary prefix $\alpha_0$};
\node[proc label, text=\pzerocolor] at (-0.18,-0.000) {$p_0$};
\draw[process line] (0,-0.000) -- (14.700,-0.000);
\node[proc label, text=\ponecolor] at (-0.18,-0.560) {$p_1$};
\draw[process line] (0,-0.560) -- (14.700,-0.560);
\node[proc label, text=black] at (-0.18,-1.120) {$p_2$};
\draw[process line] (0,-1.120) -- (14.700,-1.120);
\node[proc label, text=black] at (-0.18,-1.680) {$p_3$};
\draw[process line] (0,-1.680) -- (14.700,-1.680);
\draw[operation, draw=\ponecolor] (1.050,-0.560) -- (14.700,-0.560);
\draw[operation tick, draw=\ponecolor] (1.050,-0.710) -- (1.050,-0.410);
\draw[message, draw=\ponecolor] (1.050,-0.560) -- (3.150,-1.120);
\node[event, fill=\ponecolor] at (3.150,-1.120) {};
\node[event, fill=\ponecolor] at (1.050,-0.560) {};
\draw[operation, draw=\pzerocolor] (8.400,-0.000) -- (12.600,-0.000);
\draw[operation tick, draw=\pzerocolor] (8.400,-0.150) -- (8.400,0.150);
\draw[message, draw=\pzerocolor] (8.400,-0.000) to[out=30,in=150,looseness=.65] (10.395,-0.000);
\node[event, fill=\pzerocolor] at (10.395,-0.000) {};
\draw[message, draw=\pzerocolor] (8.400,-0.000) -- (10.500,-1.120);
\node[event, fill=\pzerocolor] at (10.500,-1.120) {};
\draw[message, draw=\pzerocolor] (8.400,-0.000) -- (10.605,-1.680);
\node[event, fill=\pzerocolor] at (10.605,-1.680) {};
\draw[message, draw=\pzerocolor] (10.395,-0.000) to[out=30,in=150,looseness=.65] (12.600,-0.000);
\draw[message, draw=\pzerocolor] (10.500,-1.120) -- (12.600,-0.000);
\draw[message, draw=\pzerocolor] (10.605,-1.680) -- (12.600,-0.000);
\draw[operation tick, draw=\pzerocolor] (12.600,-0.150) -- (12.600,0.150);
\node[event, fill=\pzerocolor] at (12.600,-0.000) {};
\node[return label, text=\pzerocolor] at (12.600,-0.000) {$r_0$};
\node[event, fill=\pzerocolor] at (8.400,-0.000) {};
\node[font=\small, anchor=north] at (7.350,-6.090) {Auxiliary prefix $\alpha_1$};
\node[proc label, text=\pzerocolor] at (-0.18,-3.830) {$p_0$};
\draw[process line] (0,-3.830) -- (14.700,-3.830);
\node[proc label, text=\ponecolor] at (-0.18,-4.390) {$p_1$};
\draw[process line] (0,-4.390) -- (14.700,-4.390);
\node[proc label, text=black] at (-0.18,-4.950) {$p_2$};
\draw[process line] (0,-4.950) -- (14.700,-4.950);
\node[proc label, text=black] at (-0.18,-5.510) {$p_3$};
\draw[process line] (0,-5.510) -- (14.700,-5.510);
\draw[operation, draw=\pzerocolor] (1.050,-3.830) -- (14.700,-3.830);
\draw[operation tick, draw=\pzerocolor] (1.050,-3.980) -- (1.050,-3.680);
\draw[message, draw=\pzerocolor] (1.050,-3.830) -- (3.255,-5.510);
\node[event, fill=\pzerocolor] at (3.255,-5.510) {};
\node[event, fill=\pzerocolor] at (1.050,-3.830) {};
\draw[operation, draw=\ponecolor] (8.400,-4.390) -- (12.600,-4.390);
\draw[operation tick, draw=\ponecolor] (8.400,-4.540) -- (8.400,-4.240);
\draw[message, draw=\ponecolor] (8.400,-4.390) to[out=30,in=150,looseness=.65] (10.395,-4.390);
\node[event, fill=\ponecolor] at (10.395,-4.390) {};
\draw[message, draw=\ponecolor] (8.400,-4.390) -- (10.500,-4.950);
\node[event, fill=\ponecolor] at (10.500,-4.950) {};
\draw[message, draw=\ponecolor] (8.400,-4.390) -- (10.605,-5.510);
\node[event, fill=\ponecolor] at (10.605,-5.510) {};
\draw[message, draw=\ponecolor] (10.395,-4.390) to[out=30,in=150,looseness=.65] (12.600,-4.390);
\draw[message, draw=\ponecolor] (10.500,-4.950) -- (12.600,-4.390);
\draw[message, draw=\ponecolor] (10.605,-5.510) -- (12.600,-4.390);
\draw[operation tick, draw=\ponecolor] (12.600,-4.540) -- (12.600,-4.240);
\node[event, fill=\ponecolor] at (12.600,-4.390) {};
\node[return label, text=\ponecolor] at (12.600,-4.390) {$r_1$};
\node[event, fill=\ponecolor] at (8.400,-4.390) {};
\end{tikzpicture}
    \caption{Auxiliary schedule prefixes, in which all processes follow algorithm~$\alg$.}
    \label{fig:byz-n4f1-auxiliary}
  \end{subfigure}

  \bigskip

  \begin{subfigure}{\textwidth}
    \centering
    \begin{tikzpicture}[x=1cm,y=1cm]
  \tikzset{
    process line/.style={draw=black!35, line width=0.55pt, densely dotted},
    proc label/.style={anchor=east, font=\scriptsize},
    operation/.style={line width=0.8pt},
    operation tick/.style={line width=0.8pt},
    message/.style={-{Latex}, line width=0.55pt, shorten >=1.4pt},
    byzantine message/.style={-{Latex}, line width=0.55pt, shorten >=1.4pt, decorate, decoration={snake, amplitude=.35mm, segment length=2mm}},
    event/.style={circle, draw=none, inner sep=1.45pt},
    return label/.style={anchor=west, font=\scriptsize, xshift=2pt},
  }
  \node[proc label] at (-0.18,-1.120) {\phantom{$\byzprocess{p_2}$}};
\node[font=\small, anchor=north] at (7.350,-2.260) {Principal prefix $\sigma_{01}$: $r_0<r_1$};
\node[proc label, text=\pzerocolor] at (-0.18,-0.000) {$p_0$};
\draw[process line] (0,-0.000) -- (14.700,-0.000);
\node[proc label, text=\ponecolor] at (-0.18,-0.560) {$p_1$};
\draw[process line] (0,-0.560) -- (14.700,-0.560);
\node[proc label, text=\byzantinecolor] at (-0.18,-1.120) {$\byzprocess{p_2}$};
\draw[process line, draw=\byzantinecolor] (0,-1.120) -- (14.700,-1.120);
\node[proc label, text=black] at (-0.18,-1.680) {$p_3$};
\draw[process line] (0,-1.680) -- (14.700,-1.680);
\draw[operation, draw=\pzerocolor] (1.050,-0.000) -- (5.250,-0.000);
\draw[operation tick, draw=\pzerocolor] (1.050,-0.150) -- (1.050,0.150);
\draw[message, draw=\pzerocolor] (1.050,-0.000) to[out=30,in=150,looseness=.65] (3.045,-0.000);
\node[event, fill=\pzerocolor] at (3.045,-0.000) {};
\draw[message, draw=\pzerocolor] (1.050,-0.000) -- (3.150,-1.120);
\node[event, fill=\pzerocolor] at (3.150,-1.120) {};
\draw[message, draw=\pzerocolor] (1.050,-0.000) -- (3.255,-1.680);
\node[event, fill=\pzerocolor] at (3.255,-1.680) {};
\draw[message, draw=\pzerocolor] (3.045,-0.000) to[out=30,in=150,looseness=.65] (5.250,-0.000);
\draw[byzantine message, draw=\pzerocolor] (3.150,-1.120) -- (5.250,-0.000);
\draw[message, draw=\pzerocolor] (3.255,-1.680) -- (5.250,-0.000);
\draw[operation tick, draw=\pzerocolor] (5.250,-0.150) -- (5.250,0.150);
\node[event, fill=\pzerocolor] at (5.250,-0.000) {};
\node[return label, text=\pzerocolor] at (5.250,-0.000) {$r_0$};
\node[event, fill=\pzerocolor] at (1.050,-0.000) {};
\draw[operation, draw=\ponecolor] (8.400,-0.560) -- (12.600,-0.560);
\draw[operation tick, draw=\ponecolor] (8.400,-0.710) -- (8.400,-0.410);
\draw[message, draw=\ponecolor] (8.400,-0.560) to[out=30,in=150,looseness=.65] (10.395,-0.560);
\node[event, fill=\ponecolor] at (10.395,-0.560) {};
\draw[message, draw=\ponecolor] (8.400,-0.560) -- (10.500,-1.120);
\node[event, fill=\ponecolor] at (10.500,-1.120) {};
\draw[message, draw=\ponecolor] (8.400,-0.560) -- (10.605,-1.680);
\node[event, fill=\ponecolor] at (10.605,-1.680) {};
\draw[message, draw=\ponecolor] (10.395,-0.560) to[out=30,in=150,looseness=.65] (12.600,-0.560);
\draw[byzantine message, draw=\ponecolor] (10.500,-1.120) -- (12.600,-0.560);
\draw[message, draw=\ponecolor] (10.605,-1.680) -- (12.600,-0.560);
\draw[operation tick, draw=\ponecolor] (12.600,-0.710) -- (12.600,-0.410);
\node[event, fill=\ponecolor] at (12.600,-0.560) {};
\node[return label, text=\ponecolor] at (12.600,-0.560) {$r_1$};
\node[event, fill=\ponecolor] at (8.400,-0.560) {};
\node[font=\small, anchor=north] at (7.350,-6.090) {Principal prefix $\sigma_{10}$: $r_1<r_0$};
\node[proc label, text=\pzerocolor] at (-0.18,-3.830) {$p_0$};
\draw[process line] (0,-3.830) -- (14.700,-3.830);
\node[proc label, text=\ponecolor] at (-0.18,-4.390) {$p_1$};
\draw[process line] (0,-4.390) -- (14.700,-4.390);
\node[proc label, text=black] at (-0.18,-4.950) {$p_2$};
\draw[process line] (0,-4.950) -- (14.700,-4.950);
\node[proc label, text=\byzantinecolor] at (-0.18,-5.510) {$\byzprocess{p_3}$};
\draw[process line, draw=\byzantinecolor] (0,-5.510) -- (14.700,-5.510);
\draw[operation, draw=\ponecolor] (1.050,-4.390) -- (5.250,-4.390);
\draw[operation tick, draw=\ponecolor] (1.050,-4.540) -- (1.050,-4.240);
\draw[message, draw=\ponecolor] (1.050,-4.390) to[out=30,in=150,looseness=.65] (3.045,-4.390);
\node[event, fill=\ponecolor] at (3.045,-4.390) {};
\draw[message, draw=\ponecolor] (1.050,-4.390) -- (3.150,-4.950);
\node[event, fill=\ponecolor] at (3.150,-4.950) {};
\draw[message, draw=\ponecolor] (1.050,-4.390) -- (3.255,-5.510);
\node[event, fill=\ponecolor] at (3.255,-5.510) {};
\draw[message, draw=\ponecolor] (3.045,-4.390) to[out=30,in=150,looseness=.65] (5.250,-4.390);
\draw[message, draw=\ponecolor] (3.150,-4.950) -- (5.250,-4.390);
\draw[byzantine message, draw=\ponecolor] (3.255,-5.510) -- (5.250,-4.390);
\draw[operation tick, draw=\ponecolor] (5.250,-4.540) -- (5.250,-4.240);
\node[event, fill=\ponecolor] at (5.250,-4.390) {};
\node[return label, text=\ponecolor] at (5.250,-4.390) {$r_1$};
\node[event, fill=\ponecolor] at (1.050,-4.390) {};
\draw[operation, draw=\pzerocolor] (8.400,-3.830) -- (12.600,-3.830);
\draw[operation tick, draw=\pzerocolor] (8.400,-3.980) -- (8.400,-3.680);
\draw[message, draw=\pzerocolor] (8.400,-3.830) to[out=30,in=150,looseness=.65] (10.395,-3.830);
\node[event, fill=\pzerocolor] at (10.395,-3.830) {};
\draw[message, draw=\pzerocolor] (8.400,-3.830) -- (10.500,-4.950);
\node[event, fill=\pzerocolor] at (10.500,-4.950) {};
\draw[message, draw=\pzerocolor] (8.400,-3.830) -- (10.605,-5.510);
\node[event, fill=\pzerocolor] at (10.605,-5.510) {};
\draw[message, draw=\pzerocolor] (10.395,-3.830) to[out=30,in=150,looseness=.65] (12.600,-3.830);
\draw[message, draw=\pzerocolor] (10.500,-4.950) -- (12.600,-3.830);
\draw[byzantine message, draw=\pzerocolor] (10.605,-5.510) -- (12.600,-3.830);
\draw[operation tick, draw=\pzerocolor] (12.600,-3.980) -- (12.600,-3.680);
\node[event, fill=\pzerocolor] at (12.600,-3.830) {};
\node[return label, text=\pzerocolor] at (12.600,-3.830) {$r_0$};
\node[event, fill=\pzerocolor] at (8.400,-3.830) {};
\end{tikzpicture}
    \caption{Principal schedule prefixes, in which the process marked \textcolor{\byzantinecolor}{Byz} deviates from algorithm~$\alg$.}
    \label{fig:byz-n4f1-principal}
  \end{subfigure}
  \caption{\textbf{Proof of \Cref{thm:byz-lower-bound-n4f1}.}
    In a system with $n=4$ processes, at most $f=1$ of which may be Byzantine, suppose toward a contradiction that some loquacious implementation~$\alg$ of~\ts{} has $\latency{\alg}\leq2$.
    Each dot marks a step and is colored by the invocation it serves, while each arrow shows a message received within the prefix.
    A horizontal segment tracks a \nextOpParens{} operation: \opinvocationmark{} marks its invocation step, and \opreturnmark{} labeled~$r$ indicates that it has returned~$r$ by that step;
    without a return marker, the figure makes no claim about whether the operation returns within the prefix.
    A solid arrow~\solidmessagearrow{} shows a message sent according to~$\alg$; a wavy arrow~\wavymessagearrow{} in a principal prefix shows a Byzantine message with the same content as the corresponding message in an auxiliary prefix.
    By \Cref{thm:two-delay-completion-byz}, the operations of~$p_0$ in~$\alpha_0$ and~$p_1$ in~$\alpha_1$ return by the final depicted steps; let $\textcolor{\pzerocolor}{r_0}$ and~$\textcolor{\ponecolor}{r_1}$ denote their respective return values.
    Process~$p_0$ cannot distinguish either principal prefix from~$\alpha_0$, and $p_1$ cannot distinguish either from~$\alpha_1$.
    Thus, in both principal prefixes, $p_0$ returns~$\textcolor{\pzerocolor}{r_0}$ and $p_1$ returns~$\textcolor{\ponecolor}{r_1}$.
    Condition~(\ref{cond:byz-timestamp-ordering}) then gives $\textcolor{\pzerocolor}{r_0}<\textcolor{\ponecolor}{r_1}$ in~$\sigma_{01}$ and $\textcolor{\ponecolor}{r_1}<\textcolor{\pzerocolor}{r_0}$ in~$\sigma_{10}$, a contradiction.
  }
  \label{fig:byz-n4f1}
\end{figure}

\begin{proof}
  \proofqed{thm:byz-lower-bound-n4f1}
  Consider a system with $n=4$ processes named $\{p_0,p_1,p_2,p_3\}$, at most $f=1$ of which may be Byzantine.
  Suppose, toward a contradiction, that in this system there exists an implementation of~\ts{} with worst-case message delay less than three.
  By \Cref{thm:loquacious-transformation-byz}, there is a loquacious implementation~$\alg$ of~\ts{} with $\latency{\alg}<3$.
  Since finite worst-case message delays are integers (\Cref{thm:integer-worst-case-delay}), $\latency{\alg}\leq2$.
  By \Cref{thm:byz-timestamp-characterization}, $\alg$ satisfies ordering~(\ref{cond:byz-timestamp-ordering}).

  We construct four schedule prefixes of~$\alg$: two \emph{auxiliary} prefixes $\alpha_0$ and~$\alpha_1$ as shown in \Cref{fig:byz-n4f1-auxiliary}, and two \emph{principal} prefixes $\sigma_{01}$ and~$\sigma_{10}$ as shown in \Cref{fig:byz-n4f1-principal}.%

  First consider the auxiliary prefixes $\alpha_0$ and~$\alpha_1$, shown in \Cref{fig:byz-n4f1-auxiliary}.
  Schedule prefix~$\alpha_0$ begins with process~$p_1$ invoking an operation and~$p_2$ receiving its broadcast.
  Next, $p_0$ invokes an operation and broadcasts a message.
  Processes $p_0,p_2,p_3$ each receive this message and broadcast a reply, in three separate steps, after which $p_0$ receives all three replies in a single step.
  All processes follow~$\alg$, so by \Cref{thm:two-delay-completion-byz} (with $p=p_0$ and $Q=\{p_0,p_2,p_3\}$), $p_0$'s operation has returned by the final step; let~$r_0$ be the value returned.
  Similarly, in~$\alpha_1$, process~$p_0$ invokes an operation and~$p_3$ receives its broadcast before~$p_1$ invokes an operation and completes an analogous round trip with $p_1,p_2,p_3$.
  Again, by \Cref{thm:two-delay-completion-byz} (with $p=p_1$ and $Q=\{p_1,p_2,p_3\}$), $p_1$'s operation has returned by the final step; let~$r_1$ be the value returned.

  Next, construct~$\sigma_{01}$ as shown in \Cref{fig:byz-n4f1-principal}: Process~$p_0$ completes a round trip with $p_0,p_2,p_3$, followed by~$p_1$ with $p_1,p_2,p_3$.
  Processes $p_0,p_1,p_3$ follow~$\alg$, whereas $p_2$ deviates from~$\alg$ by sending its reply from~$\alpha_0$ to~$p_0$ and its reply from~$\alpha_1$ to~$p_1$.%
  \footnote{%
    Note that $p_2$ may not be able to do this if $\alg$ used digital signatures.
    In that case, $p_2$'s reply to~$p_0$ in~$\alpha_0$ may include data signed by~$p_1$, received in~$p_1$'s invocation broadcast.
    Process~$p_2$ could not reproduce that reply in~$\sigma_{01}$: When $p_2$ must send it, $p_1$ has not yet invoked, so the signature does not exist and~$p_2$ cannot forge it.
    The same applies to~$p_3$ in~$\sigma_{10}$.%
  }
  Process~$p_3$ receives $p_0$'s invocation message and then~$p_1$'s, so its state at each reply is the one it has in~$\alpha_0$ and in~$\alpha_1$; it therefore sends those replies to~$p_0$ and~$p_1$, respectively.
  All remaining messages are delayed until after~$\sigma_{01}$.
  These steps form a schedule prefix of~$\alg$, which extends to a schedule in which $p_2$ is the only Byzantine process.
  Prefix~$\sigma_{01}$ is indistinguishable from~$\alpha_0$ to~$p_0$ and from~$\alpha_1$ to~$p_1$, so their operations return $r_0$ and~$r_1$, respectively.
  The former returns before the latter is invoked, so applying condition~(\ref{cond:byz-timestamp-ordering}) to the extended schedule yields $r_0<r_1$.

  Symmetrically, construct~$\sigma_{10}$ as shown in \Cref{fig:byz-n4f1-principal}, with $p_1$'s operation first and $p_0$'s second.
  Here, processes $p_0,p_1,p_2$ follow~$\alg$, whereas $p_3$ deviates from~$\alg$ by sending its reply from~$\alpha_1$ to~$p_1$ and its reply from~$\alpha_0$ to~$p_0$.
  As above, $\sigma_{10}$ extends to a schedule of~$\alg$ in which $p_3$ is the only Byzantine process.
  It is indistinguishable from~$\alpha_1$ to~$p_1$ and from~$\alpha_0$ to~$p_0$, so the operations return $r_1$ and~$r_0$, respectively.
  Applying condition~(\ref{cond:byz-timestamp-ordering}) to the extended schedule yields $r_1<r_0$, contradicting $r_0<r_1$.
\end{proof}

\subsection{Summary and Open Question}
\Cref{thm:byzantine-resilience-no-signatures} and \Cref{thm:byz-upper-bound-ngt3f} together give the following necessary and sufficient condition for implementing~\ts{} in a system with Byzantine processes.

\begin{corollary}
  In a system with Byzantine processes, there exists an implementation of~\ts{} if and only if $n>3f$.
\end{corollary}

We now turn to the optimal latency.
When $n>3f$ and $f\geq 2$, the latency upper bound in \Cref{thm:byz-upper-bound-ngt3f} matches the latency lower bound in \Cref{thm:byz-lower-bound-fgeq2}, yielding the following.

\begin{corollary}
  In a system with Byzantine processes, if $n>3f$ and $f\geq 2$, then the optimal latency of implementing~\ts{} is three message delays.
\end{corollary}

When $n=4$ and $f=1$, the latency upper bound in \Cref{thm:byz-upper-bound-ngt3f} matches the latency lower bound in \Cref{thm:byz-lower-bound-n4f1}, yielding the following.

\begin{corollary}
  In a system with Byzantine processes, if $n=4$ and $f=1$, then the optimal latency of implementing~\ts{} is three message delays.
\end{corollary}

One case remains open: $n>4$ and $f=1$.
Here the upper bound is three delays (\Cref{thm:byz-upper-bound-ngt3f}), whereas the lower bound is two delays (\Cref{thm:byz-lower-bound-f1}).

\begin{question}\label{question:three-delay-lower-bound-byzantine-feq1}
  In a system with Byzantine processes, if $n>4$ and $f=1$, what is the optimal latency of implementing~\ts{}?
\end{question}

\section{Related Work}\label{sec:related-work}
\paragraph{Applications of timestamps.}
Timestamps have long been used in concurrent and distributed algorithms to order requests and distinguish versions of shared data.
For example, in \citeauthor{DBLP:journals/cacm/Lamport78}'s mutual-exclusion algorithm~\cite[561]{DBLP:journals/cacm/Lamport78}, processes order requests for a shared resource by their timestamps; once a request is granted, any subsequently issued request receives a larger timestamp.
In optimistic concurrency control for databases, \textcite[Section~3.2]{DBLP:journals/tods/KungR81} assign transaction numbers from a global counter so that a transaction that completes before another starts receives a smaller number.
In an early shared-memory construction of MWMR atomic registers, \textcite[Section~3]{DBLP:conf/focs/VitanyiA86} attach ordered tags to written values to determine which values reads return.
Motivated by such uses, we study the resilience and latency of \ts{}~implementations in the presence of crash and Byzantine failures.

\paragraph{Message-passing register algorithms.}
Many register implementations in message-passing systems generate timestamps internally, and we can extract this logic to obtain an implementation of~\ts{}.
For example, in the MWMR atomic register of \textcite[Section~4.1]{DBLP:conf/ftcs/LynchS97}, a writer queries a read quorum for tags and forms a new tag by incrementing the largest sequence number received and pairing it with its own identifier; the writer then propagates the tagged value to a write quorum and waits for acknowledgments.
If we instantiate this algorithm with majority quorums, omit register values, and have each writer return its new tag as a positive integer preserving tag order, we can obtain an implementation of~\ts{} in systems with crash failures when $n>2f$.
However, the resulting implementation takes two round trips per operation, for a worst-case message delay of four.
In contrast, our algorithms achieve the optimal latency: two message delays when $f=1$ (\Cref{alg:crash-two-delays}) and three when $f\geq 2$ (\Cref{alg:crash-three-delays}).

As mentioned in \Cref{sec:crash3delay}, \Cref{alg:crash-three-delays} adopts the max-min technique~\cite[Section~1.1]{DBLP:conf/podc/DuttaGLC04}, previously used to obtain three-delay reads for MWMR atomic registers~\cite[Section~4]{DBLP:conf/netys/HadjistasiNS17}.

Some storage protocols that tolerate Byzantine servers~\cites[Section~3.2]{DBLP:journals/dc/MalkhiR98}[Section~3, Figure~1]{DBLP:conf/wdag/MartinAD02} use the same general write pattern as \citeauthor{DBLP:conf/ftcs/LynchS97}'s register algorithm.
A writer queries a quorum for timestamps, chooses a fresh timestamp larger than those received, and propagates the timestamped value, waiting for acknowledgments from a quorum.
To extract their timestamping logic, we can omit the stored data, wait for $n-f$ responses in each phase, and have each write return its new timestamp.
This yields an implementation of~\ts{} that does not rely on digital signatures in systems with Byzantine processes when $n>3f$.
However, its worst-case message delay is four, whereas \Cref{alg:byz-three-delays} achieves three.

\paragraph{Timestamp services in distributed databases.}
Some distributed databases use a dedicated timestamp service for transaction processing.
Percolator's timestamp oracle~\cite[Section~2.3]{DBLP:conf/osdi/PengD10} runs on a single server with stable storage and supplies strictly increasing timestamps.
It can serve requests from memory in one round trip (two message delays).
To preserve timestamp monotonicity across restarts, the oracle periodically records an upper bound on its allocated timestamps in stable storage and resumes above that bound after a restart.
The oracle cannot serve requests while it is down, whereas our model requires operations invoked by correct processes to complete despite up to~$f$ (permanent) crashes.

\paragraph{Timestamp abstractions in shared memory.}
\textcite[Algorithm~4]{DBLP:journals/jacm/AspnesACE15} give a limited-use, wait-free implementation of~\ts{} from a max register.
Other work implements slightly weaker objects from read/write registers: the \emph{weak counter} of \textcite[Section~4]{DBLP:conf/wdag/GuerraouiR05}, designed for anonymous systems, and \emph{fetch-and-increment with multiplicity} of \textcite[Section~4.2]{DBLP:journals/dc/CastanedaRR23}.
For both objects, an operation returns an integer larger than those returned by preceding operations, but concurrent operations may return the same integer.
Breaking such ties by process identifiers yields wait-free implementations of~\ts{}, although for the weak counter this sacrifices anonymity.
All of these shared-memory constructions can be run in message-passing systems via register emulation~\cite{DBLP:journals/jacm/AttiyaBD95,DBLP:conf/ftcs/LynchS97}, but the resulting implementations do not achieve the optimal latency of implementing~\ts{}.
\textcite{DBLP:conf/wdag/EllenFR07} and \textcite{DBLP:conf/podc/HelmiHPW11} study how many read/write registers are needed to implement variants of~\ts{}, whereas we study latency in message-passing systems.

\section{Conclusion}
We have determined the maximum resilience of~\ts{} implementations in asynchronous message-passing systems with crash or Byzantine failures, as well as the optimal latency of implementing~\ts{} in all but one case.
The open case concerns systems with Byzantine processes and no digital signatures when $n>4$ and $f=1$: Do two message delays suffice, or are three necessary?
When $n=4$ and $f=1$, we have shown that digital signatures reduce the optimal latency from three message delays to two, so resolving this question would show whether digital signatures offer the same advantage for $n>4$.

\section{AI Methodology}\label{sec:ai-methodology}
The ideas behind \Cref{alg:byz-two-delays-sig} emerged from a multi-agent search using GPT-5.6 Sol in Codex.
We then put considerable effort into simplifying the construction (with help from GPT and Claude), developing our own notation, and writing our own presentation of the algorithm and its proof to make them easier to follow.

The ideas behind \Cref{alg:byz-three-delays} emerged from a conversation with GPT-5.5 Pro that began with two of our own three-delay \ts{}~implementations in systems with Byzantine processes: one with signatures for $n>3f$, and one without signatures for $n>4f$.
GPT proposed an implementation without signatures for $n=4$ and $f=1$, and then generalized it to all $n>3f$.
We simplified the algorithm in discussion with GPT, then wrote our own presentation of the algorithm and its proof to make them easier to follow.

GPT also contributed to proving two auxiliary results.
GPT-5.4 Pro proposed a retiming argument for \Cref{thm:floor-worst-case-delay} after we found an error in its first attempt, and GPT-5.5 developed the proof of \nameCref{thm:liveness-from-finite-delay}~\hyperref[proof:liveness-from-finite-delay]{\ref*{thm:liveness-from-finite-delay}} under our guidance.
We simplified the time mapping in the former and wrote the exposition of both proofs.

We also used GPT, Claude, and Gemini for literature search, writing code to generate figures, and copyediting.
We checked all technical content, including every algorithm and proof, and take full responsibility for this paper.

\newrefcontext[sorting=anyt]  %
\printbibliography

\appendix

\section{Auxiliary Notations}\label{sec:auxiliary-notation}
We use the following notation for finite sequences of steps:
\begin{itemize}
  \item $|\sigma|$ is the number of steps in a schedule prefix~$\sigma$; and
  \item $\langle s_1,\dots,s_k\rangle$ is the finite sequence consisting of the steps $s_1,\dots,s_k$, in that order.
\end{itemize}

We also introduce notation for the affine maps used below to retime runs.
In interval notation, an overstruck bracket~$\openclosed$ or~$\closedopen$ denotes an endpoint that may be open or closed.

\begin{definition}[Affine map]\label{def:affine-map}
  Let $I=\openclosed a,b\closedopen$ and $J=\openclosed c,d\closedopen$ be bounded real intervals with the same endpoint types, where $a<b$ and $c<d$.
  We write
  \[
    \Aff_{I\to J}:I\to J
  \]
  for the map defined by
  \[
    \Aff_{I\to J}(x) \Deq c+\frac{d-c}{b-a}(x-a).
  \]

  We use the same notation for half-infinite intervals bounded on the same side and having the same finite-endpoint type:
  \begin{align*}
    \Aff_{\openclosed a,+\infty)\to\openclosed c,+\infty)}(x)
      &\Deq x+(c-a), \\
    \Aff_{(-\infty,b\closedopen\to(-\infty,d\closedopen}(x)
      &\Deq x+(d-b).
  \end{align*}
\end{definition}

\begin{proposition}\label{thm:affine-map}
  $\Aff_{I\to J}$ is a strictly increasing bijection from~$I$ to~$J$.
\end{proposition}

\section{Properties of Finite Worst-Case Message Delay}\label{sec:finite-delay}
In this section, we prove the two consequences of finite worst-case message delay stated in \Cref{sec:delay}.
Although originally stated for crash failures, these results---along with their proofs---hold under Byzantine failures as well, with ``correct'' interpreted in the respective sense.

\subsection{Finite Worst-Case Message Delay Implies Liveness}\label{sec:finite-delay:liveness}
We first prove \Cref{thm:liveness-from-finite-delay}, restated below.

\phantomsection\label{proof:liveness-from-finite-delay}
\livenessFromFiniteDelay*
\begin{proof}
  \proofqed[proof:liveness-from-finite-delay]{thm:liveness-from-finite-delay}
  Let~$\alg$ be an algorithm with finite worst-case message delay~$\alpha$.
  Let $S$ be an (infinite) schedule of~$\alg$, and let~$O$ be an operation in~$S$ invoked by a correct process in step~$S[a]$.
  We show that~$O$ returns in~$S$.

  We assign times to the steps of~$S$, obtaining a run in which every received message has latency at most one.
  For each $i\geq 0$, let~$e_i$ be the largest index~$j$ such that step~$S[j]$ receives a message sent in step~$S[i]$;
  if no such message is received, set $e_i=i$.

  Partition the step indices into consecutive finite blocks $B_k=[b_k,b_{k+1})$, for $k\geq 0$, as follows.
  Set $b_0=0$, and choose each $b_{k+1}>b_k$ large enough that $e_i<b_{k+1}$ for every $i<b_k$.
  By construction, if a message sent in block~$B_k$ is later received, then its receive step lies in either block~$B_k$ or~$B_{k+1}$.

  Assign times~$T$ to the steps in~$S$ so that the steps in each block~$B_k$ get strictly increasing times in the interval $[k/2,(k+1)/2)$.
  These intervals are disjoint, ordered by~$k$, and unbounded, so~$T$ is strictly increasing and unbounded, and therefore $R=(S,T)$ is a run of~$\alg$.

  Now consider any message sent in block~$B_k$ and later received.
  By the preceding paragraph, its receive step lies in~$B_k$ or~$B_{k+1}$.
  Its send time is at least~$k/2$, and its receive time is less than $k/2+1$, so its latency is less than~$1$.
  Hence every message that is received has latency at most one in~$R$.

  Every message sent to a process that takes infinitely many steps is received, and so messages are timely during~$O$ in~$R$.
  Since $\alg$ has worst-case message delay~$\alpha$, we have $\latency{O}\leq \alpha < \infty$.
  Hence~$O$ is complete in~$R$, and thus returns in schedule~$S$.
\end{proof}

\subsection{A Finite Worst-Case Message Delay Is an Integer}\label{sec:finite-delay:integer}
Next, we prove \Cref{thm:integer-worst-case-delay}:
that a finite worst-case message delay is always a nonnegative integer.
The key step is the following rounding lemma.

\begin{lemma}\label{thm:floor-worst-case-delay}
  Let $\alg$ be an algorithm.
  If $\latency{\alg} \le \alpha$ for some nonnegative real number~$\alpha$, then $\latency{\alg} \le \lfloor \alpha \rfloor$.
\end{lemma}
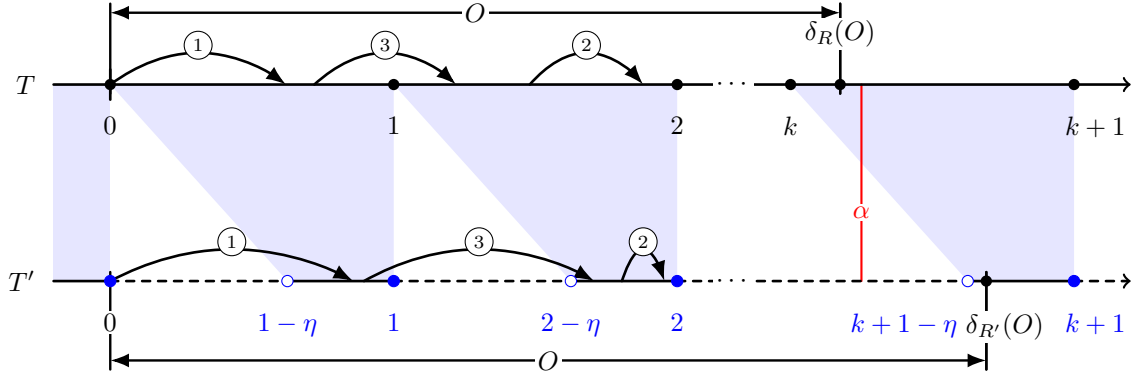
\begin{figure}
    \centering
    \begin{tikzpicture}[scbase]
        \newif\ifdrawintervalrects
        \drawintervalrectsfalse
        \newif\ifdrawtransformarrows
        \drawtransformarrowsfalse

        \def\yTop{2.6}
        \def\yBot{0.0}
        \def\labelAboveOffset{13pt}
        \def\labelBelowOffset{11pt}
        \def\rowLabelXShift{0.2}
        \def\intervalHalfHeight{0.24}
        \def\opTopY{3.55}
        \def\opBotY{-1.05}
        \def\opBarExtend{0.12}
        \def\arcHeight{0.42}
        \tikzset{scmsgcaselabel/.style={circle, draw=black, fill=white, inner sep=0pt, minimum size=10.5pt, font=\scriptsize}}
        \ifdrawtransformarrows
            \tikzset{sctoptimelab/.style={sctlab}}
        \else
            \tikzset{sctoptimelab/.style={sctlab, fill=none}}
        \fi

        \pgfmathsetmacro{\alphaVal}{2.25}
        \pgfmathsetmacro{\xScale}{3.75}
        \pgfmathtruncatemacro{\kVal}{floor(\alphaVal)}
        \pgfmathtruncatemacro{\kponeVal}{\kVal+1}
        \pgfmathsetmacro{\etaVal}{(\kponeVal-\alphaVal)/2}

        \pgfmathsetmacro{\Lratio}{\kVal + 0.7*(\alphaVal-\kVal)}

        \pgfmathsetmacro{\xA}{0.0}
        \pgfmathsetmacro{\xLeft}{-0.2*\xScale}
        \pgfmathsetmacro{\xGapWidth}{0.4*\xScale}

        \pgfmathsetmacro{\xIoneL}{0}
        \pgfmathsetmacro{\xIoneR}{\xScale}
        \pgfmathsetmacro{\xItwoL}{\xIoneR}
        \pgfmathsetmacro{\xItwoR}{2*\xScale}
        \pgfmathsetmacro{\xIlastL}{\xItwoR + \xGapWidth}
        \pgfmathsetmacro{\xIlastR}{\xIlastL + \xScale}
        \pgfmathsetmacro{\xL}{\xIlastL + (\Lratio-\kVal)*\xScale}
        \pgfmathsetmacro{\xAlpha}{\xIlastL + (\alphaVal-\kVal)*\xScale}

        \pgfmathsetmacro{\xJoneL}{\xIoneR - \etaVal*\xScale}
        \pgfmathsetmacro{\xJoneR}{\xScale}
        \pgfmathsetmacro{\xJtwoL}{\xItwoR - \etaVal*\xScale}
        \pgfmathsetmacro{\xJtwoR}{2*\xScale}
        \pgfmathsetmacro{\xKBot}{\xIlastL}
        \pgfmathsetmacro{\xJlastL}{\xIlastR - \etaVal*\xScale}
        \pgfmathsetmacro{\xJlastR}{\xIlastR}
        \pgfmathsetmacro{\xPhiL}{\xJlastL + (\Lratio-\kVal)*\etaVal*\xScale}

        \pgfmathsetmacro{\xTopDots}{0.5*(\xItwoR+\xIlastL)}
        \pgfmathsetmacro{\xBotDots}{\xTopDots}
        \pgfmathsetmacro{\xRight}{max(\xIlastR,\xJlastR) + 0.2*\xScale}

        \pgfmathsetmacro{\xMsgCaseOneSendTop}{\xA}
        \pgfmathsetmacro{\xMsgCaseOneRecvTop}{0.62*\xScale}
        \pgfmathsetmacro{\xMsgCaseTwoSendTop}{\xItwoL + 0.48*\xScale}
        \pgfmathsetmacro{\xMsgCaseTwoRecvTop}{\xItwoL + 0.88*\xScale}
        \pgfmathsetmacro{\xMsgCaseThreeSendTop}{0.72*\xScale}
        \pgfmathsetmacro{\xMsgCaseThreeRecvTop}{\xItwoL + 0.22*\xScale}

        \pgfmathsetmacro{\xMsgCaseOneSendBot}{\xA}
        \pgfmathsetmacro{\xMsgCaseOneRecvBot}{\xJoneL + 0.62*\etaVal*\xScale}
        \pgfmathsetmacro{\xMsgCaseTwoSendBot}{\xJtwoL + 0.48*\etaVal*\xScale}
        \pgfmathsetmacro{\xMsgCaseTwoRecvBot}{\xJtwoL + 0.88*\etaVal*\xScale}
        \pgfmathsetmacro{\xMsgCaseThreeSendBot}{\xJoneL + 0.72*\etaVal*\xScale}
        \pgfmathsetmacro{\xMsgCaseThreeRecvBot}{\xJtwoL + 0.22*\etaVal*\xScale}

        \assertlt{0}{\etaVal}{0 < etaVal}
        \assertlt{\etaVal}{0.5}{etaVal < 1/2}
        \assertlt{\kVal}{\alphaVal}{kVal < alphaVal}
        \assertlt{\alphaVal}{\kponeVal}{alphaVal < k+1}

        \assertlt{\xIlastL}{\xL}{k < latency(O)}
        \assertlt{\xL}{\xAlpha}{latency(O) < alpha}
        \assertlt{\xAlpha}{\xIlastR}{alpha < k+1}

        \assertlt{\xAlpha}{\xJlastL}{alpha < k+1-eta}
        \assertlt{\xJlastL}{\xPhiL}{k+1-eta < phi(latency(O))}
        \assertlt{\xPhiL}{\xJlastR}{phi(latency(O)) < k+1}

        \assertlt{\xLeft}{\xA}{xLeft < xA}
        \assertlt{\xA}{\xIoneR}{xA < xIoneR}
        \assertlt{\xIoneR}{\xItwoR}{xIoneR < xItwoR}
        \assertlt{\xItwoR}{\xTopDots}{xItwoR < xTopDots}
        \assertlt{\xTopDots}{\xIlastL}{xTopDots < xIlastL}
        \assertlt{\xIlastL}{\xL}{xIlastL < xL}
        \assertlt{\xL}{\xAlpha}{xL < xAlpha}
        \assertlt{\xAlpha}{\xIlastR}{xAlpha < xIlastR}
        \assertlt{\xIlastR}{\xRight}{xIlastR < xRight}

        \assertlt{\xA}{\xJoneL}{xA < xJoneL}
        \assertlt{\xJoneL}{\xJoneR}{xJoneL < xJoneR}
        \assertlt{\xJoneR}{\xJtwoL}{xJoneR < xJtwoL}
        \assertlt{\xJtwoL}{\xJtwoR}{xJtwoL < xJtwoR}
        \assertlt{\xJtwoR}{\xBotDots}{xJtwoR < xBotDots}
        \assertlt{\xBotDots}{\xKBot}{xBotDots < xKBot}
        \assertlt{\xKBot}{\xAlpha}{xKBot < xAlpha}
        \assertlt{\xAlpha}{\xJlastL}{xAlpha < xJlastL}
        \assertlt{\xJlastL}{\xPhiL}{xJlastL < xPhiL}
        \assertlt{\xPhiL}{\xJlastR}{xPhiL < xJlastR}
        \assertlt{\xJlastR}{\xRight}{xJlastR < xRight}

        \assertlt{\xMsgCaseOneSendTop}{\xMsgCaseOneRecvTop}{case 1 send < receive}
        \assertlt{\xIoneR}{\xMsgCaseTwoSendTop}{case 2 send after 1}
        \assertlt{\xMsgCaseTwoSendTop}{\xMsgCaseTwoRecvTop}{case 2 send < receive}
        \assertlt{\xMsgCaseTwoRecvTop}{\xItwoR}{case 2 receive before 2}
        \assertlt{\xMsgCaseThreeSendTop}{\xIoneR}{case 3 send before 1}
        \assertlt{\xIoneR}{\xMsgCaseThreeRecvTop}{case 3 receive after 1}

        \ifdrawtransformarrows\else
            \fill[transformshade] (\xLeft,\yTop) -- (\xA,\yTop) -- (\xA,\yBot) -- (\xLeft,\yBot) -- cycle;
            \fill[transformshade] (\xA,\yTop) -- (\xIoneR,\yTop) -- (\xJoneR,\yBot) -- (\xJoneL,\yBot) -- cycle;
            \fill[transformshade] (\xItwoL,\yTop) -- (\xItwoR,\yTop) -- (\xJtwoR,\yBot) -- (\xJtwoL,\yBot) -- cycle;
            \fill[transformshade] (\xIlastL,\yTop) -- (\xIlastR,\yTop) -- (\xJlastR,\yBot) -- (\xJlastL,\yBot) -- cycle;
        \fi

        \draw[red, line width=0.8pt] (\xAlpha,\yTop) -- (\xAlpha,\yBot);
        \node[sctlab, text=red] at (\xAlpha,{0.35*(\yTop+\yBot)}) {$\alpha$};

        \node[sctlab, anchor=east] at (\xLeft-\rowLabelXShift,\yTop) {$T$};
        \ifdrawintervalrects
            \filldraw[fill=gray!18, draw=black!35] (\xIoneL,\yTop-\intervalHalfHeight) rectangle (\xIoneR,\yTop+\intervalHalfHeight);
            \filldraw[fill=gray!18, draw=black!35] (\xItwoL,\yTop-\intervalHalfHeight) rectangle (\xItwoR,\yTop+\intervalHalfHeight);
            \filldraw[fill=gray!18, draw=black!35] (\xIlastL,\yTop-\intervalHalfHeight) rectangle (\xIlastR,\yTop+\intervalHalfHeight);
        \fi
        \draw[sctimelineend] (\xLeft,\yTop) -- (\xRight,\yTop);
        \path (\xA,\yTop) pic {scopabove={\xL-\xA/\opTopY-\yTop/O}};
        \draw[scopmark] (\xL,\yTop-\opBarExtend) -- (\xL,\opTopY+\opBarExtend);

        \path (\xMsgCaseOneSendTop,\yTop) pic {scgoodmsgarc={\xMsgCaseOneRecvTop-\xMsgCaseOneSendTop}};
        \path (\xMsgCaseTwoSendTop,\yTop) pic {scgoodmsgarc={\xMsgCaseTwoRecvTop-\xMsgCaseTwoSendTop}};
        \path (\xMsgCaseThreeSendTop,\yTop) pic {scgoodmsgarc={\xMsgCaseThreeRecvTop-\xMsgCaseThreeSendTop}};
        \node[scmsgcaselabel] at ({0.5*(\xMsgCaseOneSendTop+\xMsgCaseOneRecvTop)}, {\yTop+\arcHeight+0.10}) {1};
        \node[scmsgcaselabel] at ({0.5*(\xMsgCaseTwoSendTop+\xMsgCaseTwoRecvTop)}, {\yTop+\arcHeight+0.10}) {2};
        \node[scmsgcaselabel] at ({0.5*(\xMsgCaseThreeSendTop+\xMsgCaseThreeRecvTop)}, {\yTop+\arcHeight+0.10}) {3};

        \node[sctlab] at (\xTopDots,\yTop) {$\cdots$};

        \path (\xA,\yTop) pic {scendpointtick};
        \node[scstepdot] at (\xA,\yTop) {};
        \node[scstepdot] at (\xIoneR,\yTop) {};
        \node[scstepdot] at (\xItwoR,\yTop) {};
        \node[scstepdot] at (\xIlastL,\yTop) {};
        \node[scstepdot] at (\xIlastR,\yTop) {};

        \node[scstepdot] at (\xL,\yTop) {};
        \node[sctlab, above=\labelAboveOffset] at (\xL,\yTop) {$\latency[R]{O}$};

        \ifdrawtransformarrows
            \draw[scconnector] (\xA,\yTop) -- (\xA,\yBot);
            \draw[scconnector] (\xA,\yTop) -- (\xJoneL,\yBot);
            \draw[scconnector] (\xIoneR,\yTop) -- (\xJoneR,\yBot);
            \draw[scconnector] (\xIoneR,\yTop) -- (\xJtwoL,\yBot);
            \draw[scconnector] (\xItwoR,\yTop) -- (\xJtwoR,\yBot);
            \draw[scconnector] (\xIlastL,\yTop) -- (\xKBot,\yBot);
            \draw[scconnector] (\xIlastL,\yTop) -- (\xJlastL,\yBot);
            \draw[scconnector] (\xIlastR,\yTop) -- (\xJlastR,\yBot);
        \fi

        \node[sctoptimelab, below=\labelBelowOffset] at (\xA,\yTop) {$0$};
        \node[sctoptimelab, below=\labelBelowOffset] at (\xIoneR,\yTop) {$1$};
        \node[sctoptimelab, below=\labelBelowOffset] at (\xItwoR,\yTop) {$2$};
        \node[sctoptimelab, below=\labelBelowOffset] at (\xIlastL,\yTop) {$k$};
        \node[sctoptimelab, anchor=west, xshift=8pt, below=\labelBelowOffset] at (\xIlastR,\yTop) {$k+1$};

        \node[sctlab, anchor=east] at (\xLeft-\rowLabelXShift,\yBot) {$T'$};
        \ifdrawintervalrects
            \filldraw[fill=blue!12, draw=blue] (\xJoneL,\yBot-\intervalHalfHeight) rectangle (\xJoneR,\yBot+\intervalHalfHeight);
            \filldraw[fill=blue!12, draw=blue] (\xJtwoL,\yBot-\intervalHalfHeight) rectangle (\xJtwoR,\yBot+\intervalHalfHeight);
            \filldraw[fill=blue!12, draw=blue] (\xJlastL,\yBot-\intervalHalfHeight) rectangle (\xJlastR,\yBot+\intervalHalfHeight);
        \fi
        \draw[sctimeline] (\xLeft,\yBot) -- (\xA,\yBot);
        \draw[sctimeline, dashed] (\xA,\yBot) -- (\xJoneL,\yBot);
        \draw[sctimeline] (\xJoneL,\yBot) -- (\xJoneR,\yBot);
        \draw[sctimeline, dashed] (\xJoneR,\yBot) -- (\xJtwoL,\yBot);
        \draw[sctimeline] (\xJtwoL,\yBot) -- (\xJtwoR,\yBot);
        \draw[sctimeline, dashed] (\xJtwoR,\yBot) -- (\xJlastL,\yBot);
        \draw[sctimeline] (\xJlastL,\yBot) -- (\xJlastR,\yBot);
        \draw[sctimelineend, dashed] (\xJlastR,\yBot) -- (\xRight,\yBot);
        \path (\xA,\yBot) pic {scopbelow={\xPhiL-\xA/\opBotY-\yBot/O}};
        \draw[scopmark] (\xPhiL,\opBotY-\opBarExtend) -- (\xPhiL,\yBot+\opBarExtend);

        \path (\xMsgCaseOneSendBot,\yBot) pic {scgoodmsgarc={\xMsgCaseOneRecvBot-\xMsgCaseOneSendBot}};
        \path (\xMsgCaseTwoSendBot,\yBot) pic {scgoodmsgarc={\xMsgCaseTwoRecvBot-\xMsgCaseTwoSendBot}};
        \path (\xMsgCaseThreeSendBot,\yBot) pic {scgoodmsgarc={\xMsgCaseThreeRecvBot-\xMsgCaseThreeSendBot}};
        \node[scmsgcaselabel] at ({0.5*(\xMsgCaseOneSendBot+\xMsgCaseOneRecvBot)}, {\yBot+\arcHeight+0.10}) {1};
        \node[scmsgcaselabel] at ({0.5*(\xMsgCaseTwoSendBot+\xMsgCaseTwoRecvBot)}, {\yBot+\arcHeight+0.10}) {2};
        \node[scmsgcaselabel] at ({0.5*(\xMsgCaseThreeSendBot+\xMsgCaseThreeRecvBot)}, {\yBot+\arcHeight+0.10}) {3};

        \node[sctlab] at (\xBotDots,\yBot) {$\cdots$};

        \node[circle, draw=blue, fill=blue, inner sep=1.5pt] at (\xA,\yBot) {};
        \node[circle, draw=blue, fill=white, inner sep=1.5pt] at (\xJoneL,\yBot) {};
        \node[circle, draw=blue, fill=blue, inner sep=1.5pt] at (\xJoneR,\yBot) {};
        \node[circle, draw=blue, fill=white, inner sep=1.5pt] at (\xJtwoL,\yBot) {};
        \node[circle, draw=blue, fill=blue, inner sep=1.5pt] at (\xJtwoR,\yBot) {};
        \ifdrawtransformarrows
            \node[circle, draw=blue, fill=blue, inner sep=1.5pt] at (\xKBot,\yBot) {};
        \fi
        \node[circle, draw=blue, fill=white, inner sep=1.5pt] at (\xJlastL,\yBot) {};
        \node[circle, draw=blue, fill=blue, inner sep=1.5pt] at (\xJlastR,\yBot) {};

        \node[scstepdot] at (\xPhiL,\yBot) {};
        \node[sctlab, below=\labelBelowOffset, xshift=7pt] at (\xPhiL,\yBot) {$\latency[R']{O}$};

        \node[sctlab, below=\labelBelowOffset] at (\xA,\yBot) {$0$};
        \node[sctlab, below=\labelBelowOffset] at (\xJoneL,\yBot) {\textcolor{blue}{$1-\eta$}};
        \node[sctlab, below=\labelBelowOffset] at (\xJoneR,\yBot) {\textcolor{blue}{$1$}};
        \node[sctlab, below=\labelBelowOffset] at (\xJtwoL,\yBot) {\textcolor{blue}{$2-\eta$}};
        \node[sctlab, below=\labelBelowOffset] at (\xJtwoR,\yBot) {\textcolor{blue}{$2$}};
        \ifdrawtransformarrows
            \node[sctlab, below=\labelBelowOffset] at (\xKBot,\yBot) {\textcolor{blue}{$k$}};
        \fi
        \node[sctlab, anchor=east, xshift=-24pt, below=\labelBelowOffset] at (\xJlastL,\yBot) {\textcolor{blue}{$k+1-\eta$}};
        \node[sctlab, anchor=west, xshift=8pt, below=\labelBelowOffset] at (\xJlastR,\yBot) {\textcolor{blue}{$k+1$}};
    \end{tikzpicture}
    \caption{
    \textbf{Proof of~\Cref{thm:floor-worst-case-delay}:} retiming the run~$R=(S,T)$ (top) into the run~$R'=(S,T')$ (bottom) when $k<\alpha<k+1$.
    For each integer~$i\geq 1$, the unit interval $(i-1,i]$ in~$T$ is \textcolor{blue}{compressed} into $(i-\eta,i]$ in~$T'$, where $\eta=(k+1-\alpha)/2\in(0,1/2)$.
    In~$T'$, ``\textcolor{blue}{$\circ$}'' marks open interval endpoints, ``\textcolor{blue}{$\bullet$}'' marks closed interval endpoints, and dashed segments mark times at which no step is taken.
    Every message~$m$ sent during~$O$ in~$R$ with $\latency[R]{m}\leq1$ still satisfies $\latency[R']{m}\leq1$ (as illustrated by cases~\circcase{1},~\circcase{2}, and~\circcase{3}).
    The operation satisfies $\latency[R]{O}\leq\alpha$, but this choice of~$\eta$ ensures that $\latency[R']{O}=\phi(\latency[R]{O})>\alpha$.
    }\label{fig:floor-worst-case-delay}
\end{figure}

\begin{proof}
  \proofqed{thm:floor-worst-case-delay}
  Let $k\deq \lfloor \alpha \rfloor$.
  If $\alpha = k$, the claim is immediate, so assume $k < \alpha < k+1$.

  Suppose, for the sake of contradiction, that $\latency{\alg} > k$.
  Then there exist a run~$R=(S,T)$ of~$\alg$, and an operation~$O$ invoked by a correct process, such that messages are timely during~$O$ in~$R$ and $\latency[R]{O} > k$.
  By the definition of worst-case message delay, $\latency[R]{O}\leq \latency{\alg}\leq \alpha<\infty$; hence~$O$ is complete.
  Let~$O$ be invoked in step~$S[a]$ and return in step~$S[b]$, so $\latency[R]{O}=T[b]-T[a]$.
  Putting the inequalities together, we have
  \begin{equation}\label{eq:floor-worst-case-delay-inequality}
    k < \latency[R]{O} \le \alpha < k+1.
  \end{equation}

  By subtracting~$T[a]$ from every entry of~$T$, we may assume without loss of generality that $T[a]=0$.
  Then $T[b]=\latency[R]{O}$.

  We will now retime~$R$ to obtain another run~$R'$ with schedule~$S$ such that $\latency[R']{O}>\alpha$; see \Cref{fig:floor-worst-case-delay}.
  Set
  \[
    \eta\Deq\frac{k+1-\alpha}{2}.
  \]
  Since $k<\alpha<k+1$, we have $0<\eta<1/2$.
  We construct a time mapping~$\phi:\R\to\R$ that leaves times at or before~$0$ unchanged and, for each integer $i\geq 1$, compresses the unit interval $(i-1,i]$ affinely onto the shorter interval $(i-\eta,i]$:
  \[
      \phi(t)\Deq
      \begin{cases}
          t, & t\le 0,\\
          \Aff_{(i-1, i]\to (i-\eta, i]}(t), & \text{$i-1 < t \leq i$ for some $i\geq 1$}.
      \end{cases}
  \]
  Then $\phi$ is strictly increasing and unbounded.

  Define the infinite sequence~$T'$ by $T'[i]\deq\phi(T[i])$, and let $R'\deq(S,T')$.
  Since $\phi$~is strictly increasing and unbounded, $R'$ is a run of~$\alg$; moreover, it has the same correct processes as~$R$.
  By construction,
  \[
    \interval[R']{O}
    = [T'[a], T'[b]]
    = [\phi(T[a]), \phi(T[b])]
    = [0, \phi(\latency[R]{O})].
  \]

  \begin{claim}\label{claim:floor-worst-case-delay-timely}
    Messages are timely during~$O$ in~$R'$.
  \end{claim}
  \begin{proof}
    \proofqed{claim:floor-worst-case-delay-timely}
    Fix an arbitrary message~$m$ sent during~$O$ in~$R'$ to a process that takes infinitely many steps.
    Let $s<t$ denote, respectively, the send and receive times of~$m$ in~$R$.
    Since $T'=\phi\circ T$, its send and receive times in~$R'$ are $\phi(s)$ and~$\phi(t)$, respectively.

    Thus, $\phi(0)\leq \phi(s)\leq \phi(\latency[R]{O})$.
    Because $\phi$ is strictly increasing, this implies $0\leq s\leq \latency[R]{O}$, so $m$~is sent during~$O$ in~$R$.
    Since messages are timely during~$O$ in~$R$, $\latency[R]{m}=t-s\leq 1$.

    It remains to show that $\latency[R']{m}=\phi(t)-\phi(s)\leq 1$.
    Given that $0 < t-s \leq 1$, there are three cases for the positions of $s$ and $t$, exemplified by messages labeled~\circcase{1},~\circcase{2}, and~\circcase{3} in \Cref{fig:floor-worst-case-delay}:
    \begin{enumerate}
      \item If $s=0$, then $0<t\le 1$, and we get $\phi(t)-\phi(s) = \phi(t) \leq \phi(1) = 1$.
      \item If $s, t\in (i-1, i]$ for some $i\geq 1$, then since $\phi$ is affine on $(i-1,i]$ with slope~$\eta$,
        \[
          \phi(t)-\phi(s)=\eta(t-s)\leq \eta<1.
        \]
      \item Otherwise, $s\in(i-1,i]$ and $t\in(i,i+1]$ for some $i\geq 1$. Then
        \[
          \phi(t)-\phi(s)
          = \bigl[(i+1-\eta)+\eta(t-i)\bigr]-\bigl[(i-\eta)+\eta(s-(i-1))\bigr]
          = 1+\eta(t-s-1)
          \leq 1.
        \]
    \end{enumerate}
    Thus, $\latency[R']{m}\leq 1$ in every case.
    Since $m$ was arbitrary, messages are timely during~$O$ in~$R'$.
  \end{proof}

  \begin{claim}\label{claim:floor-worst-case-delay-latency}
    $\latency[R']{O}>\alpha$.
  \end{claim}
  \begin{proof}
    \proofqed{claim:floor-worst-case-delay-latency}
    By~\eqref{eq:floor-worst-case-delay-inequality}, $\latency[R]{O}\in(k,k+1]$.
    Since~$\phi$ maps $(k,k+1]$ onto $(k+1-\eta,k+1]$, it follows that
    \[
      \phi(\latency[R]{O})>k+1-\eta.
    \]
    Therefore,
    \[
      \latency[R']{O}
      =\phi(\latency[R]{O})-\phi(0)
      >k+1-\eta
      =\frac{k+1+\alpha}{2}
      >\alpha.\qedhere
    \]
  \end{proof}

  By the first claim and the definition of worst-case message delay,
  $\latency[R']{O}\leq\latency{\alg}\leq\alpha$, contradicting the second claim.
  Hence, $\latency{\alg}\leq k=\lfloor \alpha \rfloor$.
\end{proof}

\Cref{thm:integer-worst-case-delay}, which we restate below, immediately follows from this lemma.

\phantomsection\label{proof:integer-worst-case-delay}
\integerWorstCaseDelay*
\begin{proof}
  \proofqed[proof:integer-worst-case-delay]{thm:integer-worst-case-delay}
  Let the algorithm's worst-case message delay be $\alpha<\infty$.
  Since operation latencies are nonnegative, we have $\alpha \ge 0$.
  By \Cref{thm:floor-worst-case-delay}, the algorithm's worst-case message delay is at most~$\lfloor \alpha \rfloor$, so $\alpha \le \lfloor \alpha \rfloor$.
  Since always $\lfloor \alpha \rfloor \le \alpha$, it follows that $\alpha = \lfloor \alpha \rfloor$.
  Thus~$\alpha$ is a nonnegative integer.
\end{proof}

\section{Auxiliary Results for Latency Lower Bounds}\label{sec:lower-bound-ingredients}
\subsection{Proofs of the Loquacious-Transformation Lemmas}\label{sec:loquacious-transformation-proofs}

We prove the two loquacious-transformation lemmas in turn, beginning with crash failures and then adapting the argument to Byzantine failures.

\subsubsection{Crash Failures}\label{sec:loquacious-transformation-proofs:crash}

We first prove \Cref{thm:loquacious-transformation} from \Cref{sec:crash}, restated below for convenience.

\phantomsection\label{proof:loquacious-transformation}
\loquaciousTransformation*
\begin{proof}
  \proofqed[proof:loquacious-transformation]{thm:loquacious-transformation}
  We construct an algorithm~$\alg'$ in which each process simulates its counterpart in~$\alg$ but broadcasts one \emph{packed message} in every step.
  A packed message encodes a finite, possibly empty set of pairs~$(m,q)$, where $m$~is an \emph{original message} of~$\alg$ and $q$~is its destination.

  For each process~$p$, we define the automaton~$\alg'_p$ that~$\alg'$ assigns to~$p$ in terms of the automaton~$\alg_p$ that~$\alg$ assigns to~$p$.
  The local states and initial state of~$\alg'_p$ are those of~$\alg_p$.
  From every state~$s$, set~$M'$ of received messages, and operation invocation (if any), $\alg'_p$ has a unique transition, which:
  \begin{itemize}
    \item extracts from the packed messages in~$M'$ the set~$M$ of original messages paired with destination~$p$;
    \item follows the unique transition~$t$ of~$\alg_p$ from state~$s$ on~$M$ and the same operation invocation, if any, moving to the state that~$t$ moves to;
    \item broadcasts the packed message containing the pair~$(m,q)$ for every original message~$m$ that~$t$ sends to a process~$q$; and
    \item returns the same value as~$t$, if~$t$ returns from an operation.
  \end{itemize}
  By construction, $\alg'$ is loquacious.

  Given a schedule~$S'$ of~$\alg'$, define an infinite sequence~$S$ step by step.
  For each~$i$, let~$p$ be the process that takes step~$S'[i]$.
  Step~$S[i]$ is taken by~$p$ and is defined as follows:
  \begin{itemize}
    \item $S[i]$ invokes an operation iff $S'[i]$ does;
    \item for every process~$q$, step~$S[i]$ receives from~$q$ precisely those original messages~$m$ for which~$(m,p)$ appears in a packed message received from~$q$ in~$S'[i]$;
    \item for every process~$q$, step~$S[i]$ sends to~$q$ precisely those original messages~$m$ for which~$(m,q)$ appears in the packed message broadcast in~$S'[i]$; and
    \item $S[i]$ returns from an operation with value~$v$ iff $S'[i]$ does with the same value.
  \end{itemize}

  We now isolate the properties of~$S$ that we will later need.
  \begin{claim}\label{claim:loquacious-transformation-correspondence}
    Let $S$ be defined as above from $S'$. Then:
    \begin{enumerate}[label=(\alph*),ref=(\alph*)]
      \item\label{item:loquacious-transformation-operation} For every~$i$, steps~$S[i]$ and~$S'[i]$ are taken by the same process and have the same invocation and return.
      \item\label{item:loquacious-transformation-steps} The same processes take infinitely many steps in~$S$ as in~$S'$.
      \item\label{item:loquacious-transformation-correct} Every process correct in~$S'$ is correct in~$S$.%
        \footnote{Of course, in the crash-failure model, a process is correct exactly when it takes infinitely many steps. We nevertheless state this separately from the preceding property so that both remain true in the Byzantine model, where a process may take infinitely many steps yet be faulty.}
      \item\label{item:loquacious-transformation-message} If a message sent in~$S[i]$ is received in~$S[j]$, then some message sent in~$S'[i]$ is received in~$S'[j]$.
      \item\label{item:loquacious-transformation-schedule} $S$ is a schedule of~$\alg$.
    \end{enumerate}
  \end{claim}
  \begin{proof}
    \proofqed{claim:loquacious-transformation-correspondence}
    An induction on~$i$ shows that every step~$S[i]$ follows~$\alg$: It performs precisely the transition of~$\alg$ that~$S'[i]$ simulates, from the same local state.
    Moreover, corresponding steps of~$S$ and~$S'$ are taken by the same process and have the same invocation and return.
    Thus \ref{item:loquacious-transformation-operation} and~\ref{item:loquacious-transformation-steps} hold.
    Since every step of~$S$ follows~$\alg$, \ref{item:loquacious-transformation-steps} implies \ref{item:loquacious-transformation-correct}.

    For every original message~$m$ sent to~$q$ in~$S[i]$, let~$m'$ be the copy of the packed message broadcast in~$S'[i]$ that is sent to~$q$.
    By construction, $m'$ contains the pair~$(m,q)$, and $m$ is received in~$S[j]$ iff $m'$ is received in~$S'[j]$, proving \ref{item:loquacious-transformation-message}.
    Together with \ref{item:loquacious-transformation-steps} and the fact that~$S'$ is a schedule, this correspondence implies that in~$S$, every received message was sent earlier, no message is received more than once, and every message sent to a process that takes infinitely many steps is eventually received.
    Moreover, \ref{item:loquacious-transformation-steps} implies that at least $n-f$ processes take infinitely many steps in~$S$, since they do in~$S'$.
    By~\ref{item:loquacious-transformation-operation}, each process's operations remain sequential.
    Thus $S$ is a schedule of~$\alg$, proving \ref{item:loquacious-transformation-schedule}.
  \end{proof}

  \begin{claim}\label{claim:loquacious-transformation-implementation}
    $\alg'$ is an implementation of~\ts{}.
  \end{claim}
  \begin{proof}
    \proofqed{claim:loquacious-transformation-implementation}
    Let $S'$ be any schedule of~$\alg'$ and let $S$ be as defined above.
    By \Crefpart{claim:loquacious-transformation-correspondence}{item:loquacious-transformation-schedule}, $S$ is a schedule of~$\alg$.

    By \Crefpart{claim:loquacious-transformation-correspondence}{item:loquacious-transformation-operation}, $S$ and~$S'$ have the same history.
    Since $\alg$ satisfies the safety condition~(\ref{cond:timestamp-safety}), this history is linearizable with respect to~\ts{}, so $\alg'$ satisfies~(\ref{cond:timestamp-safety}) as well.

    To show that $\alg'$ satisfies the liveness condition~(\ref{cond:timestamp-liveness}), let $O'$ be any operation in~$S'$ invoked by a correct process~$p$.
    By \Crefpart{claim:loquacious-transformation-correspondence}{item:loquacious-transformation-operation}, $S$ has a corresponding operation~$O$ invoked by the same process~$p$.
    By \Crefpart{claim:loquacious-transformation-correspondence}{item:loquacious-transformation-correct}, $p$ is correct in~$S$ as well.
    Since $S$ is a schedule of~$\alg$, which implements~\ts{}, $O$ eventually returns.
    By \Crefpart{claim:loquacious-transformation-correspondence}{item:loquacious-transformation-operation}, the corresponding return occurs in~$S'$, so $O'$ eventually returns as well.
  \end{proof}

  \begin{claim}\label{claim:loquacious-transformation-latency}
    $\latency{\alg'}\leq\latency{\alg}$.
  \end{claim}
  \begin{proof}
    \proofqed{claim:loquacious-transformation-latency}
    Let $R'=(S',T')$ be any run of~$\alg'$, and let $R=(S,T')$, where~$S$ is as defined above.
    We first verify that $R$ is a run of~$\alg$:
    \begin{itemize}
      \item By \Crefpart{claim:loquacious-transformation-correspondence}{item:loquacious-transformation-schedule}, $S$ is a schedule of~$\alg$.
      \item Since $R'$ is a run, $T'$ is non-decreasing and unbounded.
      \item If a message is sent in~$S[i]$ and received in~$S[j]$, then \Crefpart{claim:loquacious-transformation-correspondence}{item:loquacious-transformation-message} gives a message sent in~$S'[i]$ and received in~$S'[j]$; hence $T'[i]<T'[j]$.
      \item Likewise, if $S[i]$ and~$S[j]$, with $i<j$, are taken by the same process, then \Crefpart{claim:loquacious-transformation-correspondence}{item:loquacious-transformation-operation} implies that $S'[i]$ and~$S'[j]$ are taken by the same process; hence again $T'[i]<T'[j]$.
    \end{itemize}
    Thus all conditions of the definition of a run are satisfied.

    Let $O'$ be any operation in~$R'$ invoked by a correct process~$p$ such that messages are timely during~$O'$.
    By \Crefpart{claim:loquacious-transformation-correspondence}{item:loquacious-transformation-operation}, $S$ has a corresponding operation~$O$ invoked by~$p$, with the same invocation- and return-step indices.
    Since $R$ and $R'$ share the time sequence~$T'$,
    \[
      \interval{O}=\interval{O'}
      \quad\text{and hence}\quad
      \latency{O}=\latency{O'}.
    \]
    By \Crefpart{claim:loquacious-transformation-correspondence}{item:loquacious-transformation-correct}, $p$ is correct in~$R$ as well.

    We now show that messages are timely during~$O$ in~$R$.
    Let~$m$ be any message sent during~$O$ to a process~$q$ that takes infinitely many steps in~$S$, so $m$ is eventually received.
    Suppose it is sent in step~$S[i]$ and received in step~$S[j]$.
    By \Crefpart{claim:loquacious-transformation-correspondence}{item:loquacious-transformation-message}, some message~$m'$ is sent in~$S'[i]$ and received in~$S'[j]$, and so
    \[
      \latency{m'}=T'[j]-T'[i]=\latency{m}.
    \]
    By \Crefpart{claim:loquacious-transformation-correspondence}{item:loquacious-transformation-operation}, step~$S'[j]$ is taken by~$q$, so $m'$ is sent to~$q$.
    By \Crefpart{claim:loquacious-transformation-correspondence}{item:loquacious-transformation-steps}, $q$ takes infinitely many steps in~$S'$.
    Moreover, because $m$ and~$m'$ are sent at the same time and $\interval{O}=\interval{O'}$, message~$m'$ is sent during~$O'$.
    Because messages are timely during~$O'$ in~$R'$, we have
    \[
      \latency{m}=\latency{m'}\leq1.
    \]
    Thus every message sent during~$O$ to a process that takes infinitely many steps has latency at most one, so messages are timely during~$O$ in~$R$.

    Since $O$ is invoked by a correct process, the definition of $\latency{\alg}$ yields
    \[
      \latency{O'}
      =\latency{O}
      \leq\latency{\alg}.
    \]
    Taking the supremum over all such $R'$ and~$O'$ yields
    $\latency{\alg'}\leq\latency{\alg}$.
  \end{proof}

  The construction and the latter two claims establish the result.
\end{proof}

\subsubsection{Byzantine Failures}\label{sec:loquacious-transformation-proofs:byz}

We next restate \Cref{thm:loquacious-transformation-byz} from \Cref{sec:byz:auxiliary-lemmas} and sketch its proof.

\phantomsection\label{proof:loquacious-transformation-byz}
\loquaciousTransformationByz*
\begin{proof}[Proof sketch]
  \proofqed[proof:loquacious-transformation-byz]{thm:loquacious-transformation-byz}
  We reuse the construction of~$\alg'$ and the transformation of~$S'$ into~$S$ from the proof of the crash-failure counterpart, \nameCref{thm:loquacious-transformation}~\hyperref[proof:loquacious-transformation]{\ref*{thm:loquacious-transformation}}, modifying only how the messages sent in~$S[i]$ are defined:
  \begin{quote}
    For every process~$q$, step~$S[i]$ sends to~$q$ precisely those original messages~$m$ for which~$(m,q)$ appears in \uline{a packed message sent to~$q$} in~$S'[i]$.
  \end{quote}
  This modification is necessary because the crash-failure transformation refers to \emph{the packed message broadcast} in~$S'[i]$, whereas a Byzantine process may send different messages to different recipients or send no message at all.

  Because the construction of~$\alg'$ is unchanged, the algorithm is loquacious as before.
  Moreover, under the modified transformation, \Cref{claim:loquacious-transformation-correspondence} continues to hold under Byzantine failures.
  We omit the details of the argument, but two points merit mention:
  \begin{itemize}
    \item Malformed messages in~$S'$ are simply ignored, as desired.
      This is because both the construction of~$\alg'$ and the transformation of~$S'$ into~$S$ consider only (valid) packed messages.
    \item A process that is faulty in~$S'$ and takes infinitely many steps may nevertheless be correct in~$S$, if all of its transformed steps happen to conform to~$\alg$.
      This is harmless: \Cref{claim:loquacious-transformation-correspondence} requires only that every process correct in~$S'$ also be correct in~$S$.
      The subsequent arguments use this implication only in that direction.
  \end{itemize}

  Let $S'$ be any schedule of~$\alg'$, and let $S$ be the sequence obtained by the modified transformation.
  Since $\alg$ satisfies the safety condition~(\ref{cond:byz-timestamp-safety}), there is a history~$\widehat H$ that is linearizable with respect to~\ts{} and agrees with the history of~$S$ on the processes correct in~$S$.
  By \Cref{claim:loquacious-transformation-correspondence}, $S$ and~$S'$ have the same history, and every process correct in~$S'$ is correct in~$S$.
  Hence $\widehat H$ agrees with the history of~$S'$ on the processes correct in~$S'$, so $\alg'$ satisfies the safety condition~(\ref{cond:byz-timestamp-safety}).
  The liveness and latency arguments carry over unchanged, using~(\ref{cond:byz-timestamp-liveness}) for the liveness condition.

\end{proof}

\subsection{Weaker Timeliness Requirements}
The definitions and results in this subsection apply in both the crash-failure and Byzantine settings, with ``correct'' interpreted in the respective sense.

Recall from \Cref{def:worst-case-message-delay} that an algorithm's worst-case message delay bounds the latency of an operation if messages are timely (i.e., have latency at most one) \emph{throughout} the operation.
This bound should continue to hold under a weaker requirement that exempts messages sent within one unit of time before the operation returns.
Indeed, if such a message is untimely, it arrives only after the operation has returned and cannot affect the timing of the return step.

To state this requirement more conveniently, we extend timeliness from operations to arbitrary time intervals:
Given a run~$R=(S,T)$ and an interval $I\subseteq\R$, we say that \emph{messages are timely during~$I$ in~$R$} if $\latency{m}\leq1$ for every message~$m$ sent at a time in~$I$ to a process that takes infinitely many steps in~$S$.
Timeliness during an operation~$O$ (\Cref{def:message-timeliness}) is the special case $I=\interval[R]{O}$.

\Cref{def:ultimately-timely} formalizes the weaker requirement described above, called \emph{ultimate timeliness}; \Cref{thm:ultimately-timely} shows that the message-delay bound still applies.

\begin{definition}[Ultimate timeliness]\label{def:ultimately-timely}
  Let $R=(S, T)$ be a run, and let $O$ be an operation in~$R$.
  We say that \emph{messages are ultimately timely during~$O$ in~$R$} if either of the following applies:
  \begin{enumerate}[label=(\arabic*), ref=(\arabic*)]
    \item\label{case:ultimately-timely-incomplete}
      $O$ is incomplete, and messages are timely during~$O$ in~$R$; or
    \item $O$ is complete, having been invoked in step~$S[a]$ and returned in step~$S[b]$, and messages are timely during~$[T[a], T[b]-1]$ in~$R$.
  \end{enumerate}
\end{definition}

Note that if $\latency{O} < 1$, then messages are vacuously ultimately timely during~$O$.

\begin{lemma}\label{thm:ultimately-timely}
  Let $\alg$ be an algorithm, let $R$ be a run of~$\alg$, and let $O$ be an operation in~$R$ invoked by a correct process.
  If messages are \emph{ultimately timely} during~$O$ in~$R$, then $\latency{O}\leq \latency{\alg}$.
\end{lemma}
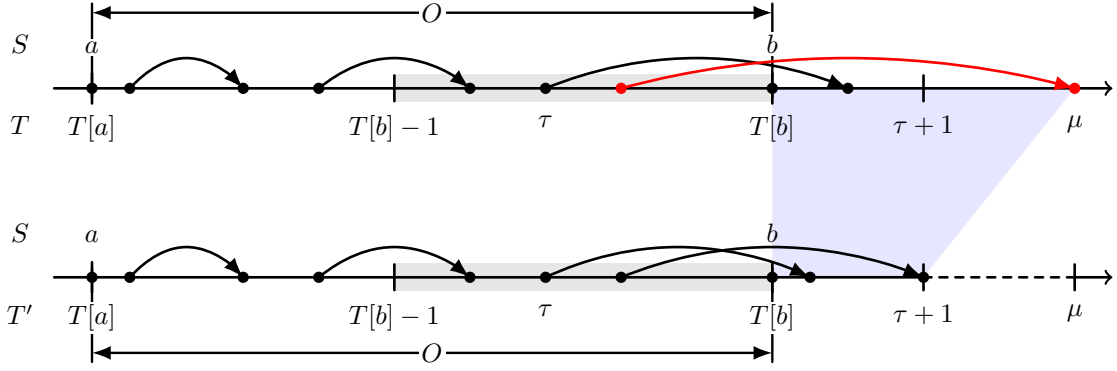
\begin{figure}
    \centering
    \begin{tikzpicture}[scbase]
    \newif\ifdrawtransformarrows
    \drawtransformarrowsfalse

    \def\yTop{2.5}
    \def\yBot{0.0}
    \def\labelAboveOffset{12pt}
    \def\labelBelowOffset{9pt}
    \def\rowLabelXShift{0.45}
    \def\opTopY{3.5}
    \def\opBotY{-1.0}
    \def\opBarExtend{0.12}
    \def\arcHeight{0.4}
    \ifdrawtransformarrows
    \tikzset{sctranslab/.style={sctlab}}
    \else
    \tikzset{sctranslab/.style={sctlab, fill=none}}
    \fi
    
    \def\xLeft{-.5}
    \def\xRight{13.5}
    
    \def\xA{0.0}
    \def\xOne{.5}
    \def\xTwo{2}
    \def\xThree{3}
    \def\xFour{5}
    \def\xTau{6}
    \def\xFive{7}
    \def\xB{9}
    \def\xSixTopArrive{10}
    \def\xDeltaWidth{5} %
    \pgfmathsetmacro{\xTbmd}{\xB-\xDeltaWidth}      %
    \pgfmathsetmacro{\xTauPlusDelta}{\xTau+\xDeltaWidth}  %
    
    \def\xMuTop{13}
    \pgfmathsetmacro{\xPhiEps}{(\xTauPlusDelta-\xB)/(\xMuTop-\xB)}
    \pgfmathsetmacro{\xSixBotArrive}{\xB + (\xSixTopArrive-\xB)*\xPhiEps}

    \assertlt{\xA}{\xOne}{xA < xOne}
    \assertlt{\xOne}{\xTwo}{xOne < xTwo}
    \assertlt{\xTwo}{\xThree}{xTwo < xThree}
    \assertlt{\xThree}{\xTbmd}{xThree < xTbmd}
    \assertlt{\xTbmd}{\xFour}{xTbmd < xFour}
    \assertlt{\xFour}{\xTau}{xFour < xTau}
    \assertlt{\xTau}{\xFive}{xTau < xFive}
    \assertlt{\xFive}{\xB}{xFive < xB}
    \assertlt{\xB}{\xSixTopArrive}{xB < xSixTopArrive}
    \assertlt{\xSixTopArrive}{\xTauPlusDelta}{xSixTopArrive < xTauPlusDelta}
    \assertlt{\xTauPlusDelta}{\xMuTop}{xTauPlusDelta < xMuTop}

    \assertlt{\xB}{\xSixBotArrive}{xB < xSixBotArrive}
    \assertlt{\xSixBotArrive}{\xTauPlusDelta}{xSixBotArrive < xTauPlusDelta}

    \assertle{\xTwo}{\xOne + \xDeltaWidth}{xTwo <= xOne + xDeltaWidth}
    \assertle{\xFour}{\xThree + \xDeltaWidth}{xFour <= xThree + xDeltaWidth}
    \assertle{\xSixTopArrive}{\xTau + \xDeltaWidth}{xSixTopArrive <= xTau + xDeltaWidth}
    \assertlt{\xDeltaWidth}{\xMuTop - \xFive}{xDeltaWidth < xMuTop - xFive}

    \assertle{\xSixBotArrive}{\xTau + \xDeltaWidth}{xSixBotArrive <= xTau + xDeltaWidth}
    \assertle{\xTauPlusDelta}{\xTau + \xDeltaWidth}{xTauPlusDelta <= xTau + xDeltaWidth}

    \ifdrawtransformarrows\else
    \fill[transformshade] (\xB,\yTop) -- (\xMuTop,\yTop) -- (\xTauPlusDelta,\yBot) -- (\xB,\yBot) -- cycle;
    \fi

    \node[sctlab, anchor=east, above=\labelAboveOffset] at (\xLeft-\rowLabelXShift,\yTop) {$S$};
    \node[sctranslab, anchor=east, below=\labelBelowOffset] at (\xLeft-\rowLabelXShift,\yTop) {$T$};
    
    \fill[scwindow] (\xTbmd,\yTop-0.18) rectangle (\xB,\yTop+0.18);
    \draw[sctimeline] (\xLeft,\yTop) -- (\xB,\yTop);
    \draw[sctimeline] (\xB,\yTop) -- (\xMuTop,\yTop);
    \draw[sctimelineend] (\xMuTop,\yTop) -- (\xRight,\yTop);
    \path (\xA,\yTop) pic {scopabove={\xB-\xA/\opTopY-\yTop/O}};
    
    \path (\xA,\yTop) pic {scendpointtick};
    \path (\xTbmd,\yTop) pic {screftick};
    \path (\xB,\yTop) pic {scendpointtick};
    \path (\xTauPlusDelta,\yTop) pic {screftick};
    
    \node[sctlab, above=\labelAboveOffset] at (\xA,\yTop) {$a$};
    \node[sctranslab, below=\labelBelowOffset] at (\xA,\yTop) {$T[a]$};
    
    \node[sctranslab, below=\labelBelowOffset] at (\xTbmd,\yTop) {$T[b]-1$};
    \node[sctranslab, below=\labelBelowOffset] at (\xTau,\yTop) {$\tau$};
    
    \node[sctlab, above=\labelAboveOffset] at (\xB,\yTop) {$b$};
    \node[sctranslab, below=\labelBelowOffset] at (\xB,\yTop) {$T[b]$};
    
    \node[sctranslab, below=\labelBelowOffset] at (\xTauPlusDelta,\yTop) {$\tau+1$};
    \node[sctranslab, below=\labelBelowOffset] at (\xMuTop,\yTop) {$\mu$};
    \ifdrawtransformarrows
    \draw[scconnector] (\xB,\yTop) -- (\xB,\yBot);
    \draw[scconnector] (\xMuTop,\yTop) -- (\xTauPlusDelta,\yBot);
    \fi
    
    \node[scstepdot] at (\xA,\yTop) {};
    \node[scstepdot] at (\xOne,\yTop) {};
    \node[scstepdot] at (\xTwo,\yTop) {};
    \node[scstepdot] at (\xThree,\yTop) {};
    \node[scstepdot] at (\xFour,\yTop) {};
    \node[scstepdot] at (\xTau,\yTop) {};
    \node[scbadstep] at (\xFive,\yTop) {};
    \node[scstepdot] at (\xB,\yTop) {};
    \node[scstepdot] at (\xSixTopArrive,\yTop) {};
    \node[scbadstep] at (\xMuTop,\yTop) {};
    
    \path (\xOne,\yTop) pic {scgoodmsgarc={\xTwo-\xOne}};
    
    \path (\xThree,\yTop) pic {scgoodmsgarc={\xFour-\xThree}};
    
    \path (\xTau,\yTop) pic {scgoodmsgarc={\xSixTopArrive-\xTau}};
    
    \path (\xFive,\yTop) pic {scbadmsgarc={\xMuTop-\xFive}};
    
    \node[sctranslab, anchor=east, above=\labelAboveOffset] at (\xLeft-\rowLabelXShift,\yBot) {$S$};
    \node[sctlab, anchor=east, below=\labelBelowOffset] at (\xLeft-\rowLabelXShift,\yBot) {$T'$};
    
    \fill[scwindow] (\xTbmd,\yBot-0.18) rectangle (\xB,\yBot+0.18);
    \draw[sctimeline] (\xLeft,\yBot) -- (\xB,\yBot);
    \draw[sctimeline] (\xB,\yBot) -- (\xTauPlusDelta,\yBot);
    \draw[sctimeline, dashed] (\xTauPlusDelta,\yBot) -- (\xMuTop,\yBot);
    \draw[sctimelineend] (\xMuTop,\yBot) -- (\xRight,\yBot);
    \path (\xA,\yBot) pic {scopbelow={\xB-\xA/\opBotY-\yBot/O}};
    
    \path (\xA,\yBot) pic {scendpointtick};
    \path (\xTbmd,\yBot) pic {screftick};
    \path (\xB,\yBot) pic {scendpointtick};
    \path (\xTauPlusDelta,\yBot) pic {screftick};
    \path (\xMuTop,\yBot) pic {screftick};
    
    \node[sctranslab, above=\labelAboveOffset] at (\xA,\yBot) {$a$};
    \node[sctlab, below=\labelBelowOffset] at (\xA,\yBot) {$T[a]$};
    
    \node[sctlab, below=\labelBelowOffset] at (\xTbmd,\yBot) {$T[b]-1$};
    \node[sctlab, below=\labelBelowOffset] at (\xTau,\yBot) {$\tau$};
    
    \node[sctranslab, above=\labelAboveOffset] at (\xB,\yBot) {$b$};
    \node[sctlab, below=\labelBelowOffset] at (\xB,\yBot) {$T[b]$};
    
    \node[sctlab, below=\labelBelowOffset] at (\xTauPlusDelta,\yBot) {$\tau+1$};
    \node[sctlab, below=\labelBelowOffset] at (\xMuTop,\yBot) {$\mu$};
    
    \node[scstepdot] at (\xA,\yBot) {};
    \node[scstepdot] at (\xOne,\yBot) {};
    \node[scstepdot] at (\xTwo,\yBot) {};
    \node[scstepdot] at (\xThree,\yBot) {};
    \node[scstepdot] at (\xFour,\yBot) {};
    \node[scstepdot] at (\xTau,\yBot) {};
    \node[scstepdot] at (\xFive,\yBot) {};
    \node[scstepdot] at (\xB,\yBot) {};
    \node[scstepdot] at (\xSixBotArrive,\yBot) {};
    \node[scstepdot] at (\xTauPlusDelta,\yBot) {};
    
    \path (\xOne,\yBot) pic {scgoodmsgarc={\xTwo-\xOne}};
    
    \path (\xThree,\yBot) pic {scgoodmsgarc={\xFour-\xThree}};
    
    \path (\xTau,\yBot) pic {scgoodmsgarc={\xSixBotArrive-\xTau}};
    
    \path (\xFive,\yBot) pic {scgoodmsgarc={\xTauPlusDelta-\xFive}};
    
    \end{tikzpicture}
    \caption{
    \textbf{Proof of \Cref{thm:ultimately-timely}:} retiming the run~$R=(S,T)$ (top) into the run~$R'=(S,T')$ (bottom).
    Both runs share the same schedule~$S$, which contains an invocation at step~$S[a]$ and a corresponding response at step~$S[b]$.
    Each horizontal line is a ``unified'' timeline depicting steps taken by all processes.
    Dots are steps, arcs are messages (\textcolor{red}{red}: latency exceeding~$1$; \textcolor{black}{black}: at most~$1$), and dashed portions of a timeline indicate no step is taken.
    Observe that in~$R$, messages are ultimately timely, but not timely, during the operation (denoted~$O$).
    The \colorbox{windowshade}{\smash{shaded band}\vphantom{shaded band}} is the \emph{ultimate interval} $(T[b]-1, T[b]]$.
    Time~$\tau$ is the earliest sending time of a message in~$\ultMsgs$; time~$\mu$ is the latest receiving time of such a message.
    $T'$ is constructed from~$T$ by \textcolor{blue}{compressing} the time interval $(T[b], \mu]$ into $(T[b], \tau+1]$, pulling late message receptions to within~$1$ time unit of the send times.
    The operation~$O$ has the same latency in~$R'$ as in~$R$, but messages are timely during~$O$ in~$R'$.
    }\label{fig:ultimately-timely-proof}
\end{figure}

\begin{proof}
  \proofqed{thm:ultimately-timely}
  Write $R=(S,T)$.
  If messages are timely during~$O$ in~$R$, the claim follows directly from the definition of worst-case message delay, so assume otherwise.
  Since ultimate timeliness coincides with timeliness for an incomplete operation (\Crefpart{def:ultimately-timely}{case:ultimately-timely-incomplete}), $O$~must be complete.
  Let $O$~be invoked in step~$S[a]$ and return in step~$S[b]$.
  We will construct a run $R'=(S,T')$ such that $\latency[R']{O}=\latency[R]{O}$ and messages are timely during~$O$ in~$R'$; see \Cref{fig:ultimately-timely-proof}.

  Let $\allMsgs$ denote the set of messages sent during~$O$ in~$R$ to a process that takes infinitely many steps.
  Since $S$~is a schedule, every message in~$\allMsgs$ is eventually received.
  Let $\ultMsgs\subseteq\allMsgs$ consist of the messages sent during the \emph{ultimate interval} $(T[b]-1,T[b]]$.
  Since $T$~is non-decreasing and unbounded, only finitely many steps occur during this interval, hence $\ultMsgs$ is finite.
  Because messages are ultimately timely but not timely during~$O$ in~$R$, some message in~$\ultMsgs$ has latency exceeding~$1$, so $\ultMsgs\neq\emptyset$.

  Let $\tau$ and~$\mu$ be, respectively, the earliest send time and latest receive time of a message in~$\ultMsgs$.
  These times are well-defined because $\ultMsgs$ is finite, nonempty, and consists of messages that are received.
  By definition, $\tau>T[b]-1$.
  Moreover, if some message $m\in\ultMsgs$ with $\latency[R]{m}>1$ is sent at time~$s$ and received at time~$r$, then
  \[
    \mu\geq r>s+1\geq\tau+1.
  \]
  Thus, $T[b]<\tau+1<\mu$.

  We ``retime''~$R$ by compressing the interval $(T[b],\mu]$ into $(T[b],\tau+1]$ and leaving other times unchanged.
  Define $\phi:\R\to\R$ by
  \[
    \phi(t) \Deq
    \begin{cases}
      t, & t\leq T[b],\\
      \Aff_{(T[b],\mu]\to(T[b],\tau+1]}(t), & T[b]<t\leq\mu,\\
      t, & t>\mu.
    \end{cases}
  \]
  Since $T[b]<\tau+1<\mu$, the map~$\phi$ is strictly increasing and unbounded.

  Define $T'[i]\deq\phi(T[i])$.
  Because $T$ is non-decreasing and unbounded, so is~$T'$.
  Moreover, since~$\phi$ is strictly increasing, $T[i]<T[j]$ iff $T'[i]<T'[j]$ for all indices~$i$ and~$j$.
  Thus, $R'=(S,T')$ is a run of~$\alg$ with the same schedule as~$R$.
  Since $T[a]\leq T[b]$, we have $T'[a]=\phi(T[a])=T[a]$ and $T'[b]=\phi(T[b])=T[b]$, so $\latency[R']{O}=\latency[R]{O}$.
  
  It remains to show that messages are timely during~$O$ in~$R'$.
  Because $\phi$ is strictly increasing and fixes~$T[a]$ and~$T[b]$, for any time~$t$ we have
  \begin{equation}\label{eq:ultimately-timely-proof-interval}
    t\in\interval[R]{O} \qquad\iff\qquad \phi(t)\in\interval[R']{O}.
  \end{equation}
  Fix any message~$m$ sent during~$O$ in~$R'$ to a process that takes infinitely many steps.
  Let $S[i]$ be its send step and set $s\deq T[i]$.
  Its send time in~$R'$ is $\phi(s)\in \interval[R']{O}$, so \eqref{eq:ultimately-timely-proof-interval} implies $s\in \interval[R]{O}$, i.e., $m$ is also sent during~$O$ in~$R$.
  Hence $m\in\allMsgs$.
  Let $S[j]$ be its receive step and set $r\deq T[j]$.
  Since $s\leq T[b]$, we have $\phi(s)=s$, and therefore
  \[
    \latency[R']{m}
    =\phi(r)-\phi(s)
    =\phi(r)-s.
  \]

  There are two cases, depending on whether $m$ was sent before or during the ultimate interval:
  \begin{enumerate}
      \item If $s\leq T[b]-1$, ultimate timeliness during~$O$ in~$R$ implies $r-s\leq1$, and hence $r\leq T[b]$.
      Thus~$\phi(r)=r$, so
      \[
        \latency[R']{m}=\phi(r)-s=r-s\leq1.
      \]
      \item If $s>T[b]-1$, then $m\in\ultMsgs$, so $s\geq\tau$ and $r\leq\mu$.
      Therefore,
      \[
        \latency[R']{m}=\phi(r)-s
        \leq\phi(\mu)-\tau
        =(\tau+1)-\tau
        =1.
      \]
  \end{enumerate}
  In both cases, $\latency[R']{m}\leq1$.
  So every message sent during~$O$ in~$R'$ to a process that takes infinitely many steps has latency at most one, meaning messages are timely during~$O$ in~$R'$.

  By the definition of worst-case message delay, $\latency[R']{O}\leq\latency{\alg}$.
  Since $\latency[R]{O}=\latency[R']{O}$, the claim follows.
\end{proof}

We can weaken the timeliness requirement once more.
Write~$t$ for the time a complete operation returns.
The idea is to additionally exempt messages sent during the time interval $(t-2, t-1]$ to processes other than the invoker~$p$.
Any run that meets this weaker requirement can be transformed into one in which messages are ultimately timely during the operation: We move the receipt of each exempted message that arrived late to just after time~$t-1$.
Since~$p$ is not the recipient of any such message, this transformation leaves~$p$'s steps---and hence the operation's return time and latency---unchanged, so \Cref{thm:ultimately-timely} bounds the latency.

\Cref{def:penultimately-timely} states this weaker requirement precisely and names it \emph{penultimate timeliness}; \Cref{thm:penultimately-timely} then shows that the message-delay bound still applies under it.

\begin{definition}[Penultimate timeliness]\label{def:penultimately-timely}
  Let $R = (S, T)$ be a run, and let $O$ be an operation in~$R$.
  We say that \emph{messages are penultimately timely during~$O$ in~$R$} if either of the following applies:
  \begin{enumerate}[label=(\arabic*), ref=(\arabic*)]
    \item\label{case:penultimately-timely-incomplete}
      $O$ is incomplete and messages are timely during~$O$ in~$R$; or
    \item\label{case:penultimately-timely-complete}
      $O$ is complete, having been invoked by a process~$p$ in step~$S[a]$ and returned in step~$S[b]$;
      writing
      \[
        I_{\mathrm{early}} \Deq [T[a], T[b]-2]
        \qquad\text{and}\qquad
        I_{\mathrm{pen}} \Deq [T[a], T[b]-1] \setminus I_{\mathrm{early}},
      \]
      both of the following conditions apply:
      \begin{enumerate}[label=(\alph*), ref=\theenumi(\alph*)]
        \item\label{cond:penultimately-timely-1}
          messages are timely during~$I_{\mathrm{early}}$ in~$R$; and
        \item\label{cond:penultimately-timely-2}
          if $p$ takes infinitely many steps in~$S$, then every message sent to~$p$ during~$I_{\mathrm{pen}}$ has latency at most one in~$R$.
      \end{enumerate}
  \end{enumerate}
\end{definition}

Note that condition~\ref{cond:penultimately-timely-1} holds vacuously when $\latency{O} < 2$, and condition~\ref{cond:penultimately-timely-2} holds vacuously when $p$~crashes or when $\latency{O} < 1$.
In particular, if $\latency{O} < 1$, then messages are penultimately timely during~$O$.

\begin{lemma}\label{thm:penultimately-timely}
  Let $\alg$ be an algorithm, let $R$ be a run of~$\alg$, and let $O$ be an operation in~$R$ invoked by a correct process.
  If messages are penultimately timely during~$O$ in~$R$, then $\latency{O}\leq \latency{\alg}$.
\end{lemma}
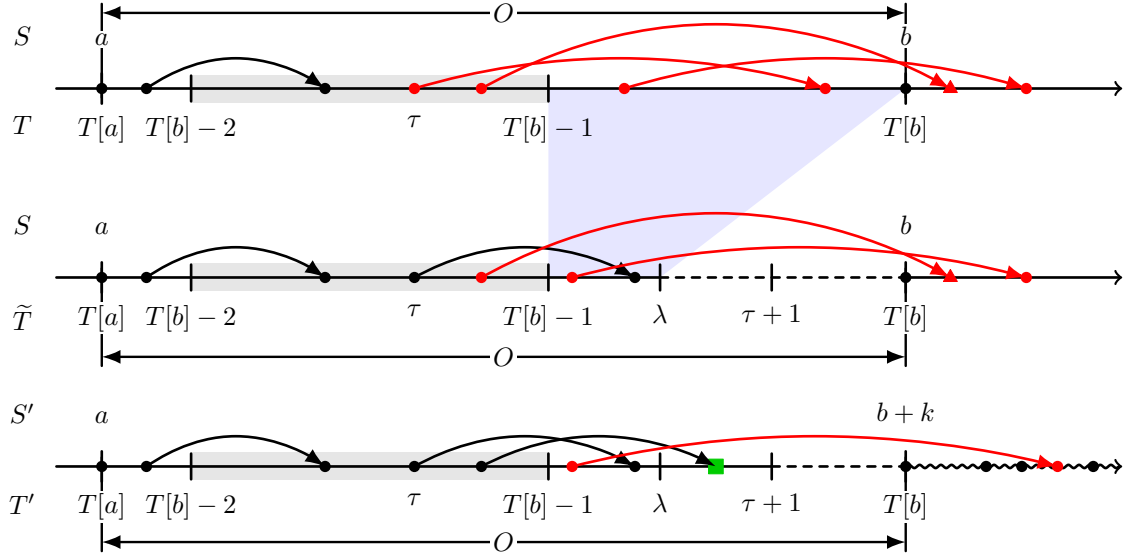
\begin{figure}
    \centering
    \begin{tikzpicture}[scbase]
        \newif\ifdrawtransformarrows
        \drawtransformarrowsfalse

        \def\yTop{5.0}
        \def\yMid{2.5}
        \def\yBot{0.0}
        \def\labelAboveOffset{15pt}
        \def\labelBelowOffset{9pt}
        \def\rowLabelXShift{0.45}
        \def\bandHalfHeight{0.18}
        \def\arcHeight{0.4}
        \def\opTopY{6.0}
        \def\opMidY{3.5}
        \def\opMidBelowY{1.45}
        \def\opBotY{-1.0}
        \def\opBarExtend{0.12}
        \ifdrawtransformarrows
        \tikzset{sctimelab/.style={sctlab}}
        \else
        \tikzset{sctimelab/.style={sctlab, fill=none}}
        \fi
        \tikzset{scmidtimelab/.style={sctlab}}

        \def\xLeft{-0.6}
        \def\xA{0.0}
        \def\xStrictGap{0.001} %
        \pgfmathsetmacro{\xGap}{(13-\xStrictGap)/22} %
        \pgfmathsetmacro{\xDeltaWidth}{8*\xGap} %
        \pgfmathsetmacro{\xB}{\xA + 18*\xGap}
        \pgfmathsetmacro{\xTbTwoDelta}{\xB-2*\xDeltaWidth}    %
        \pgfmathsetmacro{\xEarlySend}{\xA + \xGap}
        \pgfmathsetmacro{\xEarlyRecv}{\xA + 5*\xGap}
        \pgfmathsetmacro{\xTau}{\xA + 7*\xGap}
        \pgfmathsetmacro{\xPenSend}{\xA + 8.5*\xGap}
        \pgfmathsetmacro{\xTbOneDelta}{\xB-\xDeltaWidth}      %
        \pgfmathsetmacro{\xLambda}{0.5*(\xTau+\xB)}
        \pgfmathsetmacro{\xTauPlusDelta}{\xTau+\xDeltaWidth}
        \pgfmathsetmacro{\xTauRecvTop}{\xTauPlusDelta + 1.2*\xGap}
        \pgfmathsetmacro{\xTauRecvMid}{
            \xTbOneDelta + (\xTauRecvTop-\xTbOneDelta)*((\xLambda-\xTbOneDelta)/(\xB-\xTbOneDelta))
        }
        \pgfmathsetmacro{\xUltTop}{\xA + 11.7*\xGap}
        \pgfmathsetmacro{\xUltMid}{
            \xTbOneDelta + (\xUltTop-\xTbOneDelta)*((\xLambda-\xTbOneDelta)/(\xB-\xTbOneDelta))
        }
        \pgfmathsetmacro{\xPenRecvBot}{0.5*(\xLambda+\xTauPlusDelta)}
        \pgfmathsetmacro{\xPenLateRecv}{\xA + 19*\xGap}
        \pgfmathsetmacro{\xUltLateRecv}{\xA + 20.7*\xGap}
        \pgfmathsetmacro{\xUltLateRecvBot}{\xA + 21.4*\xGap}
        \pgfmathsetmacro{\xBotSuffixOne}{\xA + 19.8*\xGap}
        \pgfmathsetmacro{\xBotSuffixTwo}{\xA + 20.6*\xGap}
        \pgfmathsetmacro{\xBotSuffixThree}{\xA + 22.2*\xGap}
        \def\xRight{13.5}

        \assertlt{\xA}{\xTbTwoDelta}{xA < xTbTwoDelta}
        \assertlt{\xA}{\xEarlySend}{xA < xEarlySend}
        \assertlt{\xEarlySend}{\xTbTwoDelta}{xEarlySend < xTbTwoDelta}
        \assertlt{\xTbTwoDelta}{\xEarlyRecv}{xTbTwoDelta < xEarlyRecv}
        \assertlt{\xTbTwoDelta}{\xTau}{xTbTwoDelta < xTau}
        \assertlt{\xTau}{\xPenSend}{xTau < xPenSend}
        \assertlt{\xTbOneDelta}{\xLambda}{xTbOneDelta < xLambda}
        \assertlt{\xTbOneDelta}{\xUltTop}{xTbOneDelta < xUltTop}
        \assertlt{\xUltTop}{\xTauRecvTop}{xUltTop < xTauRecvTop}
        \assertlt{\xLambda}{\xTauPlusDelta}{xLambda < xTauPlusDelta}
        \assertlt{\xPenSend}{\xTbOneDelta}{xPenSend < xTbOneDelta}
        \assertlt{\xTauRecvTop}{\xB}{xTauRecvTop < xB}
        \assertlt{\xB}{\xPenLateRecv}{xB < xPenLateRecv}
        \assertlt{\xPenLateRecv}{\xUltLateRecv}{xPenLateRecv < xUltLateRecv}
        \assertlt{\xUltLateRecv}{\xRight}{xUltLateRecv < xRight}

        \assertlt{\xTauRecvMid}{\xLambda}{xTauRecvMid < xLambda}
        \assertlt{\xTbOneDelta}{\xUltMid}{xTbOneDelta < xUltMid}
        \assertlt{\xUltMid}{\xLambda}{xUltMid < xLambda}
        \assertlt{\xLambda}{\xB}{xLambda < xB}
        \assertlt{\xB}{\xPenLateRecv}{xB < xPenLateRecv}
        \assertlt{\xPenLateRecv}{\xUltLateRecv}{xPenLateRecv < xUltLateRecv}
        \assertlt{\xPenSend}{\xTauRecvMid}{xPenSend < xTauRecvMid}

        \assertlt{\xLambda}{\xPenRecvBot}{xLambda < xPenRecvBot}
        \assertlt{\xPenRecvBot}{\xTauPlusDelta}{xPenRecvBot < xTauPlusDelta}
        \assertlt{\xTauPlusDelta}{\xB}{xTauPlusDelta < xB}
        \assertlt{\xB}{\xBotSuffixOne}{xB < xBotSuffixOne}
        \assertlt{\xBotSuffixOne}{\xBotSuffixTwo}{xBotSuffixOne < xBotSuffixTwo}
        \assertlt{\xBotSuffixTwo}{\xUltLateRecvBot}{xBotSuffixTwo < xUltLateRecvBot}
        \assertlt{\xUltLateRecvBot}{\xBotSuffixThree}{xUltLateRecvBot < xBotSuffixThree}
        \assertlt{\xBotSuffixTwo}{\xBotSuffixThree}{xBotSuffixTwo < xBotSuffixThree}
        \assertlt{\xBotSuffixThree}{\xRight}{xBotSuffixThree < xRight}
        \assertlt{\xUltLateRecvBot}{\xRight}{xUltLateRecvBot < xRight}
        \assertlt{\xB}{\xRight}{xB < xRight}

        \assertle{\xEarlyRecv}{\xEarlySend + \xDeltaWidth}{xEarlyRecv <= xEarlySend + xDeltaWidth}
        \assertlt{\xTau + \xDeltaWidth}{\xTauRecvTop}{tau top latency > 1}
        \assertlt{\xPenSend + \xDeltaWidth}{\xPenLateRecv}{penultimate top latency > 1}
        \assertlt{\xUltTop + \xDeltaWidth}{\xUltLateRecv}{ultimate top latency > 1}

        \assertlt{\xTauRecvMid}{\xTau + \xDeltaWidth}{xTauRecvMid < xTau + xDeltaWidth}
        \assertlt{\xPenSend + \xDeltaWidth}{\xPenLateRecv}{penultimate middle latency > 1}
        \assertlt{\xUltMid + \xDeltaWidth}{\xUltLateRecv}{ultimate middle latency > 1}

        \assertle{\xTauRecvMid}{\xTau + \xDeltaWidth}{xTauRecvMid <= xTau + xDeltaWidth}
        \assertle{\xPenRecvBot}{\xPenSend + \xDeltaWidth}{xPenRecvBot <= xPenSend + xDeltaWidth}
        \assertlt{\xUltMid + \xDeltaWidth}{\xUltLateRecvBot}{ultimate bottom latency > 1}

        \ifdrawtransformarrows\else
        \fill[transformshade] (\xTbOneDelta,\yTop) -- (\xB,\yTop) -- (\xLambda,\yMid) -- (\xTbOneDelta,\yMid) -- cycle;
        \fi

        \node[sctlab, anchor=east, above=\labelAboveOffset] at (\xLeft-\rowLabelXShift,\yTop) {$S$};
        \node[sctlab, anchor=east, below=\labelBelowOffset] at (\xLeft-\rowLabelXShift,\yTop) {$T$};

        \fill[scwindow] (\xTbTwoDelta,\yTop-\bandHalfHeight) rectangle (\xTbOneDelta,\yTop+\bandHalfHeight);
        \draw[sctimelineend] (\xLeft,\yTop) -- (\xRight,\yTop);
        \path (\xA,\yTop) pic {scopabove={\xB-\xA/\opTopY-\yTop/O}};

        \ifdrawtransformarrows
        \draw[scconnector] (\xTbOneDelta,\yTop) -- (\xTbOneDelta,\yMid);
        \draw[scconnector] (\xB,\yTop) -- (\xLambda,\yMid);
        \fi

        \path (\xA,\yTop) pic {scendpointtick};
        \path (\xTbTwoDelta,\yTop) pic {screftick};
        \path (\xTbOneDelta,\yTop) pic {screftick};
        \path (\xB,\yTop) pic {scendpointtick};

        \node[sctlab, above=\labelAboveOffset] at (\xA,\yTop) {$a$};
        \node[sctimelab, below=\labelBelowOffset] at (\xA,\yTop) {$T[a]$};
        \node[sctimelab, below=\labelBelowOffset] at (\xTbTwoDelta,\yTop) {$T[b]-2$};
        \node[sctimelab, below=\labelBelowOffset] at (\xTau,\yTop) {$\tau$};
        \node[sctimelab, below=\labelBelowOffset] at (\xTbOneDelta,\yTop) {$T[b]-1$};
        \node[sctlab, above=\labelAboveOffset] at (\xB,\yTop) {$b$};
        \node[sctimelab, below=\labelBelowOffset] at (\xB,\yTop) {$T[b]$};

        \node[scstepdot] at (\xA,\yTop) {};
        \node[scstepdot] at (\xEarlySend,\yTop) {};
        \node[scstepdot] at (\xEarlyRecv,\yTop) {};
        \node[scbadstep] at (\xTau,\yTop) {};
        \node[scbadstep] at (\xPenSend,\yTop) {};
        \node[scbadstep] at (\xTauRecvTop,\yTop) {};
        \node[scbadstep] at (\xUltTop,\yTop) {};
        \node[scbadspecialstep] at (\xPenLateRecv,\yTop) {};
        \node[scbadstep] at (\xUltLateRecv,\yTop) {};
        \node[scstepdot] at (\xB,\yTop) {};

        \path (\xEarlySend,\yTop) pic {scgoodmsgarc={\xEarlyRecv-\xEarlySend}};
        \path (\xTau,\yTop) pic {scbadmsgarc={\xTauRecvTop-\xTau}};
        \path (\xPenSend,\yTop) pic[scmsgarcraise=0.45] {scbadmsgarc={\xPenLateRecv-\xPenSend}};
        \path (\xUltTop,\yTop) pic {scbadmsgarc={\xUltLateRecv-\xUltTop}};

        \node[sctlab, anchor=east, above=\labelAboveOffset] at (\xLeft-\rowLabelXShift,\yMid) {$S$};
        \node[sctlab, anchor=east, below=\labelBelowOffset] at (\xLeft-\rowLabelXShift,\yMid) {$\widetilde T$};

        \fill[scwindow] (\xTbTwoDelta,\yMid-\bandHalfHeight) rectangle (\xTbOneDelta,\yMid+\bandHalfHeight);
        \draw[sctimeline] (\xLeft,\yMid) -- (\xTauRecvMid,\yMid);
        \draw[sctimeline] (\xTauRecvMid,\yMid) -- (\xLambda,\yMid);
        \draw[sctimeline, dashed] (\xLambda,\yMid) -- (\xB,\yMid);
        \draw[sctimelineend] (\xB,\yMid) -- (\xRight,\yMid);
        \path (\xA,\yMid) pic {scopbelow={\xB-\xA/\opMidBelowY-\yMid/O}};

        \path (\xA,\yMid) pic {scendpointtick};
        \path (\xTbTwoDelta,\yMid) pic {screftick};
        \path (\xTbOneDelta,\yMid) pic {screftick};
        \path (\xLambda,\yMid) pic {screftick};
        \path (\xTauPlusDelta,\yMid) pic {screftick};
        \path (\xB,\yMid) pic {scendpointtick};

        \node[sctlab, above=\labelAboveOffset] at (\xA,\yMid) {$a$};
        \node[scmidtimelab, below=\labelBelowOffset] at (\xA,\yMid) {$T[a]$};
        \node[scmidtimelab, below=\labelBelowOffset] at (\xTbTwoDelta,\yMid) {$T[b]-2$};
        \node[scmidtimelab, below=\labelBelowOffset] at (\xTau,\yMid) {$\tau$};
        \node[scmidtimelab, below=\labelBelowOffset] at (\xTbOneDelta,\yMid) {$T[b]-1$};
        \node[scmidtimelab, below=\labelBelowOffset] at (\xLambda,\yMid) {$\lambda$};
        \node[scmidtimelab, below=\labelBelowOffset] at (\xTauPlusDelta,\yMid) {$\tau+1$};
        \node[sctlab, above=\labelAboveOffset] at (\xB,\yMid) {$b$};
        \node[scmidtimelab, below=\labelBelowOffset] at (\xB,\yMid) {$T[b]$};

        \node[scstepdot] at (\xA,\yMid) {};
        \node[scstepdot] at (\xEarlySend,\yMid) {};
        \node[scstepdot] at (\xEarlyRecv,\yMid) {};
        \node[scstepdot] at (\xTau,\yMid) {};
        \node[scbadstep] at (\xPenSend,\yMid) {};
        \node[scstepdot] at (\xTauRecvMid,\yMid) {};
        \node[scbadstep] at (\xUltMid,\yMid) {};
        \node[scbadspecialstep] at (\xPenLateRecv,\yMid) {};
        \node[scbadstep] at (\xUltLateRecv,\yMid) {};
        \node[scstepdot] at (\xB,\yMid) {};

        \path (\xEarlySend,\yMid) pic {scgoodmsgarc={\xEarlyRecv-\xEarlySend}};
        \path (\xTau,\yMid) pic {scgoodmsgarc={\xTauRecvMid-\xTau}};
        \path (\xPenSend,\yMid) pic[scmsgarcraise=0.45] {scbadmsgarc={\xPenLateRecv-\xPenSend}};
        \path (\xUltMid,\yMid) pic {scbadmsgarc={\xUltLateRecv-\xUltMid}};

        \node[sctlab, anchor=east, above=\labelAboveOffset] at (\xLeft-\rowLabelXShift,\yBot) {$S'$};
        \node[sctlab, anchor=east, below=\labelBelowOffset] at (\xLeft-\rowLabelXShift,\yBot) {$T'$};

        \fill[scwindow] (\xTbTwoDelta,\yBot-\bandHalfHeight) rectangle (\xTbOneDelta,\yBot+\bandHalfHeight);
        \draw[sctimeline] (\xLeft,\yBot) -- (\xLambda,\yBot);
        \draw[sctimeline] (\xLambda,\yBot) -- (\xTauPlusDelta,\yBot);
        \draw[sctimeline, dashed] (\xTauPlusDelta,\yBot) -- (\xB,\yBot);
        \draw[sctimeline, decorate, decoration={snake, amplitude=0.6pt, segment length=4pt}, -{Classical TikZ Rightarrow}] (\xB,\yBot) -- (\xRight,\yBot);
        \path (\xA,\yBot) pic {scopbelow={\xB-\xA/\opBotY-\yBot/O}};

        \ifdrawtransformarrows
        \draw[scspecialconnector] (\xPenLateRecv,\yMid) -- (\xPenRecvBot,\yBot);
        \fi

        \path (\xA,\yBot) pic {scendpointtick};
        \path (\xTbTwoDelta,\yBot) pic {screftick};
        \path (\xTbOneDelta,\yBot) pic {screftick};
        \path (\xLambda,\yBot) pic {screftick};
        \path (\xTauPlusDelta,\yBot) pic {screftick};
        \path (\xB,\yBot) pic {scendpointtick};

        \node[sctlab, above=\labelAboveOffset] at (\xA,\yBot) {$a$};
        \node[sctlab, below=\labelBelowOffset] at (\xA,\yBot) {$T[a]$};
        \node[sctlab, below=\labelBelowOffset] at (\xTbTwoDelta,\yBot) {$T[b]-2$};
        \node[sctlab, below=\labelBelowOffset] at (\xTau,\yBot) {$\tau$};
        \node[sctlab, below=\labelBelowOffset] at (\xTbOneDelta,\yBot) {$T[b]-1$};
        \node[sctlab, below=\labelBelowOffset] at (\xLambda,\yBot) {$\lambda$};
        \node[sctlab, below=\labelBelowOffset] at (\xTauPlusDelta,\yBot) {$\tau+1$};
        \node[sctlab, above=\labelAboveOffset] at (\xB,\yBot) {$b+k$};
        \node[sctlab, below=\labelBelowOffset] at (\xB,\yBot) {$T[b]$};

        \node[scstepdot] at (\xA,\yBot) {};
        \node[scstepdot] at (\xEarlySend,\yBot) {};
        \node[scstepdot] at (\xEarlyRecv,\yBot) {};
        \node[scstepdot] at (\xTau,\yBot) {};
        \node[scstepdot] at (\xPenSend,\yBot) {};
        \node[scstepdot] at (\xTauRecvMid,\yBot) {};
        \node[scbadstep] at (\xUltMid,\yBot) {};
        \node[scspecialstep] at (\xPenRecvBot,\yBot) {};
        \node[scstepdot] at (\xB,\yBot) {};
        \node[scbadstep] at (\xUltLateRecvBot,\yBot) {};
        \node[scstepdot] at (\xBotSuffixOne,\yBot) {};
        \node[scstepdot] at (\xBotSuffixTwo,\yBot) {};
        \node[scstepdot] at (\xBotSuffixThree,\yBot) {};

        \path (\xEarlySend,\yBot) pic {scgoodmsgarc={\xEarlyRecv-\xEarlySend}};
        \path (\xTau,\yBot) pic {scgoodmsgarc={\xTauRecvMid-\xTau}};
        \path (\xPenSend,\yBot) pic {scgoodmsgarc={\xPenRecvBot-\xPenSend}};
        \path (\xUltMid,\yBot) pic {scbadmsgarc={\xUltLateRecvBot-\xUltMid}};
    \end{tikzpicture}
    \caption{
    \textbf{Proof of \Cref{thm:penultimately-timely}:} transforming the run~$R=(S,T)$ (top) into the run~$R'=(S',T')$ (bottom), by way of the intermediate run~$\widetilde R=(S,\widetilde T)$ (middle).
    Each horizontal line is a ``unified'' timeline depicting steps taken by all processes.
    Dots are steps, arcs are messages (\textcolor{red}{red}: latency exceeding~$1$; \textcolor{black}{black}: at most~$1$), and dashed portions indicate no step is taken.
    \textbf{Top:} A correct process~$p$ invokes the operation~$O$ in step~$S[a]$ and returns in step~$S[b]$.
    Messages are penultimately timely, but not ultimately timely, during~$O$ in~$R$.
    $\tau$ is the earliest sending time of a message in~$\penMsgs$, and $\lambda\in (T[b]-1,\tau+1)$.
    Penultimate timeliness forbids a red arc from the \colorbox{windowshade}{\smash{shaded band}\vphantom{shaded band}} into step~$S[b]$, which is taken by~$p$; hence the ``\textcolor{red}{\raisebox{0.2ex}{\scalebox{0.75}{$\blacktriangle$}}}'' step is taken by a process other than~$p$.
    \textbf{Middle:} $\widetilde R$ keeps the schedule~$S$ and \textcolor{blue}{compresses} the time interval $(T[b]-1,T[b]]$ into $(T[b]-1,\lambda]$ for the steps before the return, leaving the return step~$S[b]$ at time~$T[b]$.
    $\latency[\widetilde R]{O}=\latency[R]{O}$.
    \textbf{Bottom:} $R'$ keeps the steps in~$S[0..(b-1)]$ at their times in~$\widetilde R$, inserts $k$~receive steps ``\textcolor{darkgreen}{\raisebox{0.15ex}{\rule{0.8ex}{0.8ex}}}'' in the time interval $(\lambda,\tau+1)$ for the messages in~$\latePenMsgs$, keeps the return step at time~$T[b]$, and continues with a new infinite suffix (the wavy portion).
    In~$R'$, the operation~$O$ is invoked in step~$S'[a]=S[a]$ and returns in step~$S'[b+k]=S[b]$.
    Messages are ultimately timely during~$O$ in~$R'$.
    Thus \Cref{thm:ultimately-timely} gives $\latency[R']{O}\leq\latency{\alg}$; since the construction preserves the latency of~$O$, we have $\latency[R]{O}=\latency[\widetilde R]{O}=\latency[R']{O}\leq\latency{\alg}$, as desired.
    }\label{fig:penultimately-timely-proof}
\end{figure}

\begin{proof}
  \proofqed{thm:penultimately-timely}
  Write $R=(S,T)$.
  If messages are timely during~$O$ in~$R$, the claim follows directly from the definition of worst-case message delay, so assume otherwise.
  Since penultimate timeliness coincides with timeliness for an incomplete operation (\Crefpart{def:penultimately-timely}{case:penultimately-timely-incomplete}), $O$~must be complete.
  Let $O$~be invoked by a correct process~$p$ in step~$S[a]$ and return in step~$S[b]$.

  If messages are ultimately timely during~$O$ in~$R$, the claim follows from \Cref{thm:ultimately-timely}, so assume messages are \emph{not} ultimately timely during~$O$ in~$R$.
  This means $\latency{O}\geq1$, and hence $a<b$.

  \paragraph{Proof plan.}
  The remainder of the proof is devoted to constructing a run~$R'$ of~$\alg$ in which $O$~is invoked by a correct process, messages are ultimately timely during~$O$, and $\latency[R']{O}=\latency[R]{O}$; see \Cref{fig:penultimately-timely-proof}.
  Given such a run, \Cref{thm:ultimately-timely} would give $\latency[R]{O}=\latency[R']{O}\leq\latency{\alg}$, completing the proof.

  Since messages are already timely during~$I_{\mathrm{early}}$ in~$R$, only messages sent during~$I_{\mathrm{pen}}$ to processes that take infinitely many steps can prevent ultimate timeliness.
  Let $\penMsgs$ denote the set of such messages; we must ensure that each message in~$\penMsgs$ has latency at most one in~$R'$.

  A message in~$\penMsgs$ received in~$S[b]$ is received by~$p$.
  Because $p$~is correct and hence takes infinitely many steps, condition~\ref{cond:penultimately-timely-2} already bounds the message's latency by one; we will leave such a message alone.
  We will deal with the remaining messages in~$\penMsgs$ in two stages of construction:
  \begin{itemize}
    \item First, we will retime~$R$ into an intermediate run~$\widetilde R$ in which every step preceding~$S[b]$ occurs early enough to make those messages received before~$S[b]$ timely.
      This retiming also clears a stretch of time that precedes~$S[b]$ and yet lies less than one time unit after every message in~$\penMsgs$ is sent.
    \item Then, we will insert into that stretch receive steps for the messages received after~$S[b]$.
    Such steps may be placed before~$S[b]$ because they are taken by processes other than~$p$ and so commute with~$S[b]$.
    This yields the desired run~$R'$.
  \end{itemize}

  We now give the full proof following the plan.

  \paragraph{Setup.}
  Let $I_{\mathrm{early}}$ and~$I_{\mathrm{pen}}$ be as in \Crefpart{def:penultimately-timely}{case:penultimately-timely-complete} for~$O$; recall that they partition $[T[a],T[b]-1]$.
  Let $\allMsgs$ denote the (finite) set of messages in~$R$ sent during $[T[a],T[b]-1]$ to a process that takes infinitely many steps, and let $\penMsgs\subseteq\allMsgs$ consist of those sent during~$I_{\mathrm{pen}}$.
  Recall that by condition~\ref{cond:penultimately-timely-1}, messages are timely during~$I_{\mathrm{early}}$ in~$R$.
  Since messages are \emph{not} ultimately timely during~$O$, $\penMsgs$ must contain a message whose latency exceeds one, hence $\penMsgs\neq\emptyset$.

  Let $\tau$~be the earliest send time of a message in~$\penMsgs$, so that $T[a]\leq\tau\leq T[b]-1$ and $\tau>T[b]-2$.
  Since $\tau+1>T[b]-1$, we may fix an arbitrary time $\lambda\in(T[b]-1,\tau+1)$, obtaining
  \begin{equation}\label{eq:penultimately-timely-lambda-bounds}
    T[b]-1<\lambda<\tau+1\leq T[b].
  \end{equation}

  \paragraph{An intermediate run.}
  We now retime~$R$ so that every step preceding the response step~$S[b]$ occurs at time at most~$\lambda$: For such steps we will compress the interval $(T[b]-1,T[b]]$ into $(T[b]-1,\lambda]$, while $S[b]$ and every later step keeps its time.

  To do so, define $\phi:(-\infty,T[b]]\to\R$ by
  \[
    \phi(t) \Deq
    \begin{cases}
      t, & t\leq T[b]-1,\\
      \Aff_{(T[b]-1,T[b]]\to(T[b]-1,\lambda]}(t), & T[b]-1<t\leq T[b].
    \end{cases}
  \]
  By~\eqref{eq:penultimately-timely-lambda-bounds}, the map~$\phi$ is strictly increasing; $\phi(t)=t$ for all $t\leq T[b]-1$; and the range of~$\phi$ is contained in $(-\infty,\lambda]$.
  Define the infinite sequence~$\widetilde T$ by
  \[
    \widetilde T[i] \Deq
    \begin{cases}
      \phi(T[i]), & i<b,\\
      T[i], & i\geq b.
    \end{cases}
  \]
  The sequence~$\widetilde T$ is non-decreasing and unbounded.
  Furthermore, it can be verified that $T[i] < T[j]$ only if $\widetilde T[i] < \widetilde T[j]$.
  Thus, $\widetilde R=(S,\widetilde T)$ is a run of~$\alg$ with the same schedule, and hence the same correct processes, as~$R$.
  Since $a<b$ and $T[a]\leq T[b]-1$, we have $\widetilde T[a]=\phi(T[a])=T[a]$; also, $\widetilde T[b]=T[b]$.
  Hence $\latency[\widetilde R]{O}=\latency[R]{O}$.

  Since $\widetilde T[a]=T[a]$ and~$\widetilde T[b]=T[b]$, and since $\widetilde T$ agrees with~$T$ at every step whose time is at most $T[b]-1$ in either run, $\allMsgs$ is also the set of messages in~$\widetilde R$ sent during $[\widetilde T[a],\widetilde T[b]-1]$ that timeliness constrains, and each has the same send time in both runs.
  Let $\latePenMsgs\subseteq\allMsgs$ consist of the (finitely many) messages received in a step~$S[j]$ with $j>b$.

  \begin{claim}\label{claim:penultimately-timely-message}
    Fix a message~$m\in\allMsgs$ with send time~$s$, and let $S[i]$ and~$S[j]$ be its send and receive steps.
    Then:
    \begin{enumerate}[label=(\alph*), ref=(\alph*)]
      \item\label{item:intermediate-send}
        $i<b$ and $T[i]=\widetilde T[i]=s$.
      \item\label{item:intermediate-early}
        If $s\leq T[b]-2$, then $T[j]\leq s+1\leq T[b]-1$ and $j<b$.
      \item\label{item:intermediate-pen}
        If $s>T[b]-2$, then $s\in I_{\mathrm{pen}}$, so $m\in\penMsgs$ and $s\geq\tau$.
    \end{enumerate}
  \end{claim}
  \begin{proof}
    \proofqed{claim:penultimately-timely-message}
    We prove each part in turn.
    \begin{enumerate}[label=(\alph*), ref=(\alph*)]
      \item The monotonicity of~$\widetilde T$, with $\widetilde T[i]=s\leq T[b]-1<\widetilde T[b]$, yields $i<b$.
        The definitions of~$\phi$ and~$\widetilde T$ then give $T[i]=\widetilde T[i]=s$.
      \item Since $s\in I_{\mathrm{early}}$, condition~\ref{cond:penultimately-timely-1} of penultimate timeliness gives $\latency[R]{m}\leq1$, and hence $T[j]\leq T[i]+1=s+1\leq T[b]-1<T[b]$.
        The monotonicity of~$T$ therefore implies $j<b$.
      \item Since $s\in[T[a],T[b]-1]$ and $s>T[b]-2$, we have $s\in I_{\mathrm{pen}}$ and hence $m\in\penMsgs$.
        By the definition of~$\tau$, it follows that $s\geq\tau$.\qedhere
    \end{enumerate}
  \end{proof}

  \begin{claim}\label{claim:penultimately-timely-intermediate}
    In~$\widetilde R$, the following hold:
    \begin{enumerate}[label=(\alph*), ref=(\alph*)]
      \item\label{item:intermediate-times}
        Every $i<b$ satisfies $\widetilde T[i]\leq\lambda$, while
        \[
          \widetilde T[b]-1<\lambda<\tau+1\leq\widetilde T[b].
        \]
      \item\label{item:intermediate-late}
        Every message $m\in\latePenMsgs$ is sent to a process other than~$p$, in a step~$S[i]$ with $i<b$ and $\widetilde T[i]\geq\tau$.
      \item\label{item:intermediate-timely}
        Every message in $\allMsgs\setminus\latePenMsgs$ is received in a step~$S[j]$ with $j\leq b$ and has latency at most one.
    \end{enumerate}
  \end{claim}
  \begin{proof}
    \proofqed{claim:penultimately-timely-intermediate}
    \begin{enumerate}[label=(\alph*), ref=(\alph*)]
      \item For every $i<b$, we have $\widetilde T[i]=\phi(T[i])\leq\lambda$ because the range of~$\phi$ is contained in $(-\infty,\lambda]$.
        Moreover, since $\widetilde T[b]=T[b]$, the displayed inequalities are a restatement of~\eqref{eq:penultimately-timely-lambda-bounds}.

      \item Fix a message~$m\in\latePenMsgs$, let $S[i]$ and~$S[j]$ be its send and receive steps, and write $s\deq\widetilde T[i]$.
        The definition of~$\latePenMsgs$ gives $j>b$.
        \Crefpart{claim:penultimately-timely-message}{item:intermediate-send} gives $i<b$.

        \Crefpart{claim:penultimately-timely-message}{item:intermediate-early} then rules out $s\leq T[b]-2$, so \Crefpart{claim:penultimately-timely-message}{item:intermediate-pen} applies and gives $s\in I_{\mathrm{pen}}$ and $s\geq\tau$.
        Finally, if $m$~were sent to~$p$, then $S[b]$~and~$S[j]$ would both be steps of~$p$, implying $T[j]>T[b]$.
        But $p$~takes infinitely many steps and $s\in I_{\mathrm{pen}}$, so condition~\ref{cond:penultimately-timely-2} of penultimate timeliness would yield $T[j]\leq s+1\leq T[b]$, a contradiction.

      \item Fix a message~$m\in\allMsgs\setminus\latePenMsgs$ with send time~$s$, and let $S[i]$ and~$S[j]$ be its send and receive steps.
        \Cref{claim:penultimately-timely-message} applies to~$m$, and the definition of~$\latePenMsgs$ gives $j\leq b$.
        We bound $\latency[\widetilde R]{m}=\widetilde T[j]-s$ in three cases.

        If $s\leq T[b]-2$, then \Crefpart{claim:penultimately-timely-message}{item:intermediate-early} gives $j<b$ and $T[j]\leq s+1\leq T[b]-1$.
        Thus $\phi(T[j])=T[j]$, and
        $\latency[\widetilde R]{m}=T[j]-s\leq1$.

        Now suppose $s>T[b]-2$.
        By \Crefpart{claim:penultimately-timely-message}{item:intermediate-pen}, $s\in I_{\mathrm{pen}}$ and $s\geq\tau$.
        If $j<b$, then $\widetilde T[j]\leq\lambda$ by part~\ref{item:intermediate-times}, so
        $\latency[\widetilde R]{m}\leq\lambda-\tau<1$.
        If instead $j=b$, then $m$~is received by~$p$, which takes infinitely many steps, so condition~\ref{cond:penultimately-timely-2} of penultimate timeliness gives $T[b]-s\leq1$.
        Hence $\latency[\widetilde R]{m}=T[b]-s\leq1$.\qedhere
    \end{enumerate}
  \end{proof}

  \paragraph{Inserting receive steps.}
  Let $\sigma_0\deq S[0..(b-1)]$, a schedule prefix of~$\alg$.
  By \Crefpart{claim:penultimately-timely-intermediate}{item:intermediate-late}, every message in~$\latePenMsgs$ is sent in~$\sigma_0$ to a process other than~$p$.
  Moreover, its receive step in~$S$ has index greater than~$b$, so it has not yet been received by the end of~$\sigma_0$.

  We may therefore extend~$\sigma_0$ by a finite sequence~$\sigma_1$ of steps, all taken by processes other than~$p$, in which every message in~$\latePenMsgs$ is received:
  A recipient that is correct in~$S$ takes a step according to~$\alg$, and a recipient that is faulty in~$S$---possible only under Byzantine failures, since these recipients take infinitely many steps---takes an arbitrary step receiving the message.
  Write $k\deq|\sigma_1|$.
  Because $\sigma_1$ contains no step of~$p$ and step~$S[b]$ is a step of~$p$,
  \[
    \sigma' \Deq \sigma_0 \circ \sigma_1 \circ \langle S[b]\rangle
  \]
  is a schedule prefix of~$\alg$ by commutativity.
  Extend~$\sigma'$ to a schedule~$S'$ of~$\alg$ such that exactly the processes that take infinitely many steps in~$S$ do so in~$S'$, and every correct process in~$S$ is correct in~$S'$.
  In particular, $p$~is correct in~$S'$.
  By construction,
  \[
    S'[0..(b-1)]=S[0..(b-1)],\qquad
    S'[b..(b+k-1)]=\sigma_1,\qquad
    S'[b+k]=S[b].
  \]

  \paragraph{Assigning times.}
  To construct~$R'$, we keep the steps in~$S[0..(b-1)]$ at their times in~$\widetilde R$, place the $k$~inserted steps in the interval $(\lambda,\tau+1)$, and keep the return step at time~$\widetilde T[b]$.
  Specifically, define
  \[
    T'[i] \Deq
    \begin{cases}
      \widetilde T[i], & i<b,\\
      \Aff_{(b-1,b+k)\to(\lambda,\tau+1)}(i), & b\leq i<b+k,\\
      \Aff_{[b+k,+\infty)\to[\widetilde T[b],+\infty)}(i), & i\geq b+k.
    \end{cases}
  \]
  By \Crefpart{claim:penultimately-timely-intermediate}{item:intermediate-times}, every step in~$S[0..(b-1)]$ has time at most~$\lambda$.
  The inserted steps have strictly increasing times in $(\lambda,\tau+1)$, the return step has time $\widetilde T[b]\geq\tau+1$, and all subsequent steps have strictly increasing times exceeding~$\widetilde T[b]$.
  Thus $T'$~is non-decreasing and unbounded.
  Moreover, all required strict inequalities are either inherited from~$\widetilde R$ or follow from these strict separations.
  Thus $R'=(S',T')$ is a run of~$\alg$.

  In~$R'$, the operation~$O$ is invoked in step~$S'[a]=S[a]$ and returns in step~$S'[b+k]=S[b]$; it is invoked by~$p$, which is correct in~$S'$.
  Since $T'[a]=\widetilde T[a]$ and $T'[b+k]=\widetilde T[b]$, we have $\latency[R']{O}=\latency[\widetilde R]{O}$.

  \begin{claim}\label{claim:penultimately-timely-final}
    Messages are ultimately timely during~$O$ in~$R'$.
  \end{claim}
  \begin{proof}
    \proofqed{claim:penultimately-timely-final}
    Fix a message~$m$ sent at a time in $[T'[a],T'[b+k]-1]=[\widetilde T[a],\widetilde T[b]-1]$ to a process that takes infinitely many steps in~$S'$, and let~$S'[i]$ be its send step.
    By the definition of~$T'$ and \Crefpart{claim:penultimately-timely-intermediate}{item:intermediate-times}, every $\ell\geq b$ satisfies $T'[\ell]>\lambda>\widetilde T[b]-1$.
    Hence $i<b$, so $S'[i]=S[i]$ and $T'[i]=\widetilde T[i]$; in particular, $m$~is also sent in~$\widetilde R$, with the same send time, and by the choice of~$S'$ its recipient is the same in both schedules and takes infinitely many steps in each.
    Thus $m\in\allMsgs$.

    \begin{itemize}
      \item If $m\notin\latePenMsgs$, then by \Crefpart{claim:penultimately-timely-intermediate}{item:intermediate-timely}, $m$~is received in a step~$S[j]$ with $j\leq b$ and $\latency[\widetilde R]{m}\leq1$.
        Its receive step in~$R'$ is~$S'[j]$ if $j<b$, and~$S'[b+k]$ if $j=b$; in either case, its receive time is the same in~$R'$ as in~$\widetilde R$.
        Therefore $\latency[R']{m}=\latency[\widetilde R]{m}\leq1$.
      \item If $m\in\latePenMsgs$, then \Crefpart{claim:penultimately-timely-intermediate}{item:intermediate-late} gives $T'[i]=\widetilde T[i]\geq\tau$.
        By construction, $m$~is received in an inserted step~$S'[j']$ with $b\leq j'<b+k$; the definition of~$T'$ therefore gives $T'[j']<\tau+1$.
        Hence $\latency[R']{m}=T'[j']-T'[i]<(\tau+1)-\tau=1$.
    \end{itemize}
    Since $m$~was arbitrary, messages are ultimately timely during~$O$ in~$R'$.
  \end{proof}

  By \Cref{thm:ultimately-timely}, $\latency[R']{O}\leq\latency{\alg}$.
  Since $\latency[R]{O}=\latency[\widetilde R]{O}=\latency[R']{O}$, the claim follows.
\end{proof}

\subsection{Proofs of the Completion Lemmas}\label{sec:completion-lemma-proofs}

We prove the crash- and Byzantine-failure completion lemmas in turn.

\subsubsection{Crash Failures}\label{sec:completion-lemma-proofs:crash}

We first prove the crash-failure completion lemma using penultimate timeliness; its statement is restated below.

\phantomsection\label{proof:two-delay-completion}
\twoDelayCompletion*
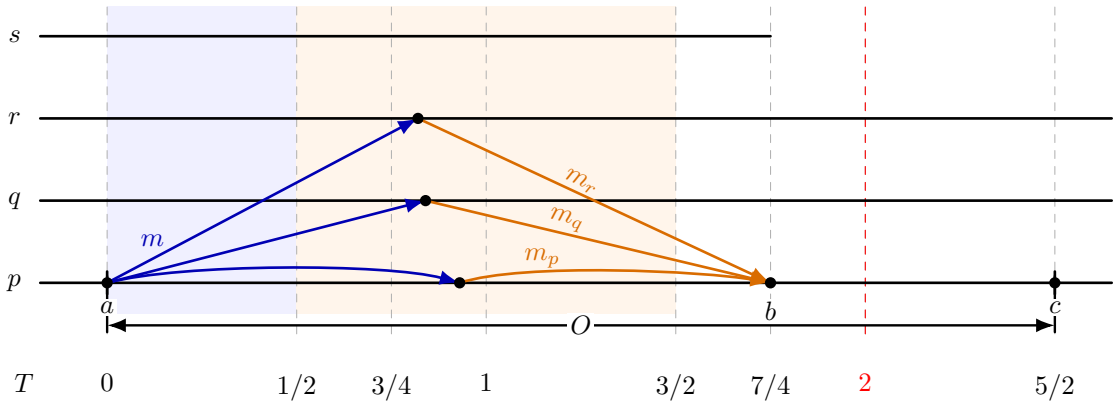
\begin{figure}
  \centering
  \begin{tikzpicture}[scbase, x=1.18cm, y=.75cm]
    \def\yP{0.0}
    \def\yQMid{1.45}
    \def\yQTop{2.9}
    \def\yCrash{4.35}
    \def\labelBelowOffset{9pt}
    \def\timeBelowOffset{33pt}
    \def\rowLabelXShift{0.18}
    \def\opBotY{-0.75}
    \def\opBarExtend{0.12}
    \def\xLeft{-.75}
    \newif\ifshowDontCareEdges
    \showDontCareEdgesfalse %

    \tikzset{
        sctimeguide/.style={draw=gray!60, dashed},
        scfirstwave/.style={scgoodmsg, draw=blue!70!black},
        screturnmsg/.style={scgoodmsg, draw=orange!85!black},
        scmsglab/.style={sctlab, fill=none},
        scothermsg/.style={->, line width=0.9pt, dotted, draw=black!45},
        scearlywindow/.style={fill=blue!6, draw=none},
        scpenwindow/.style={fill=orange!8, draw=none},
    }

    \pgfmathsetmacro{\rhoVal}{2.5}
    \pgfmathsetmacro{\xScale}{4.25}

    \pgfmathsetmacro{\xZero}{0}
    \pgfmathsetmacro{\xRhoMinusTwo}{(\rhoVal-2)*\xScale}
    \pgfmathsetmacro{\xThreeQuarters}{0.75*\xScale}
    \pgfmathsetmacro{\xOne}{1*\xScale}
    \pgfmathsetmacro{\xRhoMinusOne}{(\rhoVal-1)*\xScale}
    \pgfmathsetmacro{\xSevenQuarters}{1.75*\xScale}
    \pgfmathsetmacro{\xTwo}{2*\xScale}
    \pgfmathsetmacro{\xRho}{\rhoVal*\xScale}
    \pgfmathsetmacro{\xRight}{\xRho + 0.15*\xScale}

    \def\xA{\xZero}
    \def\xB{\xSevenQuarters}
    \pgfmathsetmacro{\xIone}{0.82*\xScale}
    \pgfmathsetmacro{\xItwo}{0.84*\xScale}
    \pgfmathsetmacro{\xIp}{0.93*\xScale}
    \pgfmathsetmacro{\xMidLate}{1.95*\xScale}
    \pgfmathsetmacro{\xQTopLate}{2.05*\xScale}
    \pgfmathsetmacro{\xPLateRecv}{2.18*\xScale}
    \def\xCrashEnd{\xB}

    \assertlt{\xZero}{\xRhoMinusTwo}{0 < rho-2}
    \assertlt{\xRhoMinusTwo}{\xThreeQuarters}{rho-2 < 3/4}
    \assertlt{\xThreeQuarters}{\xIone}{3/4 < i1}
    \assertlt{\xIone}{\xItwo}{i1 < i2}
    \assertlt{\xItwo}{\xIp}{i2 < ip}
    \assertlt{\xIp}{\xOne}{ip < 1}
    \assertlt{\xOne}{\xRhoMinusOne}{1 < rho-1}
    \assertlt{\xRhoMinusOne}{\xB}{rho-1 < b}
    \assertlt{\xB}{\xMidLate}{b < midLate}
    \assertlt{\xMidLate}{\xTwo}{midLate < 2}
    \assertlt{\xTwo}{\xQTopLate}{2 < qTopLate}
    \assertlt{\xQTopLate}{\xPLateRecv}{qTopLate < lateRecv}
    \assertlt{\xPLateRecv}{\xRho}{lateRecv < rho}

    \assertlt{\xA}{\xIone}{a -> i1}
    \assertlt{\xA}{\xIp}{a -> ip}
    \assertlt{\xA}{\xItwo}{a -> i2}
    \assertlt{\xIp}{\xB}{ip -> b}
    \assertlt{\xIone}{\xB}{i1 -> b}
    \assertlt{\xItwo}{\xB}{i2 -> b}
    \assertlt{\xIone}{\xMidLate}{i1 -> midLate}
    \assertlt{\xItwo}{\xQTopLate}{i2 -> qTopLate}
    \assertlt{\xQTopLate}{\xPLateRecv}{qTopLate -> lateRecv}

    \fill[scearlywindow] (\xZero,\yP-0.55) rectangle (\xRhoMinusTwo,\yCrash+0.52);
    \fill[scpenwindow] (\xRhoMinusTwo,\yP-0.55) rectangle (\xRhoMinusOne,\yCrash+0.52);

    \draw[sctimeguide] (\xZero,\yP-0.9) -- (\xZero,\yCrash+0.7);
    \draw[sctimeguide] (\xRhoMinusTwo,\yP-0.9) -- (\xRhoMinusTwo,\yCrash+0.7);
    \draw[sctimeguide] (\xThreeQuarters,\yP-0.9) -- (\xThreeQuarters,\yCrash+0.7);
    \draw[sctimeguide] (\xOne,\yP-0.9) -- (\xOne,\yCrash+0.7);
    \draw[sctimeguide] (\xRhoMinusOne,\yP-0.9) -- (\xRhoMinusOne,\yCrash+0.7);
    \draw[sctimeguide] (\xSevenQuarters,\yP-0.9) -- (\xSevenQuarters,\yCrash+0.7);
    \draw[sctimeguide] (\xTwo,\yP-0.9) -- (\xTwo,\yCrash+0.7);
    \draw[sctimeguide, draw=red] (\xTwo,\yP-0.9) -- (\xTwo,\yCrash+0.7);
    \draw[sctimeguide] (\xRho,\yP-0.9) -- (\xRho,\yCrash+0.7);

    \draw[sctimeline] (\xLeft,\yP) -- (\xRight,\yP);
    \draw[sctimeline] (\xLeft,\yQMid) -- (\xRight,\yQMid);
    \draw[sctimeline] (\xLeft,\yQTop) -- (\xRight,\yQTop);
    \draw[sctimeline] (\xLeft,\yCrash) -- (\xCrashEnd,\yCrash);

    \node[sctlab, anchor=east] at (\xLeft-\rowLabelXShift,\yP) {$p$};
    \node[sctlab, anchor=east] at (\xLeft-\rowLabelXShift,\yQMid) {$q$};
    \node[sctlab, anchor=east] at (\xLeft-\rowLabelXShift,\yQTop) {$r$};
    \node[sctlab, anchor=east] at (\xLeft-\rowLabelXShift,\yCrash) {$s$};
    \node[sctlab, anchor=east, below=\timeBelowOffset] at (\xLeft-\rowLabelXShift,\yP) {$T$};
    \path (\xA,\yP) pic {scopbelow={\xRho-\xA/\opBotY-\yP/O}};

    \path (\xA,\yP) pic {scendpointtick};
    \path (\xRho,\yP) pic {scendpointtick};

    \draw[scfirstwave] (\xA,\yP) -- (\xIone,\yQTop);
    \draw[scfirstwave] (\xA,\yP) -- (\xItwo,\yQMid);
    \draw[scfirstwave] (\xA,\yP)
        to[out=15,in=165,looseness=0.55]
        (\xIp,\yP);
    \node[scmsglab, text=blue!70!black] at ({0.12*\xScale},\yP+0.75) {$m$};
    \draw[screturnmsg] (\xIp,\yP)
        to[out=18,in=170,looseness=0.55]
        node[scmsglab, pos=0.30, above, text=orange!85!black] {$m_p$}
        (\xB,\yP);
    \draw[screturnmsg] (\xIone,\yQTop) --
        node[scmsglab, pos=0.45, above, sloped, text=orange!85!black] {$m_r$}
        (\xB,\yP);
    \draw[screturnmsg] (\xItwo,\yQMid) --
        node[scmsglab, pos=0.40, above, sloped, text=orange!85!black] {$m_q$}
        (\xB,\yP);

    \ifshowDontCareEdges
        \draw[scothermsg, draw=black!30] (\xA,\yP) to[out=65,in=190] (\xB,\yCrash);
        \draw[scothermsg, draw=black!30] (\xIone,\yQTop) to[out=35,in=175] (\xB,\yCrash);
        \draw[scothermsg, draw=black!30] (\xItwo,\yQMid) to[out=50,in=205] (\xB,\yCrash);
        \draw[scothermsg] (\xIone,\yQTop) -- (\xMidLate,\yQMid);
        \draw[scothermsg] (\xItwo,\yQMid) -- (\xQTopLate,\yQTop);
        \draw[scothermsg] (\xQTopLate,\yQTop) -- (\xPLateRecv,\yP);
    \fi

    \node[scstepdot] at (\xA,\yP) {};
    \node[scstepdot] at (\xIp,\yP) {};
    \node[scstepdot] at (\xB,\yP) {};
    \node[scstepdot] at (\xRho,\yP) {};

    \node[scstepdot] at (\xItwo,\yQMid) {};

    \node[scstepdot] at (\xIone,\yQTop) {};

    \ifshowDontCareEdges
        \node[scstepdot] at (\xPLateRecv,\yP) {};
        \node[scstepdot] at (\xMidLate,\yQMid) {};
        \node[scstepdot] at (\xQTopLate,\yQTop) {};
    \fi

    \node[sctlab, yshift=3pt, below=\labelBelowOffset] at (\xA,\yP) {$a$};
    \node[sctlab, yshift=3pt, below=\labelBelowOffset] at (\xB,\yP) {$b$};
    \node[sctlab, yshift=3pt, below=\labelBelowOffset] at (\xRho,\yP) {$c$};

    \node[sctlab, below=\timeBelowOffset] at (\xZero,\yP) {$0$};
    \node[sctlab, below=\timeBelowOffset] at (\xRhoMinusTwo,\yP) {$1/2$};
    \node[sctlab, below=\timeBelowOffset] at (\xThreeQuarters,\yP) {$3/4$};
    \node[sctlab, below=\timeBelowOffset] at (\xOne,\yP) {$1$};
    \node[sctlab, below=\timeBelowOffset] at (\xRhoMinusOne,\yP) {$3/2$};
    \node[sctlab, below=\timeBelowOffset] at (\xSevenQuarters,\yP) {$7/4$};
    \node[sctlab, text=red, below=\timeBelowOffset] at (\xTwo,\yP) {$2$};
    \node[sctlab, below=\timeBelowOffset] at (\xRho,\yP) {$5/2$};
  \end{tikzpicture}
  \caption{\textbf{Proof of \nameCref{thm:two-delay-completion}~\hyperref[proof:two-delay-completion]{\ref*{thm:two-delay-completion}}.}
  Let~$\alg$ be a loquacious algorithm with $\latency{\alg}\leq2$, and suppose toward a contradiction that the operation~$O$, invoked by~$p$ in step~$a$, has not returned by step~$b$.
  Extend the schedule prefix so that exactly the processes in~$Q$ are correct; then $O$ returns in some step~$c>b$, and the times~$T$ give $\latency{O}=T[c]-T[a]=5/2$.
  The figure depicts the instance $n=4$, $f=1$, and $Q=\{p,q,r\}$, in which $p$, $q$, and~$r$ are correct and~$s$ crashes; in general $|Q|=n-f$.
  The shaded bands show the two intervals from \Crefpart{def:penultimately-timely}{case:penultimately-timely-complete} of penultimate timeliness:
  (a)~\textcolor{blue!70!black}{$I_{\mathrm{early}}=[0,1/2]$}; and
  (b)~\textcolor{orange!85!black}{$I_{\mathrm{pen}}=(1/2,3/2]$}.
  The figure shows only the messages to which these conditions apply; each has latency less than~$1$.
  So messages are penultimately timely during~$O$, and \Cref{thm:penultimately-timely} gives $\latency{O}\leq\latency{\alg}\leq\textcolor{red}{2}$, contradicting $\latency{O}=5/2$.}
  \label{fig:less-than-three-delay}
\end{figure}

\begin{proof}
  \proofqed[proof:two-delay-completion]{thm:two-delay-completion}
  In a system with crash failures, let $\alg$ be a loquacious algorithm with $\latency{\alg}\leq 2$, and let $\sigma$ be a schedule prefix of~$\alg$ that satisfies the assumptions in the lemma statement.

  Suppose, for the sake of contradiction, that process~$p$ has not returned from the operation by step~$\sigma[b]$.
  Extend the prefix~$\sigma[0..b]$ to a schedule~$S$ of~$\alg$ in which every process in~$Q$ takes infinitely many further steps and receives every message sent to it, while every remaining process takes no further steps.
  Since $p$~is correct in~$S$ and $\latency{\alg}<\infty$, \Cref{thm:liveness-from-finite-delay} implies that the operation invoked in step~$S[a]=\sigma[a]$ eventually returns, in some step~$S[c]$ with $c>b$.

  Now, define the time sequence~$T$ as follows (see \Cref{fig:less-than-three-delay}):
  \[
    T[i] \Deq
    \begin{cases}
      \Aff_{(-\infty,a]\to(-\infty,0]}(i), & i\leq a, \\
      \Aff_{(a,b)\to(3/4,1)}(i), & a<i<b, \\
      \Aff_{[b,c]\to[7/4,5/2]}(i), & b\leq i\leq c, \\
      \Aff_{(c,+\infty)\to(5/2,+\infty)}(i), & i>c.
    \end{cases}
  \]
  It can be verified that $T$~is strictly increasing and unbounded, so $R=(S,T)$ is a run of~$\alg$.
  Let $O$~denote the operation in~$R$ invoked in step~$S[a]$ and returning in step~$S[c]$.
  It follows that $\latency{O} = T[c] - T[a] = 5/2$.

  \begin{claim*}
    Messages are penultimately timely during~$O$ in~$R$.
  \end{claim*}
  \begin{proof}
    \proofqed*{Claim}
    Since $O$ is complete, we must verify conditions~\ref{cond:penultimately-timely-1} and~\ref{cond:penultimately-timely-2} of \Cref{def:penultimately-timely}.
    \begin{itemize}
    \item For condition~\ref{cond:penultimately-timely-1}, consider the interval
      \[
        I_{\mathrm{early}}=[T[a],T[c]-2] = \left[0,\frac{1}{2}\right] \qquad(\text{blue interval in \Cref{fig:less-than-three-delay}}).
      \]
      The only step in this interval is~$S[a]$, so the messages sent during this interval to a process that takes infinitely many steps are precisely the copies of~$m$ sent in step~$S[a]$ to processes in~$Q$.
      By assumption~\ref{item:two-delay-completion-first-wave}, each such message is received in some step~$S[j_q]$ with $a<j_q<b$, and hence before time~$1$.
      Since these messages are sent at time~$0$, each has latency less than~$1$.
    \item For condition~\ref{cond:penultimately-timely-2}, we have
      \[
        I_{\mathrm{pen}}
        =[T[a],T[c]-1]\setminus I_{\mathrm{early}}
        =\left(\frac{1}{2},\frac{3}{2}\right]
        \qquad(\text{orange interval in \Cref{fig:less-than-three-delay}}).
      \]
      Since $p$ takes infinitely many steps, we must examine the messages sent in this interval to~$p$.
      Let~$\widehat m$ be such a message, sent in step~$S[j]$.
      Since
      \[ T[a]=0<\frac{1}{2}<T[j]\leq\frac{3}{2}<\frac{7}{4}=T[b] \]
      and~$T$ is strictly increasing, we have $a<j<b$.
      So by assumptions~\ref{item:two-delay-completion-first-wave} and~\ref{item:two-delay-completion-no-other-steps}, we have $\widehat m=m_q$ for some~$q\in Q$.
      By assumption~\ref{item:two-delay-completion-second-wave}, message~$\widehat m$ is received in step~$S[b]$.
      Furthermore, $j>a$ implies $T[j]>3/4$ by the definition of~$T$.
      Therefore,
      \[
        \latency{\widehat m} = T[b]-T[j] < \frac{7}{4} - \frac{3}{4} = 1.\qedhere
      \]
    \end{itemize}
  \end{proof}
  By \Cref{thm:penultimately-timely}, the claim implies $\latency{O}\leq\latency{\alg}\leq2$.
  This contradicts $\latency{O}=5/2$.
\end{proof}

\subsubsection{Byzantine Failures}\label{sec:completion-lemma-proofs:byz}
We next derive the Byzantine variant from the same argument; its statement is also restated below (with the differences from the crash version underlined).

\phantomsection\label{proof:two-delay-completion-byz}
\twoDelayCompletionByz*
\begin{proof}
  \proofqed[proof:two-delay-completion-byz]{thm:two-delay-completion-byz}
  The proof of \nameCref{thm:two-delay-completion}~\hyperref[proof:two-delay-completion]{\ref*{thm:two-delay-completion}} applies with only one change: We must justify its initial extension of~$\sigma[0..b]$ to a schedule of~$\alg$.
  Since every process in~$Q$ follows~$\alg$ throughout~$\sigma[0..b]$, we can extend this prefix to a schedule of~$\alg$ by letting these processes take infinitely many further steps according to~$\alg$ (and ensuring that every message sent to them is eventually received), and letting the remaining processes take no further steps.
  In the resulting schedule, the processes in~$Q$ are exactly the ones that take infinitely many steps and exactly the correct ones, and so there are $n-|Q|=f$ faulty processes.
  Because $p$ is correct in the schedule, its operation eventually returns.
  The remainder of the proof---the time assignment and the verification of penultimate timeliness---applies verbatim and yields the same contradiction.
\end{proof}

\section{Proof of the Two-Delay Lower Bound for Crash Failures}\label{sec:crash-lower-bound-proofs}
In this section, we prove the two-delay latency lower bound in systems with crash failures from~\Cref{sec:crash:lower-bounds}.
The proof relies on the following
lemma.

\begin{lemma}\label{thm:one-delay-completion}
  In a system with crash failures, let $\alg$~be a loquacious algorithm with $\latency{\alg} \leq 1$.
  Suppose that $\sigma$ is a schedule prefix of~$\alg$ in which some process~$p$:
  \begin{itemize}
    \item invokes an operation and broadcasts a message~$m$ in step~$\sigma[a]$, and
    \item receives its own message~$m$ in step~$\sigma[b]$.
  \end{itemize}
  Then the operation invoked by~$p$ in step~$\sigma[a]$ returns no later than step~$\sigma[b]$.
\end{lemma}
\begin{proof}
  \proofqed{thm:one-delay-completion}
  Let the worst-case message delay of~$\alg$ be~$k\leq 1$.
  Suppose, for the sake of contradiction, that the operation invoked by~$p$ in step~$\sigma[a]$ does not return by step~$\sigma[b]$.
  By assumption, $b>a$.

  Extend~$\sigma[0..b]$ to an (infinite) schedule~$S$ of~$\alg$ in which $p$~takes infinitely many steps.
  Then $p$~is correct in~$S$, and since $\alg$~has finite worst-case message delay, \Cref{thm:liveness-from-finite-delay} implies that the operation invoked in step~$S[a]$ eventually returns.
  Let $S[c]$ be its return step, where $c > b$, and call the complete operation~$O$.

  Assign times~$T$ to the steps in~$S$ by setting
  \[
    T[i] \Deq
    \begin{cases}
      \Aff_{(-\infty,a]\to(-\infty,0]}(i), & i\leq a, \\
      \Aff_{(a,b]\to(1/2,3/4]}(i), & a<i\leq b, \\
      \Aff_{(b,c]\to(3/4,3/2]}(i), & b<i\leq c, \\
      \Aff_{(c,+\infty)\to(3/2,+\infty)}(i), & i>c.
    \end{cases}
  \]
  This sequence is strictly increasing and unbounded.
  Let $R=(S,T)$ be the resulting run.

  Since $O$ is complete, we show that messages are penultimately timely during~$O$ in~$R$ by verifying the two conditions of \Crefpart{def:penultimately-timely}{case:penultimately-timely-complete}.
  Condition~\ref{cond:penultimately-timely-1} holds vacuously because $T[c]-2 < T[a]$.
  Since $p$~takes infinitely many steps, condition~\ref{cond:penultimately-timely-2} requires us to consider the messages sent to~$p$ during the relevant time interval $[0,1/2]$; the only step taken in this interval is~$S[a]$.
  Since~$\alg$ is loquacious, the only message sent to~$p$ in step~$S[a]$ is~$m$; by assumption, $p$ receives~$m$ in step~$S[b]$, with latency $T[b]-T[a]=3/4\leq 1$.

  Therefore, by \Cref{thm:penultimately-timely}, $\latency{O}\leq k \leq 1$, contradicting the fact that $\latency{O}=T[c]-T[a]=3/2$.
\end{proof}

We now restate and prove the two-delay latency lower bound.

\phantomsection\label{proof:crash-two-delay-lower-bound}
\crashTwoDelayLowerBound*
\begin{proof}
  \proofqed[proof:crash-two-delay-lower-bound]{thm:crash-two-delay-lower-bound}
  Suppose, toward a contradiction, that there is an implementation~$\alg$ of~\ts{} whose worst-case message delay is at most one.
  By \Cref{thm:loquacious-transformation}, we may assume without loss of generality that~$\alg$ is loquacious.
  By \Cref{thm:timestamp-characterization}, $\alg$ satisfies ordering~(\ref{cond:timestamp-ordering}).
  Since $n\geq 2$, let $p$ and $q$ be distinct processes in the system.

  Construct a schedule prefix~$\sigma_p$ of~$\alg$ consisting of two steps, both taken by process~$p$: In the first step, $p$ invokes an operation and broadcasts a message~$m_p$; in the second step, $p$~receives~$m_p$.
  By \Cref{thm:one-delay-completion}, this operation returns by the end of~$\sigma_p$---let its return value be~$x_p$.
  Construct~$\sigma_q$ symmetrically for process~$q$, with return value~$x_q$.

  Note that~$\sigma_p$ contains only steps of~$p$ and~$\sigma_q$ contains only steps of~$q$, so their concatenation $\sigma_p \circ \sigma_q$ is also a schedule prefix of~$\alg$.
  Extend $\sigma_p \circ \sigma_q$ to a schedule of~$\alg$; in this schedule, condition~(\ref{cond:timestamp-ordering}) implies $x_p<x_q$.
  Similarly, we can extend $\sigma_q \circ \sigma_p$ to a schedule of~$\alg$, yielding $x_q<x_p$, a contradiction.
\end{proof}

\end{document}